\documentclass[preprint,a4paper,eqsecnum,pra,twocolumn,10pt,showpacs,showkeys]{revtex4}
\usepackage{amsmath}
\usepackage{amsfonts}
\usepackage{amssymb}
\usepackage{graphicx}
\usepackage{hyperref}

\newtheorem{theorem}{Theorem}

\newtheorem{corollary}[theorem]{Corollary}

\newtheorem{definition}[theorem]{Definition}
\newtheorem{example}[theorem]{Example}

\newtheorem{lemma}[theorem]{Lemma}

\newtheorem{proposition}[theorem]{Proposition}
\newtheorem{remark}[theorem]{Remark}
\newtheorem{solution}[theorem]{Solution}

\newenvironment{proof}[1][Proof]{\noindent\textbf{#1.} }{\ \rule{0.5em}{0.5em}}

\begin{document}

\title{Theory of transformation for the diagonalization of quadratic
Hamiltonians}
\author{Ming-wen Xiao}
\affiliation{Department of Physics, Nanjing University, Nanjing 210093, People's Republic
of China}

\begin{abstract}
A theory of transformation is presented for the diagonalization of a
Hamiltonian that is quadratic in creation and annihilation operators or in
coordinates and momenta. It is the systemization and theorization of Dirac
and Bogoliubov-Valatin transformations, and thus provides us an operational
procedure to answer, in a direct manner, the questions as to whether a
quadratic Hamiltonian is diagonalizable, whether the diagonalization is
unique, and how the transformation can be constructed if the diagonalization
exists. The underlying idea is to consider the dynamic matrix. Each
quadratic Hamiltonian has a dynamic matrix of its own. The eigenvalue
problem of the dynamic matrix determines the diagonalizability of the
quadratic Hamiltonian completely. In brief, the theory ascribes the
diagonalization of a quadratic Hamiltonian to the eigenvalue problem of its
dynamic matrix, which is familiar to all of us. That makes it much easy to
use. Applications to various physical systems are discussed, with especial
emphasis on the quantum fields, such as Klein-Gordon field, phonon field,
etc..
\end{abstract}

\email{xmw@netra.nju.edu.cn}
\maketitle
\tableofcontents

\section{Introduction}

In the preface to the first edition of \textquotedblleft The Principles of
Quantum Mechanics\textquotedblright\ \cite{Dirac}, Dirac said:
\textquotedblleft The growth of the use of transformation theory, as applied
first to relativity and later to the quantum theory, is the essence of the
new method in theoretical physics.\textquotedblright

In this review, we shall present a theory of transformation, which can
perform diagonalization to the Hamiltonian that is quadratic in creation and
annihilation operators or in coordinates and momenta. Such transformations
appear most frequently in classical and quantum mechanics, statistical
mechanics, condensed-matter physics, nuclear physics, and quantum field
theory.

For the sake of simplicity, let us begin with the so-called
Bogoliubov-Valatin transformation.

\subsection{Bogoliubov-Valatin transformation \label{BVTS}}

In 1947, Bogoliubov \cite{Bogoliubov1} introduced a novel linear
transformation to diagonalize the quantum quadratic Hamiltonian present in
superfluidity. This method was later extended by Bogoliubov himself \cite%
{Bogoliubov2,Bogoliubov3,Bogoliubov4} and also by Valatin \cite%
{Valatin1,Valatin2} to the Fermi case in the theory of superconductivity. It
has ever since got widely used in different fields \cite{Fetter,Wagner,Nambu}%
, and known as Bogoliubov-Valatin (BV) transformation, including both the
bosonic and fermionic versions.

To show the underlying idea of the method due to Bogoliubov and Valatin, let
us consider the quadratic Hamiltonian,
\begin{equation}
H=\sum_{i,j=1}^{n}(\alpha _{ij}c_{i}^{\dag }c_{j}+\frac{1}{2}\gamma
_{ij}c_{i}^{\dag }c_{j}^{\dag }+\frac{1}{2}\gamma _{ji}^{\ast }c_{i}c_{j}),
\label{Ham1}
\end{equation}%
where $n\geq 1$ is a natural number, and $c_{i}$ and $c_{i}^{\dag }$ are,
respectively, the annihilation and creation operators for bosons or
fermions. They satisfy the standard commutation or anticommution relations,
\begin{align}
\lbrack c_{i},c_{j}^{\dag }]& =c_{i}c_{j}^{\dag }\pm c_{j}^{\dag
}c_{i}=\delta _{ij},  \label{Com1} \\
\lbrack c_{i},c_{j}]& =c_{i}c_{j}\pm c_{j}c_{i}=0,  \label{Com2} \\
\lbrack c_{i}^{\dag },c_{j}^{\dag }]& =c_{i}^{\dag }c_{j}^{\dag }\pm
c_{j}^{\dag }c_{i}^{\dag }=0,  \label{Com3}
\end{align}%
where $\delta _{ij}$ is the Kronecker delta function. The coefficients $%
\alpha _{ij}\in
\mathbb{C}
$ and $\gamma _{ij}\in
\mathbb{C}
$ have the following symmetries,
\begin{equation}
\alpha _{ij}=\alpha _{ji}^{\ast },\text{ \ }\gamma _{ij}=\mp \gamma _{ji},
\end{equation}%
where $z^{\ast }$ denotes the complex conjugate of $z$. Throughout this
review, the complex field $%
\mathbb{C}
$ will be used as the base field of the Hamiltonian $H$.

Using the form of matrix, Eq. (\ref{Ham1}) can be written as%
\begin{equation}
H=\frac{1}{2}\psi ^{\dag }M\psi \pm \frac{1}{2}\mathrm{tr}(\alpha ),
\label{Ham2}
\end{equation}%
where $\mathrm{tr}(A)$ denotes the trace of the matrix $A$. The $\psi $ is a
column vector and $\psi ^{\dag }$ its Hermitian conjugate,
\begin{equation}
\psi =\left[
\begin{array}{c}
c \\
\widetilde{c^{\dag }}%
\end{array}%
\right] ,\text{ \ }\psi ^{\dag }=\left[
\begin{array}{cc}
c^{\dag }, & \widetilde{c}%
\end{array}%
\right] ,  \label{Psi}
\end{equation}%
where $c$ and $c^{\dag }$ are the subvectors of size $n$,
\begin{equation}
c=\left[
\begin{array}{c}
c_{1} \\
c_{2} \\
\vdots \\
c_{n}%
\end{array}%
\right] ,\text{ \ }c^{\dag }=\left[
\begin{array}{cccc}
c_{1}^{\dag }, & c_{2}^{\dag }, & \cdots , & c_{n}^{\dag }%
\end{array}%
\right] .
\end{equation}%
Here $\widetilde{A}$ denotes the transpose of the matrix $A$. The
coefficient matrix $M$ has the form,%
\begin{equation}
M=\left[
\begin{array}{cc}
\alpha & \gamma \\
\gamma ^{\dag } & \mp \widetilde{\alpha }%
\end{array}%
\right] ,  \label{CMH}
\end{equation}%
where $\alpha $ and $\gamma $ are the submatrices with $\alpha _{ij}$ and $%
\gamma _{ij}$ as their their entries, respectively. Obviously,%
\begin{equation}
\alpha ^{\dag }=\alpha ,\text{ \ }\widetilde{\gamma }=\mp \gamma ,\text{ \ }%
M^{\dag }=M.  \label{AGM}
\end{equation}%
That is to say, $\alpha $ and $M$ are both Hermitian matrices whereas $%
\gamma $ is a symmetric or antisymmetric matrix, which is determined by
whether the system is bosonic or fermionic. Besides, the matrices $\alpha $
and $\gamma $ will not vanish simultaneously; otherwise, the Hamiltonian $H$
is zero trivially.

If we define a new product between the two operators $c_{i}$ (or $%
c_{i}^{\dag }$) and $c_{j}$ (or $c_{j}^{\dag}$) as%
\begin{equation}
c_{i}\cdot c_{j}=[c_{i},c_{j}],
\end{equation}
then Eqs. (\ref{Com1})--(\ref{Com3}) can be expressed compactly as%
\begin{equation}
\psi\cdot\psi^{\dag}=I_{\pm},  \label{Com4}
\end{equation}
where%
\begin{equation}
I_{\pm}=\left[
\begin{array}{cc}
I & 0 \\
0 & \pm I%
\end{array}
\right] ,
\end{equation}
with $I$ being the identity matrix of size $n$.

To diagonalize the Hamiltonian of Eq. (\ref{Ham2}), Bogoliubov and Valatin
introduced a linear transformation,
\begin{equation}
c=Ad+B\widetilde{d^{\dag }},  \label{Transf}
\end{equation}%
where $A$ and $B$ are two square matrices of size $n$, and $d$ and $d^{\dag
} $ are the vectors as follows,\
\begin{equation}
d=\left[
\begin{array}{c}
d_{1} \\
d_{2} \\
\vdots \\
d_{n}%
\end{array}%
\right] ,\ d^{\dag }=\left[
\begin{array}{cccc}
d_{1}^{\dag }, & d_{2}^{\dag }, & \cdots , & d_{n}^{\dag }%
\end{array}%
\right] .
\end{equation}%
Here $d_{i}$ and $d_{j}^{\dag }$ are the new annihilation and creation
operators respectively, they satisfy the standard commutation or
anticommution relations as in Eqs. (\ref{Com1})--(\ref{Com3}), which means,%
\begin{equation}
\varphi \cdot \varphi ^{\dag }=I_{\pm },  \label{Com5}
\end{equation}%
where
\begin{equation}
\varphi =\left[
\begin{array}{c}
d \\
\widetilde{d^{\dag }}%
\end{array}%
\right] ,\text{ \ }\varphi ^{\dag }=\left[
\begin{array}{cc}
d^{\dag }, & \widetilde{d}%
\end{array}%
\right] .  \label{GM}
\end{equation}%
From Eqs. (\ref{Psi}), (\ref{GM}) and (\ref{Transf}), it follows that
\begin{equation}
\psi =T\varphi ,  \label{Transf1}
\end{equation}%
where
\begin{equation}
T=\left[
\begin{array}{cc}
A & B \\
B^{\ast } & A^{\ast }%
\end{array}%
\right] .  \label{TM}
\end{equation}%
Here $A^{\ast }$ denotes the complex conjugate of the matrix $A$. By the
way, we note that such a form of $T$ originates from the requirement that $c$
and $c^{\dag }$ must be Hermitian conjugates of each other. For convenience,
we shall call the operator vector such as $\psi $ and $\varphi $ the field
operator.

Under the transformation of Eq. (\ref{Transf1}), the Hamiltonian of Eq. (\ref%
{Ham2}) becomes%
\begin{equation}
H=\frac{1}{2}\varphi^{\dag}T^{\dag}MT\varphi\pm\frac{1}{2}\mathrm{tr}%
(\alpha),  \label{Ham3d}
\end{equation}
where $T^{\dag}MT$ is the new coefficient matrix. Meanwhile, Eq. (\ref{Com4}%
) turns into
\begin{equation}
TI_{\pm}T^{\dag}=I_{\pm},  \label{Cond1}
\end{equation}
where Eq. (\ref{Com5}) has been used. Obviously, this is a condition for the
transformation of Eq. (\ref{Transf1}).

For the Hamiltonian $H$ to be diagonalized with respect to the new
annihilation and creation operators, it is necessary that the new
coefficient matrix $T^{\dag}MT$ is diagonal, i.e.,
\begin{equation}
T^{\dag}MT=\left[
\begin{array}{cc}
\begin{array}{cc}
\omega_{1} &  \\
& \omega_{2}%
\end{array}
& \text{{\LARGE 0}} \\
\text{{\LARGE 0}} &
\begin{array}{cc}
\ddots &  \\
& \omega_{2n}%
\end{array}%
\end{array}
\right] ,  \label{Cond2}
\end{equation}
where $\omega_{i}$ for $i=1,2,\cdots,2n$ are the diagonal entries, they are
real: $\omega_{i}\in%
\mathbb{R}
$. Equation (\ref{Cond2}) means that all the off-diagonal entries of the
matrix $T^{\dag}MT$ must vanish identically. Under this condition, we have
\begin{equation}
H=\frac{1}{2}\sum_{i=1}^{n}(\omega_{i}\mp\omega_{n+i})d_{i}^{\dag}d_{i}+%
\frac{1}{2}\sum_{i=1}^{n}\omega_{n+i}\pm\frac{1}{2}\mathrm{tr}(\alpha).
\label{Ham3d1}
\end{equation}
This is the so-called diagonalized form for the Hamiltonian $H$.

To sum up, Eqs. (\ref{Cond1}) and (\ref{Cond2}) are the two conditions that
must be fulfilled by the transformation matrix $T$. The former ensures the
statistics of the system, i.e., the system will remains bosonic or fermionic
after the transformation if it is bosonic or fermionic before the
transformation, that is a physical requirement. The latter ensures the
diagonalization of the Hamiltonian, it is just a mathematical requirement.
According to Bogoliubov and Valatin, the transformation matrix $T$ can be
determined from Eqs. (\ref{Cond1}) and (\ref{Cond2}). After the
determination of $T$, the diagonal entries $\omega_{i}$ for $i=1,2,\cdots,2n$
will be obtained, which accomplishes the diagonalization procedure. That is
the main idea of the Bogoliubov-Valatin transformation.

As indicated by Eq. (\ref{AGM}), the matrix $M$ is Hermitian. So it can
always be diagonalized by a unitary transformation. At first glance, it
seems as if the Hamiltonian of Eq. (\ref{Ham3d}) could be brought into
diagonalization by the same unitary transformation. However, a close
observation shows that such a unitary transformation can, in general,
neither take the form of Eq. (\ref{TM}) nor meet the requirements of Eq. (%
\ref{Cond1}) although it always satisfies the condition of Eq. (\ref{Cond2}%
). Therefore, the unitary transformation for the diagonalization of the
coefficient matrix $M$ can not generally diagonalize the Hamiltonian of Eq. (%
\ref{Ham3d}). That is because both the field $\psi$ and the field $\varphi$
are now the vectors of operators (quantum numbers) rather than the usual
simple vectors of complex variables (classical numbers). For the latter, it
is well known that a Hermitian quadratic form can always be diagonalized by
the unitary transformation for the diagonalization of its coefficient
matrix. In short, the BV diagonalization for a quantum quadratic Hamiltonian
is much more complicated than the unitary diagonalization for the usual
Hermitian quadratic form of complex variables.

Finally, let us analyze the BV method in more detail. It can easily be seen
from Eq. (\ref{TM}) that the transformation matrix $T$ has $4n^{2}$
independent unknown entries. However, Eqs. (\ref{Cond1}) and (\ref{Cond2})
contain $4n^{2}$ and $4n^{2}-2n$ constraints on $T$, respectively. That is
to say, the constraints are much more than the total number of the free
unknown entries of $T$. Therefore, there are two possibilities: (1) Those
constraints are consistent with the requirement of $T$, and thus $T$ has
solutions. (2) The constraints are inconsistent with the requirement of $T$,
and $T$ has no solution. Theoretically, it is very difficult to judge which
case will happen because, as indicated by Eqs. (\ref{Cond1}) and (\ref{Cond2}%
), the constraints constitute $8n^{2}-2n$ coupled quadratic equations for $%
4n^{2}$ free unknowns. Furthermore, it will still be hard to solve for the
multiple unknowns from the multiple equations of second degree even if there
exist solutions for the matrix $T$. Mathematically, these difficulties arise
from the well-known fact that there is no much knowledge about the multiple
equations of second degree with multiple unknowns at present. In practice,
one often has to rely on experience and tricks when he uses the BV method to
resolve practical problems.

To overcome those difficulties, we intend to develop a new theory for BV
transformation. We expect that this theory can not only judge
straightforwardly whether a quantum quadratic Hamiltonian is BV
diagonalizable but also yield the required transformation by a simple
procedure if the Hamiltonian is BV diagonalizable. That is the main
objective of this review.

\subsection{Equation of motion \label{Sec1B}}

As shown in the preceding subsection, the diagonalization scheme adopted by
Bogoliubov and Valatin is merely algebraic. That is to say, the scheme
treats the diagonalization just as a pure algebraic problem, it does not
consider the physics in diagonalization at all. We would like to complement
it with physical contents so as to find the necessary and sufficient
conditions for the diagonalization of a quantum quadratic Hamiltonian.
Simply speaking, we shall take into account the equation of motion of the
system, i.e., the Heisenberg equation.

To show the idea, let us consider the classical system of harmonic
oscillators---the counterpart of the Bose system with a quadratic
Hamiltonian \cite{Goldstein},%
\begin{eqnarray}
H &=&\frac{1}{2}\sum_{i,j=1}^{n}K_{ij}p_{i}p_{j}+\frac{1}{2}%
\sum_{i,j=1}^{n}V_{ij}q_{i}q_{j}  \notag \\
&=&\frac{1}{2}\widetilde{p}Kp+\frac{1}{2}\widetilde{q}Vq,  \label{Ham3}
\end{eqnarray}%
where $q_{i}$ and $p_{i}$ ($i=1,2,\cdots ,n$) are, respectively, the
generalized coordinates and momenta, with $q$ and $p$ being the
corresponding column vectors,%
\begin{equation}
q=\left[
\begin{array}{c}
q_{1} \\
q_{2} \\
\vdots \\
q_{n}%
\end{array}%
\right] ,\text{ \ }p=\left[
\begin{array}{c}
p_{1} \\
p_{2} \\
\vdots \\
p_{n}%
\end{array}%
\right] .
\end{equation}%
The $K$ and $V$ are the kinetic and potential matrices with $K_{ij}$ and $%
V_{ij}$ as their entries, respectively. They are both real and symmetric,
\begin{eqnarray}
\widetilde{K} &=&K>0, \\
\widetilde{V} &=&V\geq 0.
\end{eqnarray}%
It is worthy to emphasize that $K$ is a positive definite matrix, that is
because the kinetic energy is always positive definite. In addition, the
matrix $V$ is only positive semidefinite, the bottom of potential being
chosen as zero.

As well known, $q_{i}$ and $p_{i}$ ($i=1,2,\cdots ,n$) satisfy the following
canonical relations,
\begin{eqnarray}
\{q_{i},q_{j}\} &=&0,  \label{PB1} \\
\{p_{i},p_{j}\} &=&0,  \label{PB2} \\
\{q_{i},p_{j}\} &=&\delta _{ij},  \label{PB3}
\end{eqnarray}%
or equivalently,
\begin{eqnarray}
q\cdot \widetilde{q} &=&0,  \label{Dual1} \\
p\cdot \widetilde{p} &=&0,  \label{Dual2} \\
q\cdot \widetilde{p} &=&I.  \label{Dual3}
\end{eqnarray}%
where $\{a,b\}$ denotes the Poisson bracket of $a$ and $b$, and $a\cdot
b=\{a,b\}$.

Of course, the Bogoliubov-Valatin scheme can be transplanted directly to
diagonalize the classical quadratic Hamiltonian of Eq. (\ref{Ham3}) with
respect to the new generalized coordinates and momenta. However, we would
rather here turn to another way---the canonical equation of motion.

The canonical equation of motion can be deduced from the Hamiltonian of Eq. (%
\ref{Ham3}) and the Poisson brackets of Eqs. (\ref{PB1})--(\ref{PB3}) as
follows,
\begin{eqnarray}
\frac{\mathrm{d}}{\mathrm{d}t}q &=&\{q,H\}=Kp, \\
\frac{\mathrm{d}}{\mathrm{d}t}p &=&\{p,H\}=-Vq.
\end{eqnarray}%
where $t$ denotes the time. As a result, we have
\begin{equation}
\frac{\mathrm{d}^{2}}{\mathrm{d}t^{2}}q=-KVq.  \label{DEqu}
\end{equation}%
That is a homogeneous system of linear ordinary differential equations with
constant coefficients.

From the theory of ordinary differential equations \cite{Walter}, we know
that the solution of the homogeneous linear system above depends on the
eigenvalue problem,
\begin{equation}
\omega ^{2}q=KVq.  \label{EVL}
\end{equation}%
This eigenvalue problem can be solved rigorously with the help of the
Cholesky decomposition of $K$,
\begin{equation}
K=Q\widetilde{Q},  \label{KQQ}
\end{equation}%
where $Q$ is an invertible matrix. The existence of such a decomposition
stems mathematically from the positivity of $K$ \cite{Strang}. By
introducing a temporal variable $\xi $,
\begin{equation}
\xi =Q^{-1}q,
\end{equation}%
Eq. (\ref{EVL}) can be transformed into
\begin{equation}
\omega ^{2}\xi =\Lambda \xi ,
\end{equation}%
where%
\begin{equation}
\Lambda =\widetilde{Q}VQ=\widetilde{\Lambda }\geq 0.  \label{LMD}
\end{equation}%
Just as $V$, the matrix $\Lambda $ is still real, symmetric, and nonnegative
definite. So it can be orthogonally diagonalized,
\begin{equation}
\widetilde{S}\Lambda S=\Gamma ,  \label{SLSG}
\end{equation}%
where
\begin{equation}
\widetilde{S}S=S\widetilde{S}=I,  \label{SSI}
\end{equation}%
\begin{equation}
\Gamma =\left[
\begin{array}{cc}
\begin{array}{cc}
\omega _{1}^{2} &  \\
& \omega _{2}^{2}%
\end{array}
& \text{{\LARGE 0}} \\
\text{{\LARGE 0}} &
\begin{array}{cc}
\ddots &  \\
& \omega _{n}^{2}%
\end{array}%
\end{array}%
\right] .
\end{equation}%
Here $\omega _{i}^{2}\geq 0$ ($i=1,2,\cdots n$) are the eigenvalues of $%
\Lambda $, and $S$ the orthogonal matrix with the eigenvectors of $\Lambda $
as its column vectors.

From Eqs. (\ref{KQQ}), (\ref{LMD}), (\ref{SLSG}), and (\ref{SSI}), it
follows that
\begin{equation}
T^{-1}KVT=\Gamma ,  \label{TKVT}
\end{equation}%
where%
\begin{equation}
T=QS.
\end{equation}%
If we put%
\begin{equation}
T=\left[
\begin{array}{cccc}
v_{1}, & v_{2}, & \cdots , & v_{n}%
\end{array}%
\right] ,
\end{equation}%
where $v_{i}$ ($i=1,2,\cdots ,n$) denote the column vectors of $T$. Equation
(\ref{TKVT}) shows that $v_{i}$ are the eigenvectors of the matrix $KV$,
\begin{equation}
\omega _{i}^{2}v_{i}=KVv_{i},
\end{equation}%
belonging to the eigenvalues $\omega _{i}^{2}$, respectively. In other
words, they are the solutions of the eigenvalue problem of Eq. (\ref{EVL}).
Evidently, they are orthonormal and complete,
\begin{eqnarray}
\widetilde{T}GT &=&I, \\
T\widetilde{T}G &=&I,
\end{eqnarray}%
where $G=K^{-1}$. Namely, they constitute a $n$-dimensional Hilbert space
with $G$ as its metric tensor.

Just as usual, the general solution of Eq. (\ref{DEqu}) can be expanded in
this Hilbert space as
\begin{equation}
q(t)=\sum_{i=1}^{n}\psi _{i}(t)v_{i},
\end{equation}%
where $\psi _{i}(t)$ ($i=1,2,\cdots ,n$) are the expanding coefficients.
But, not as usual, we do not care here how to determine those coefficients
from the initial conditions. Instead, we would rather view this expansion as
a linear transformation,
\begin{equation}
q(t)=T\psi (t),
\end{equation}%
where $\psi (t)$ is the column vector,
\begin{equation}
\psi (t)=\left[
\begin{array}{c}
\psi _{1}(t) \\
\psi _{2}(t) \\
\vdots \\
\psi _{n}(t)%
\end{array}%
\right] .
\end{equation}%
As will be seen later, this view is crucial for the diagonalization of the
Hamiltonian. Since $T$ has full rank, the transformation is invertible. The
inverse is
\begin{equation}
\psi (t)=T^{-1}q(t).  \label{Tf1}
\end{equation}%
Physically, $\psi (t)$ represents the new generalized coordinates, and $q(t)$
the old ones.

The corresponding transformation for the generalized momenta can be deduced
from Eq. (\ref{Dual3}). As is well known, a Poisson bracket is a bilinear
function of its two arguments. This together with Eq. (\ref{Dual3})
indicates that there exists a duality relationship between $p(t)$ and $q(t)$
\cite{Roman}. This duality implies that $p(t)$ will transform
contravariantly with $q(t)$, i.e.,
\begin{equation}
\pi (t)=\widetilde{T}p(t),  \label{Tf2}
\end{equation}%
where $\pi (t)$ represents the new generalized momenta.

Under the transformation of Eqs. (\ref{Tf1}) and (\ref{Tf2}), the
Hamiltonian of the system becomes as follows,%
\begin{eqnarray}
H &=&\frac{1}{2}\widetilde{\pi }\pi +\frac{1}{2}\widetilde{\psi }\Gamma \psi
\notag \\
&=&\frac{1}{2}\sum_{i=1}^{n}\left( \pi _{i}^{2}+\omega _{i}^{2}\psi
_{i}^{2}\right) ,  \label{Harm1}
\end{eqnarray}%
where
\begin{eqnarray}
\psi \cdot \widetilde{\psi } &=&0,  \label{Harm2} \\
\pi \cdot \widetilde{\pi } &=&0,  \label{Harm3} \\
\psi \cdot \widetilde{\pi } &=&I.  \label{Harm4}
\end{eqnarray}%
They are identical to the system of Eq. (\ref{Ham3}) and Eqs. (\ref{Dual1}%
)--(\ref{Dual3}), with the Hamiltonian being diagonalized with respect to
the new generalized coordinates and momenta.

\begin{proposition}
\label{PPS0} A classical quadratic Hamiltonian such as Eq. (\ref{Ham3}) can
be diagonalized with respect to the generalized coordinates and momenta.
\end{proposition}

This instance demonstrates clearly that the equation of motion is a very
effective and powerful weapon for the diagonalization of a quadratic
Hamiltonian, in comparison with the method of the preceding subsection. We
see that the equation of motion can generate a linear transformation in a
very natural way, which can then not only diagonalize the quadratic
Hamiltonian but also ensure the invariance of Poisson brackets.

In the picture of diagonalization, the system is represented by the normal
modes of motion. To diagonalize a quadratic Hamiltonian is therefore
equivalent to seeking the normal modes of the system. Of course, the natural
tool for seeking the normal modes is the equation of motion, from the point
of view of physics. That is the physical interpretation for the
diagonalization. All in all, the canonical equation of motion is a candidate
way to the diagonalization of a classical quadratic Hamiltonian, in addition
to the purely algebraic method due to Bogoliubov and Valatin.

Since Heisenberg equation is the quantum counterpart of the canonical
equation of motion in the classical mechanics, it encourages us to try to
employ Heisenberg equation to realize the diagonalization of the quantum
quadratic Hamiltonian. That is the main idea of this review.

Complying with this idea, we shall first study the BV transformation and
diagonalization of the Bose system, we find that a complete theory can be
established using the Heisenberg equation (Sec. \ref{DTBS}). And then we
study the BV transformation and diagonalization of the Fermi system (Sec. %
\ref{DTFS}), it is parallel to the Bose case. Their applications are
discussed in the following two sections (Sec. \ref{ABS} and Sec. \ref{AFS}).
Afterwards, we would turn to studying the Dirac transformation and
diagonalization, which concern coordinates and momenta. It is found that
they are the generalizations of the BV transformation and diagonalization
(Sec. \ref{GBVT}). An advantage of the Dirac transformation and
diagonalization is that they can be transplanted readily to the complex
collective coordinates and momenta, and therefore have wide applications in
field quantization (Sec. \ref{FQ}). Finally, we would like to clarify the
mathematical essence of the transformation and diagonalization. We find that
the equation of motion can be sublated. The transformation and
diagonalization have nothing to do the equation of motion, but are the
intrinsic and invariant property of a Hermitian quadratic form that is
equipped with commutator or Poisson bracket (Sec. \ref{MEBVD}).

Finally, it is worth noting that we shall confine our interest in this
review only to the diagonalization of quadratic Hamiltonians. It will go out
of our consideration as to whether and how a quadratic Hamiltonian can be
derived and obtained for a real system.

\section{Diagonalization Theory of Bose Systems \label{DTBS}}

In this section, we employ the Heisenberg equation of motion to study the
quadratic Hamiltonian of bososns. It is found that a whole theory of
diagonalization can be developed for the Bose system.

\subsection{Dynamic matrix}

The Heisenberg equation of motion can be derived from Eqs. (\ref{Ham1})--(%
\ref{Com3}),
\begin{eqnarray}
i\frac{\mathrm{d}}{\mathrm{d}t}c &=&\alpha c+\gamma \widetilde{c}^{\dag },
\label{Heqb1} \\
i\frac{\mathrm{d}}{\mathrm{d}t}\widetilde{c}^{\dag } &=&-\gamma ^{\dag }c-%
\widetilde{\alpha }\widetilde{c}^{\dag }.  \label{Heqb2}
\end{eqnarray}%
Here and hereafter, we shall apply the natural units of measurement, i.e., $%
\hbar =c=1$, for convenience. The two equations above can be combined as
\begin{equation}
i\frac{\mathrm{d}}{\mathrm{d}t}\psi =D\psi ,  \label{Heqb3}
\end{equation}%
where
\begin{equation}
D=\left[
\begin{array}{cc}
\alpha & \gamma \\
-\gamma ^{\dag } & -\widetilde{\alpha }%
\end{array}%
\right] .  \label{DM}
\end{equation}%
It should be pointed out that the matrix $D$ is distinct from the matrix $M$
of Eq. (\ref{CMH}),
\begin{equation}
M=\left[
\begin{array}{cc}
\alpha & \gamma \\
\gamma ^{\dag } & \widetilde{\alpha }%
\end{array}%
\right] ,  \label{CM}
\end{equation}%
except both $\alpha =0$ and $\gamma =0$, the trivial case that has been
excluded, \textit{ab initio}, in Sec. \ref{BVTS}. As indicated by Eq. (\ref%
{Ham2}), the matrix $M$ represents the constant coefficients of the
Hamiltonian. In contrast, Eq. (\ref{Heqb3}) demonstrates that the matrix $D$
will control the dynamic behavior of the system.

\begin{definition}
We shall call the matrix present in the Hamiltonian, such as $M$, the
coefficient matrix, and the matrix present in the Heisenberg equation, such
as $D$, the dynamic matrix.
\end{definition}

It is a characteristic feature of the Bose system that the dynamic matrix $D$
is different from the coefficient matrix $M$.

On the other hand, it can be readily seen from Eqs. (\ref{DM}) and (\ref{CM}%
) that the dynamic matrix $D$ and the coefficient matrix $M$ have the
relation,%
\begin{equation}
D=I_{-}M.  \label{DIM}
\end{equation}%
This relation will play a fundamental role in the diagonalization of the
Bose system. Besides, it is worth noting that $D$ is generally not Hermitian
whereas $M$ is Hermitian forever.

Following Sec. \ref{Sec1B}, let us study the eigenvalue problem of Eq. (\ref%
{Heqb3}),\
\begin{equation}
\omega\psi=D\psi.  \label{EVEq}
\end{equation}
We expect that it would generate a linear transformation that could be used
to diagonalize the quadratic Hamiltonian of bosons.

Now that $D\neq M$, there arises a question. As shown in Eqs. (\ref{Ham3d})
and (\ref{Cond2}), the Hamiltonian requires a Hermitian congruence
transformation to diagonalize the coefficient matrix $M$. However, the
Heisenberg equation can, at most, generate a similarity transformation to
diagonalize the dynamic matrix $D$, as can be seen from Eq. (\ref{EVEq}).
Not only the matrices to be diagonalized but also the manners of
diagonalization are different from each other. That is the key problem
occurring in the Bose system.

To solve the problem, let us begin with a survey on the general properties
of the dynamic matrix $D$.

\begin{lemma}
\label{Lmm1} If $\omega$ is an eigenvalue of the dynamic matrix $D$, then $%
-\omega^{\ast}$ will also be an eigenvalue of $D$.
\end{lemma}

\begin{proof}
The characteristic equation of Eq. (\ref{EVEq}) is%
\begin{equation}
\det (\omega I_{+}-D)=\left\vert
\begin{array}{cc}
\omega I-\alpha & -\gamma \\
\gamma ^{\dag } & \omega I+\widetilde{\alpha }%
\end{array}%
\right\vert =0.  \label{ChEq}
\end{equation}%
First, let us perform some elementary row and column operations to the
characteristic determinant,
\begin{eqnarray}
\det (\omega I_{+}-D) &=&\left\vert
\begin{array}{cc}
-I & 0 \\
0 & -I%
\end{array}%
\right\vert \left\vert
\begin{array}{cc}
0 & I \\
I & 0%
\end{array}%
\right\vert  \notag \\
&&\times \left\vert
\begin{array}{cc}
\omega I-\alpha & -\gamma \\
\gamma ^{\dag } & \omega I+\widetilde{\alpha }%
\end{array}%
\right\vert \left\vert
\begin{array}{cc}
0 & I \\
I & 0%
\end{array}%
\right\vert  \notag \\
&=&\left\vert
\begin{array}{cc}
-\omega I-\widetilde{\alpha } & -\gamma ^{\dag } \\
\gamma & -\omega I+\alpha%
\end{array}%
\right\vert .  \label{CHM1}
\end{eqnarray}%
And then take complex conjugate,
\begin{equation}
\det (\omega I_{+}-D)^{\ast }=\left\vert
\begin{array}{cc}
-\omega ^{\ast }I-\alpha & -\gamma \\
\gamma ^{\dag } & -\omega ^{\ast }I+\widetilde{\alpha }%
\end{array}%
\right\vert ,  \label{CHM2}
\end{equation}%
where we have used the facts for the Bose system: $\alpha ^{\dag }=\alpha $
and $\widetilde{\gamma }=\gamma $. Paying attention to
\begin{equation}
\det (\omega I_{+}-D)^{\ast }=\det (\omega I_{+}-D)=0,
\end{equation}%
we have%
\begin{equation}
\left\vert
\begin{array}{cc}
-\omega ^{\ast }I-\alpha & -\gamma \\
\gamma ^{\dag } & -\omega ^{\ast }I+\widetilde{\alpha }%
\end{array}%
\right\vert =0.
\end{equation}%
One reaches the lemma immediately by comparing this equation with Eq. (\ref%
{ChEq}).
\end{proof}

This lemma shows that the eigenvalues of the dynamic matrix $D$ will appear
in pairs if they exist. When one of a pair is $\omega$, the other is $%
-\omega^{\ast}$.

Physically, this property of the dynamic matrix $D$ originates from the
Hermitian symmetry of the Hamiltonian: $H^{\dag}=H$. This symmetry implies
that, if
\begin{equation}
c(t)=c_{0}\exp(\mp i\omega t)
\end{equation}
is a solution of Eq. (\ref{Heqb1}), then
\begin{equation}
\widetilde{c}^{\dag}(t)=\widetilde{c}_{0}^{\dag}\exp(\pm i\omega^{\ast}t)
\end{equation}
will be the solution of Eq. (\ref{Heqb2}). That is to say, if
\begin{equation}
\psi(t)=\psi_{0}\exp(-i\omega t)
\end{equation}
is a solution of Eq. (\ref{Heqb3}), the
\begin{equation}
\psi(t)=\psi_{0}\exp(i\omega^{\ast}t)
\end{equation}
must also be a solution of Eq. (\ref{Heqb3}).

\begin{lemma}
\label{Lmm2} If $v(\omega)$ is an eigenvector belonging to the eigenvalue $%
\omega$ of the dynamic matrix $D$, then $v(-\omega^{\ast})$ will be an
eigenvector belonging to the eigenvalue $-\omega^{\ast}$. Here the vector $%
v(-\omega^{\ast})$ is defined by
\begin{equation}
v(-\omega^{\ast})=\Sigma_{x}v^{\ast}(\omega)  \label{VWW}
\end{equation}
with
\begin{equation}
\Sigma_{x}=\left[
\begin{array}{cc}
0 & I \\
I & 0%
\end{array}
\right] .
\end{equation}
\end{lemma}

\begin{proof}
Substituting $v(\omega)$ into Eq. (\ref{EVEq}), one has%
\begin{equation}
(\omega I_{+}-D)v(\omega)=0.
\end{equation}
It follows that%
\begin{equation}
-\Sigma_{x}(\omega I_{+}-D)\Sigma_{x}\Sigma_{x}v(\omega)=0,
\end{equation}
and that%
\begin{equation}
\left[ -\Sigma_{x}(\omega I_{+}-D)\Sigma_{x}\right] ^{\ast}\left[
\Sigma_{x}v(\omega)\right] ^{\ast}=0.
\end{equation}
From Eqs. (\ref{CHM1}) and (\ref{CHM2}), one can easily see that%
\begin{equation}
\left[ -\Sigma_{x}(\omega I_{+}-D)\Sigma_{x}\right] ^{\ast}=-\omega^{\ast
}I_{+}-D,
\end{equation}
he thus gets%
\begin{equation}
(-\omega^{\ast}I_{+}-D)\Sigma_{x}v^{\ast}(\omega)=0.
\end{equation}
This means that $\Sigma_{x}v^{\ast}(\omega)$ is an eigenvector belonging to
the eigenvalue $-\omega^{\ast}$. That is just Eq. (\ref{VWW}).
\end{proof}

This lemma shows that, for a given pair of eigenvalues ($\omega,-\omega^{%
\ast }$), their eigenvectors can be formed into pairs according to Eq. (\ref%
{VWW}).

In the classical case, the dynamic matrix $D_{cl}$ of Eq. (\ref{DEqu}) is
\begin{equation}
D_{cl}=KV.
\end{equation}%
It indicates that $D_{cl}$ is the production of a positive definite matrix $%
K $ and a Hermitian matrix $V$. Although $D_{cl}$ is not Hermitian in
general, it is diagonalizable and all its eigenvalues are real.
Mathematically, that is because the matrix $K$ has Cholesky decomposition,
as has been seen in Sec. \ref{Sec1B}. Now, as shown in Eq. (\ref{DIM}), the
dynamic matrix $D$ for a Bose system is the production of an indefinite
matrix $I_{-}$ and a Hermitian matrix $M$. There exists no Cholesky
decomposition for the matrix $I_{-}$, and there is no guarantee for $D$ to
be diagonalizable. Furthermore, the eigenvalues of $D$ will, in general, be
complex other than real even if $D$ is diagonalizable. In a word, the
present situation is much more involved than the classical case. To be
clear, let us take a look at the simplest case, i.e., the Hamiltonian of Eq.
(\ref{Ham1}) with $n=1$.

\begin{example}
\label{Example1}%
\begin{equation}
H=\alpha\,c^{\dag}c+\frac{1}{2}\gamma\,c^{\dag}c^{\dag}+\frac{1}{2}%
\gamma^{\ast}\,cc.
\end{equation}
\end{example}

\begin{solution}
Apparently, the dynamic matrix $D$ is a $2\times 2$ matrix,
\begin{equation}
D=\left[
\begin{array}{cc}
\alpha & \gamma \\
-\gamma ^{\ast } & -\alpha%
\end{array}%
\right] .
\end{equation}%
The eigenvalue equation is
\begin{equation}
\left[
\begin{array}{cc}
\alpha & \gamma \\
-\gamma ^{\ast } & -\alpha%
\end{array}%
\right] \left[
\begin{array}{c}
x \\
y%
\end{array}%
\right] =\omega \left[
\begin{array}{c}
x \\
y%
\end{array}%
\right] .  \label{Exm1EgEv}
\end{equation}%
Obviously, the eigenvalues can be obtained from the characteristic equation,
\begin{equation}
\omega ^{2}-\alpha ^{2}+\left\vert \gamma \right\vert ^{2}=0,
\end{equation}%
the results are
\begin{equation}
\omega =\left\{
\begin{array}{ll}
\pm \sqrt{\alpha ^{2}-\left\vert \gamma \right\vert ^{2}}, & \left\vert
\alpha \right\vert >\left\vert \gamma \right\vert \\
0, & \left\vert \alpha \right\vert =\left\vert \gamma \right\vert \\
\pm i\sqrt{\left\vert \gamma \right\vert ^{2}-\alpha ^{2}}, & \left\vert
\alpha \right\vert <\left\vert \gamma \right\vert .%
\end{array}%
\right.
\end{equation}%
Namely, there are two real eigenvalues when $\left\vert \alpha \right\vert
>\left\vert \gamma \right\vert $, a zero eigenvalue when $\left\vert \alpha
\right\vert =\left\vert \gamma \right\vert $, and two imaginary eigenvalues
when $\left\vert \alpha \right\vert <\left\vert \gamma \right\vert $.

It is easy to show that, if $\left\vert \alpha\right\vert =\left\vert
\gamma\right\vert $, there exists only one eigenvector,
\begin{equation}
\left[
\begin{array}{c}
x \\
y%
\end{array}
\right] =\left[
\begin{array}{c}
1 \\
\mp\mathrm{e}^{-i\theta}%
\end{array}
\right] ,
\end{equation}
where $\theta=\arg(\gamma)$ is the argument of $\gamma$, and the signs $\mp$
correspond to $\alpha=\pm\left\vert \gamma\right\vert $, respectively.
Therefore, the dynamic matrix $D$ can not be diagonalized when $\left\vert
\alpha\right\vert =\left\vert \gamma\right\vert $.

When $\left\vert \alpha\right\vert <\left\vert \gamma\right\vert $, the
dynamic matrix $D$ has two linearly independent eigenvectors. It is thus
diagonalizable, but its eigenvalues are both imaginary.

When $\left\vert \alpha\right\vert >\left\vert \gamma\right\vert $, the
dynamic matrix $D$ has two linearly independent eigenvectors, it is also
diagonalizable. In particular, its eigenvalues are both real.

In sum, the dynamic matrix of this Bose system has the same property as that
of the classical system only when $\left\vert \alpha \right\vert >\left\vert
\gamma \right\vert $: It is diagonalizable, and its eigenvalues are real.
\end{solution}

This simple example exhibits clearly the complexity of Bose systems. To
resolve this complexity, we shall study first the necessary and then the
sufficient condition for the diagonalization of a quadratic Hamiltonian of
bosons, which constitute the themes of the following two subsections,
respectively.

\subsection{Necessary condition for diagonalization}

To be clear and definite in the following, we would first give three
definitions here.

\begin{definition}
\label{BVM} If an invertible matrix has the form as Eq. (\ref{TM}), we call
it a Bogoliubov-Valatin matrix.
\end{definition}

\begin{definition}
\label{BVT} A linear transformation defined by Eq. (\ref{Transf1}) will be
called a Bogoliubov-Valatin transformation if $T$ is a BV matrix and $\psi $%
\ is a standard bosonic field, i.e., it satisfies Eq. (\ref{Com4}).
\end{definition}

\begin{definition}
\label{BVD} The quadratic Hamiltonian defined in Eq. (\ref{Ham1}) is said to
be Bogoliubov-Valatinianly diagonalizable if there is such a BV
transformation that can fulfill both the conditions of Eqs. (\ref{Cond1})
and (\ref{Cond2}).
\end{definition}

It should be pointed out that a BV transformation is not required to satisfy
either the condition of Eq. (\ref{Cond1}) or the condition of Eq. (\ref%
{Cond2}), according to the definition \ref{BVT}.

\begin{lemma}
\label{Lmm3} The inverse of a BV matrix is also a BV matrix.
\end{lemma}

\begin{proof}
From Eq. (\ref{TM}), it is easy to show that a matrix $T$ is a BV matrix if
and only if
\begin{equation}
\Sigma_{x}T^{\ast}\Sigma_{x}=T.
\end{equation}
Taking the inverses of the two sides, we obtain
\begin{equation}
\Sigma_{x}\left( T^{-1}\right) ^{\ast}\Sigma_{x}=T^{-1}.  \label{STS}
\end{equation}
This demonstrates that $T^{-1}$ is also a BV matrix.
\end{proof}

From Eq. (\ref{Psi}), one can easily see that
\begin{equation}
\psi=\left( \widetilde{\Sigma_{x}\psi}\right) ^{\dag}.
\end{equation}
This is a basic symmetry of the field operator, we shall call it the
involution symmetry, for convenience. Mathematically, this symmetry roots
from the fact that $c_{i}$ and $c_{i}^{\dag}$ are not independent, but are
the Hermitian conjugates of each other. Conversely, if a field has the
involution symmetry as above, its component operators can not be
independent, there must exist some relationship among them.

Now, consider the new field $\varphi$ defined by the BV transformation of
Eq. (\ref{Transf1}). It can be given by the inverse transformation,
\begin{equation}
\varphi=T^{-1}\psi.  \label{FTP}
\end{equation}
Obviously,
\begin{equation}
\Sigma_{x}\varphi=\Sigma_{x}T^{-1}\Sigma_{x}\Sigma_{x}\psi,
\end{equation}
which results in
\begin{equation}
\left( \widetilde{\Sigma_{x}\varphi}\right) ^{\dag}=\Sigma_{x}\left(
T^{-1}\right) ^{\ast}\Sigma_{x}\left( \widetilde{\Sigma_{x}\psi}\right)
^{\dag}.
\end{equation}
This implies that
\begin{equation}
\varphi=\left( \widetilde{\Sigma_{x}\varphi}\right) ^{\dag}.
\end{equation}
where Eqs. (\ref{STS})--(\ref{FTP}) have been used. Therefore, the new field
will have the same involution symmetry as the old one after a BV
transformation.

\begin{lemma}
\label{Lmm4} The involution symmetry of the field operator is conserved for
the Bose system after a BV transformation. Namely, the involution symmetry
is an invariant property of the BV transformation.
\end{lemma}

Suppose that $\varphi $ is a new field, it thus has the involution symmetry.
As mentioned above, its component operators will not be independent. In
fact, it is easy to show that $\varphi $ must have the same form as the old
field $\psi $ defined in Eq. (\ref{Psi}). In other words, it can be
represented as follows,
\begin{equation}
\varphi =\left[
\begin{array}{c}
d \\
\widetilde{d^{\dag }}%
\end{array}%
\right] ,
\end{equation}%
where
\begin{equation}
d=\left[
\begin{array}{c}
d_{1} \\
d_{2} \\
\vdots \\
d_{n}%
\end{array}%
\right] ,\ d^{\dag }=\left[
\begin{array}{cccc}
d_{1}^{\dag }, & d_{2}^{\dag }, & \cdots , & d_{n}^{\dag }%
\end{array}%
\right] .
\end{equation}%
Here $d_{i}$ and $d_{i}^{\dag }$ represent a new pair of operators, which
are Hermitianly conjugate to each other. Their commutation rules can be
derived from the inverse transformation,%
\begin{equation}
\varphi \cdot \varphi ^{\dag }=T^{-1}\psi \cdot \psi ^{\dag }\left(
T^{-1}\right) ^{\dag }.
\end{equation}%
As stated in the definition \ref{BVT}, the old field $\psi $ satisfies the
standard commutation rule of Eq. (\ref{Com4}). We therefore obtain%
\begin{equation}
\varphi \cdot \varphi ^{\dag }=T^{-1}I_{-}\left( T^{-1}\right) ^{\dag }.
\label{NComm}
\end{equation}%
The commutation rule for the new field $\varphi $ is determined wholly by
the BV matrix $T$, it may not be standard,
\begin{equation}
\varphi \cdot \varphi ^{\dag }\neq I_{-}.
\end{equation}%
It is standard if and only if $T$ satisfies the condition of Eq. (\ref{Cond1}%
), which is equivalent to
\begin{equation}
T^{-1}I_{-}\left( T^{-1}\right) ^{\dag }=I_{-}.
\end{equation}%
In a word, if a BV transformation satisfies the condition of Eq. (\ref{Cond1}%
), the new field is a standard bosonic field; if it further satisfies the
condition of Eq. (\ref{Cond2}), the Hamiltonian of Eq. (\ref{Ham1}) gets BV
diagonalized.

We shall leave it to the next subsection to discuss how to obtain a BV
transformation and make it fulfill the two conditions of Eqs. (\ref{Cond1})
and (\ref{Cond2}). Here and now, we would, above all, show a basic property
of the BV transformation, which will play the central role in the theory of
diagonalization of quantum quadratic Hamiltonians.

\begin{lemma}
\label{Similar} Under a BV transformation, the two dynamic matrices
respectively for the old and new fields will be similar to each other.
\end{lemma}

\begin{proof}
Obviously, under a BV transformation as given in Eqs. (\ref{Transf1}) and (%
\ref{TM}), the equation of motion of the new field $\varphi$ will still be
linear in $\varphi$ itself, i.e.,
\begin{equation}
i\frac{\mathrm{d}}{\mathrm{d}t}\varphi=[\varphi,\,H]=D_{1}\varphi,
\end{equation}
where $D_{1}$ is the dynamic matrix for the new field $\varphi$. Apparently,
this equation has the same form as that for the old field $\psi$,
\begin{equation}
i\frac{\mathrm{d}}{\mathrm{d}t}\psi=D\psi,
\end{equation}
where $D$ is the dynamic matrix for the old field $\psi$. On the other hand,
it follows from Eq. (\ref{Transf1}) that
\begin{equation}
i\frac{\mathrm{d}}{\mathrm{d}t}\psi=Ti\frac{\mathrm{d}}{\mathrm{d}t}\varphi.
\end{equation}
With the above two equations of motion for $\psi$ and $\varphi$, this
equation can be expressed as
\begin{equation}
D\psi=TD_{1}\varphi.
\end{equation}
Substituting $\psi$ further with Eq. (\ref{Transf1}), one has
\begin{equation}
\left( T^{-1}DT-D_{1}\right) \varphi=0.
\end{equation}
This equation is equivalent to%
\begin{equation}
v_{i}\varphi=0,\text{ \ }\forall i\in\{1,2,\cdots,2n\},
\end{equation}
where $v_{i}$ are the row vectors of the matrix $T^{-1}DT-D_{1}$. It implies
that
\begin{equation}
v_{i}=0,\text{ \ }\forall i\in\{1,2,\cdots,2n\}.
\end{equation}
That is%
\begin{equation}
T^{-1}DT-D_{1}=0,
\end{equation}
viz.,%
\begin{equation}
D_{1}=T^{-1}DT.
\end{equation}
In other words, the dynamic matrix will vary in a similar manner under a BV
transformation.
\end{proof}

\begin{proposition}
\label{PPS1} If a quadratic Hamiltonian of bosons can be BV diagonalized,
then its dynamic matrix is diagonalizable, and all the eigenvalues of the
dynamic matrix will be real.
\end{proposition}

\begin{proof}
Suppose that the diagonalized form of the Hamiltonian is%
\begin{equation}
H=\sum_{i=1}^{n}\omega _{i}\,d_{i}^{\dag }d_{i}+C,
\end{equation}%
where $C$ is a real constant. Since $H$ is Hermitian, all $\omega _{i}$ must
be real, i.e., $\omega _{i}\in
\mathbb{R}
$. By definition, the new field $\varphi $ satisfies the standard
commutation rule,
\begin{equation}
\varphi \cdot \varphi ^{\dag }=I_{-}.
\end{equation}%
From the two equations above, the dynamic matrix $D_{1}$ for the new field $%
\varphi $ can be found as
\begin{equation}
D_{1}=\mathrm{diag}(\omega _{1},\omega _{2},\cdots ,\omega _{n},-\omega
_{1},-\omega _{2},\cdots ,-\omega _{n}),
\end{equation}%
where $\mathrm{diag}(a_{1},a_{2},\cdots ,a_{m})$ denotes the diagonal matrix
with $a_{1}$, $a_{2}$, $\cdots $, $a_{m}$ on the main diagonal. From this,
one can reach the proposition by the lemma \ref{Similar}.
\end{proof}

\begin{definition}
A dynamic matrix is said to be physically diagonalizable if it is
diagonalizable, and all its eigenvalues are real.
\end{definition}

By this definition, the proposition can be restated as follows.

If a quadratic Hamiltonian of bosons can be BV diagonalized, its dynamic
matrix is physically diagonalizable.

\begin{corollary}
\label{Corollary1} The quadratic Hamiltonian of bosons defined in Eq. (\ref%
{Ham1}) can not be BV diagonalized if the coefficient submatrix $\alpha$
vanishes identically, i.e., $\alpha=0$.
\end{corollary}

\begin{proof}
In this case, the dynamic matrix of Eq. (\ref{DM}) reduces to
\begin{equation}
D=\left[
\begin{array}{cc}
0 & \gamma \\
-\gamma ^{\dag } & 0%
\end{array}%
\right] .
\end{equation}%
Obviously, it is anti-Hermitian, and thus unitarily diagonalizable. Since $%
\gamma \neq 0$, the eigenvalues of $D$ can not all be zero, some of them
must be purely imaginary. In other words, the dynamic matrix $D$ is
diagonalizable but not physically diagonalizable, which is inconsistent with
the necessary condition for the BV diagonalization of a Hamiltonian.
\end{proof}

This corollary shows that it is not all the quadratic Hamiltonians of bosons
that can be BV diagonalized. What kind of Hamiltonians is BV diagonalizable?
To answer it, one needs to study the sufficient condition for the BV
diagonalization.

\subsection{Sufficient condition for diagonalization}

In the preceding subsection, we have already obtained the necessary
condition for the BV diagonalization. Henceforth, we would presume that the
necessary condition holds for the Bose system. Starting from this
presumption, we shall search the sufficient condition for the BV
diagonalization in this subsection.

By definition, the dynamic matrix $D$ is of size $2n$. If $D$ is physically
diagonalizable, it has a complete set of totally $2n$ linearly independent
eigenvectors. We have learned from the lemmas \ref{Lmm1} and \ref{Lmm2}\
that the eigenvalues and eigenvectors of $D$ will appear in pairs. Let us
continue this discussion about pairing.

\begin{lemma}
\label{Lmm5} If the dynamic matrix $D$ is physically diagonalizable, then,
for each pair of nonzero eigenvalues, i.e., $(\omega,-\omega)$ with $%
\omega\neq0$, they have the same degeneracy. In other words, their
eigenspaces have the same dimension.
\end{lemma}

\begin{proof}
According to the lemma \ref{Lmm2}, if the eigenvalue $\omega $ has $m$
eigenvectors,
\begin{equation}
v_{l}(\omega ),\text{ \ }l=1,2,\cdots ,m,
\end{equation}%
the corresponding $m$ vectors $v_{l}(-\omega )$,
\begin{equation}
v_{l}(-\omega )=\Sigma _{x}v_{l}^{\ast }(\omega ),\text{ \ }l=1,2,\cdots ,m,
\label{VP1}
\end{equation}%
are the eigenvectors belonging to the eigenvalue $-\omega $. It can be
readily confirmed that the vectors $v_{l}(-\omega )$ ($l=1,2,\cdots ,m$) are
linearly independent if and only if the vectors $v_{l}(\omega )$ ($%
l=1,2,\cdots ,m$) are linearly independent.

In other words, if $\omega$ has $m$ linearly independent eigenvectors, then $%
-\omega$ also has $m$ linearly independent eigenvectors, and vice versa.
That is to say, the eigenvalues $\omega$ and $-\omega$ have the same
degeneracy, their eigenspaces have the same dimension.
\end{proof}

The proof above shows that, if the basis vectors for the eigenspace of $%
\omega$ ($\omega\neq0$) have been determined, the basis vectors for the
eigenspace of $-\omega$ can be chosen as Eq. (\ref{VP1}), and vice versa.

\begin{lemma}
\label{Lmm6} If the dynamic matrix $D$ is physically diagonalizable and has
zero eigenvalue, the eigenspace of zero eigenvalue is even dimensional. In
particular, its basis vectors can be chosen and grouped as%
\begin{equation}
v_{m+l}(0)=\Sigma_{x}v_{l}^{\ast}(0),\text{ \ }l=1,2,\cdots,m,  \label{VP2}
\end{equation}
where $2m$ ($m\in%
\mathbb{N}
$) is the dimension of the eigenspace of zero eigenvalue.
\end{lemma}

\begin{proof}
The first point is a direct result of the lemma \ref{Lmm5}.

As to the second one, let us consider the eigenvalue equation,
\begin{equation}
Dv(0)=0.  \label{DW0}
\end{equation}%
That is a homogeneous system of $2n$ linear equations, its solution set
forms the eigenspace of zero eigenvalue.

Since the dimension of the eigenspace of zero eigenvalue is $2m$, we have $%
\mathrm{rank}(D)=2n-2m$ where $\mathrm{rank}(A)$ denotes the rank of the
matrix $A$. It means that the vector $v(0)$ has $2m$ free unknown
components. Therefore, we can choose the following $2m$ components of $v(0)$%
,
\begin{equation}
v^{\alpha }(0),\text{ \ }\alpha =1,2,\cdots ,m,n+1,n+2,\cdots ,n+m,
\end{equation}%
as the free unknowns. First, let the free unknowns be respectively as
follows,
\begin{equation}
v_{l}^{\alpha }(0)=\delta _{\alpha l},  \label{VW01}
\end{equation}%
where%
\begin{eqnarray}
l &=&1,2,\cdots ,m, \\
\alpha &=&1,2,\cdots ,m,n+1,n+2,\cdots ,n+m.
\end{eqnarray}%
We obtain from Eq. (\ref{DW0}) the first group of eigenvectors,
\begin{equation}
v_{l}(0),\text{ \ }l=1,2,\cdots ,m.
\end{equation}%
Clearly, they are linearly independent. Then, using the lemma \ref{Lmm2}, we
have the other group of eigenvectors,
\begin{equation}
v_{m+l}(0)=\Sigma _{x}v_{l}^{\ast }(0),\text{ \ }l=1,2,\cdots ,m.
\end{equation}%
They are also linearly independent. The definitions for $v_{m+l}(0)$ show
that
\begin{equation}
v_{m+l}^{\alpha }(0)=\delta _{\alpha ,m+l},  \label{VW02}
\end{equation}%
where%
\begin{eqnarray}
l &=&1,2,\cdots ,m, \\
\alpha &=&1,2,\cdots ,m,n+1,n+2,\cdots ,n+m.
\end{eqnarray}

Equations (\ref{VW01}) and (\ref{VW02}) imply that the two groups are also
linearly independent. The combination of the two groups has $2m$ linearly
independent eigenvectors, they form a basis for the eigenspace of zero
eigenvalue. All in all, the basis vectors for the eigenspace of zero
eigenvalue can be chosen and grouped as Eq. (\ref{VP2}).
\end{proof}

Following the two lemmas above, if the dynamic matrix $D$ is physically
diagonalizable, it is enough for us to find a half of the eigenvectors of $D$%
, the other half can be determined by Eqs. (\ref{VP1}) and (\ref{VP2}). In
other words, the eigenvalues and eigenvectors of $D$ can be formed into
pairs according to Eqs. (\ref{VP1}) and (\ref{VP2}). Each pair has two
linearly independent eigenvectors with opposite eigenvalues. Such a pair
will be called a dynamic mode pair. Consequently, there are totally $n$
dynamic mode pairs. Henceforth, Eqs. (\ref{VP1}) and (\ref{VP2}) will be
used as the conventions for the pairs of dynamic modes.

Suppose that the dynamic matrix $D$ is physically diagonalizable. One can
construct a linear transformation as in Sec. \ref{Sec1B},
\begin{equation}
\psi =T\varphi ,  \label{PsiToPhi}
\end{equation}%
where $\varphi $ represents the new field operator, and $T$ is the matrix
which consists of all the eigenvectors of $D$,
\begin{widetext}
\begin{equation}
T=\left[
\begin{array}{cccccccc}
v(\omega _{1}), & v(\omega _{2}), & \cdots , & v(\omega _{n}), & v(-\omega
_{1}), & v(-\omega _{2}), & \cdots , & v(-\omega _{n})%
\end{array}%
\right] .  \label{TPair}
\end{equation}%
\end{widetext}
Here each eigenvalue is counted up to its multiplicity, and the $n$ dynamic
mode pairs are separated and arranged sequentially into the left and right
halves of the matrix $T$. Since the dynamic matrix $D$ is supposed to be
physically diagonalizable, the matrix $T$ has full rank and is hence
nonsingular and invertible. This analysis demonstrates that an invertible
linear transformation can be derived from the equation of motion if the
dynamic matrix of the system is physically diagonalizable.

\begin{definition}
If the dynamic matrix $D$ is physically diagonalizable, then a linear
transformation can be defined by Eqs. (\ref{PsiToPhi}) and (\ref{TPair}). We
shall call it the derivative transformation, and call the corresponding
matrix $T$ the derivative matrix.
\end{definition}

By use of both the conventions of Eqs. (\ref{VP1}) and (\ref{VP2}), we have
\begin{widetext}
\begin{equation}
\Sigma _{x}T^{\ast }=\left[
\begin{array}{cccccccc}
v(-\omega _{1}), & v(-\omega _{2}), & \cdots , & v(-\omega _{n}), & v(\omega
_{1}), & v(\omega _{2}), & \cdots , & v(\omega _{n})%
\end{array}%
\right] .
\end{equation}%
\end{widetext}
Paying attention to the fact that $\Sigma _{x}$ is an elementary matrix, the
right multiplication by it represents switching the left and right halves of
the square matrix standing left to it. So we obtain
\begin{equation}
\Sigma _{x}T^{\ast }\Sigma _{x}=T.  \label{inv3}
\end{equation}%
That is an important property of the the derivative transformation, it
implies that the derivative matrix $T$ has the same form as that of (\ref{TM}%
). According to the definitions \ref{BVM} and \ref{BVT} as well as the lemma %
\ref{Lmm4}, we obtain the lemma for the derivative transformation.

\begin{lemma}
\label{Lmm7} If the dynamic matrix $D$ is physically diagonalizable, its
derivative matrix is a BV matrix. The corresponding derivative
transformation is a BV transformation, and thus it will conserve the
involution symmetry of the field operator of the Bose system.
\end{lemma}

As shown by the proof above, it is the two conventions of Eqs. (\ref{VP1})
and (\ref{VP2}) that guarantee that the derivative transformation is a BV
transformation. Consequently, one must comply with both of them when he
constructs a BV transformation.

Up to now, we have proved that a BV transformation can be generated by the
Heisenberg equation of motion if the dynamic matrix of the system is
physically diagonalizable.

Although the new field can inherit the involution symmetry through the
derivative BV transformation, its commutation rule is not always standard.
From now on, we shall turn to handling this problem. As already known, it is
determined by Eq. (\ref{NComm}). Usually, it does not matter what the
magnitude of an eigenvector is. Above, when constructing the derivative BV
transformation of Eqs. (\ref{PsiToPhi}) and (\ref{TPair}), we did not
consider the magnitudes of the eigenvectors either. Nevertheless, Eq. (\ref%
{NComm}) shows that the magnitudes of the eigenvectors can change the
commutation rule of the new field heavily. Therefore, it is necessary for us
to take into account the magnitudes of the eigenvectors if we want to make
the new commutation rule standard. Mathematically, the magnitude of a vector
concerns the metric on the linear space. Therefore, we introduce, first, a
sesquilinear form \cite{Roman,Chen} for the Bose system.

\begin{definition}
Using $I_{-}$, we define a sesquilinear form $\phi$ on the $2n$-dimensional
unitary space $%
\mathbb{C}
^{2n}$ as follows,%
\begin{equation}
\phi(\left\vert x\right\rangle ,\text{ }\left\vert y\right\rangle
)=\left\langle x\right\vert I_{-}\left\vert y\right\rangle ,\text{ \ }%
\forall\left\vert x\right\rangle ,\text{ }\left\vert y\right\rangle \in%
\mathbb{C}
^{2n},
\end{equation}
where Dirac notations have been used for the vectors of $%
\mathbb{C}
^{2n}$. The form $\phi$ will be directly referred to as the metric $I_{-}$,
too.
\end{definition}

The definition is proper because the matrix $I_{-}$ is Hermitian with
respect to the standard basis and standard inner product of the unitary
space $%
\mathbb{C}
^{2n}$. Besides, the metric vector space defined by the form $\phi$ is
nondegenerate because the matrix $I_{-}$ is nonsingular.

Although this form is indefinite, it is very useful for clarifying the
properties of the eigenvectors of the dynamic matrix.

\begin{lemma}
\label{Ortho} If the dynamic matrix $D$ is physically diagonalizable, its
eigenspaces will be orthogonal to each other with respect to the metric $%
I_{-}$.
\end{lemma}

\begin{proof}
Since the dynamic matrix $D$ is physically diagonalizable, the linear space $%
\mathbb{C}
^{2n}$ can be decomposed into the direct sum of the eigenspaces of $D$,
\begin{equation}
\mathbb{C}
^{2n}=E_{1}\oplus E_{2}\oplus\cdots\oplus E_{m},  \label{DSum}
\end{equation}
where $E_{k}$ ($1\leq k\leq m$ with $1\leq m\leq2n$) are the eigenspaces of $%
D$, which belong to the eigenvalues $\omega_{k}\in%
\mathbb{R}
$ respectively,
\begin{equation}
D\left\vert k\mu\right\rangle =\omega_{k}\left\vert k\mu\right\rangle ,\text{
\ }\forall\left\vert k\mu\right\rangle \in E_{k}.
\end{equation}
Here a label $\mu$ is added to distinguish the vectors of the eigenspace $%
E_{k}$.

Using the relation of Eq. (\ref{DIM}), the equation above can be
reformulated as
\begin{equation}
M\left\vert k\mu\right\rangle =\omega_{k}I_{-}\left\vert k\mu\right\rangle .
\end{equation}
As a result, we obtain%
\begin{equation}
\left\langle l\nu\right\vert M\left\vert k\mu\right\rangle =\omega
_{k}\left\langle l\nu\right\vert I_{-}\left\vert k\mu\right\rangle .
\label{vMv}
\end{equation}
By complex conjugate, we have%
\begin{equation}
\left\langle k\mu\right\vert M\left\vert l\nu\right\rangle =\omega
_{k}\left\langle k\mu\right\vert I_{-}\left\vert l\nu\right\rangle ,
\end{equation}
the eigenvalue $\omega_{k}$ being real. Again, from Eq. (\ref{vMv}), we have%
\begin{equation}
\left\langle k\mu\right\vert M\left\vert l\nu\right\rangle =\omega
_{l}\left\langle k\mu\right\vert I_{-}\left\vert l\nu\right\rangle .
\end{equation}
The combination of the two equations above leads to
\begin{equation}
\left\langle k\mu\right\vert I_{-}\left\vert l\nu\right\rangle =0,\text{ \
if }\omega_{k}\neq\omega_{l}.  \label{ViVj}
\end{equation}
This equation demonstrates that the different eigenspaces of $D$ are
orthogonal to each other, i.e., $E_{k}\perp E_{l}$ ($k\neq l$), with respect
to the metric $I_{-}$ defined above.
\end{proof}

The lemma shows that the whole space of $%
\mathbb{C}
^{2n}$ can be decomposed into an orthogonal direct sum of the eigenspaces of
the dynamic matrix $D$ if $D$ is physically diagonalizable,
\begin{equation}
\mathbb{C}
^{2n}=E_{1}\odot E_{2}\odot\cdots\odot E_{m}.  \label{ODS}
\end{equation}
Meanwhile, as indicated by Eqs. (\ref{vMv}) and (\ref{ViVj}), the
coefficient matrix $M$ becomes block diagonalized with respect to the
eigenspaces of $D$,%
\begin{equation}
\left\langle l\nu\right\vert M\left\vert k\mu\right\rangle =\omega
_{k}\left\langle k\nu\right\vert I_{-}\left\vert k\mu\right\rangle \delta
_{kl}.  \label{vMvNew}
\end{equation}

\begin{lemma}
\label{OrNormal} If the dynamic matrix $D$ is physically diagonalizable,
then, for each eigenspace of $D$, there exists an orthonormal basis with
respect to the metric $I_{-}$. That is,%
\begin{equation}
\left\langle k\mu \right\vert I_{-}\left\vert k\nu \right\rangle =\lambda
_{\mu }\delta _{\mu \nu },\text{ \ }\left\vert k\mu \right\rangle \in E_{k},%
\text{ }\left\vert k\nu \right\rangle \in E_{k},
\end{equation}%
where $1\leq \mu ,$ $\nu \leq n_{k}$ with $n_{k}$ being the dimension of the
eigenspace $E_{k}$, and $\lambda _{\mu }=+1$ or $-1$.
\end{lemma}

\begin{proof}
Taking notice of Eq. (\ref{ODS}), the metric $I_{-}$ must also be a
nonsingular sesquilinear form on each eigenspace. Otherwise, it is singular
on the whole space $%
\mathbb{C}
^{2n}$, which leads to an evident contradiction.
\end{proof}

Obviously, if
\begin{equation}
\mathcal{B}_{k}\triangleq \{\left\vert k\mu \right\rangle \in E_{k}|1\leq
\mu \leq n_{k}\}
\end{equation}%
is an orthonormal basis for the eigenspace $E_{k}$, then the union of $%
\mathcal{B}_{1}$, $\mathcal{B}_{2}$, $\cdots $, $\mathcal{B}_{m}$, i.e.,
\begin{eqnarray}
\mathcal{B} &=&\mathcal{B}_{1}\cup \mathcal{B}_{2}\cup \cdots \cup \mathcal{B%
}_{m}  \notag \\
&=&\{\left\vert k\mu \right\rangle \in E_{k}|\text{ }1\leq k\leq m,\text{ }%
1\leq \mu \leq n_{k}\},
\end{eqnarray}%
will form an orthonormal basis for $%
\mathbb{C}
^{2n}$.

Now, for a nonzero eigenvalue $\omega $ ($\omega \neq 0$), we can choose an
orthonormal basis for it,%
\begin{equation}
v_{l}(\omega ),\text{ \ }1\leq l\leq m,
\end{equation}%
where $m$ is the dimension of the eigenspace of $\omega $. Then, according
to the lemma \ref{Lmm5} and the convention of Eq. (\ref{VP1}), we have a
basis for the eigenspace of $-\omega $,%
\begin{equation}
v_{l}(-\omega )=\Sigma _{x}v_{l}^{\ast }(\omega ),\text{ \ }1\leq l\leq m.
\end{equation}%
It is also an orthonormal basis,%
\begin{eqnarray}
v_{l}^{\dag }(-\omega )I_{-}v_{k}(-\omega ) &=&\left[ v_{l}^{\dag }(\omega
)\Sigma _{x}I_{-}\Sigma _{x}v_{k}(\omega )\right] ^{\ast }  \notag \\
&=&-\lambda _{l}\delta _{lk},  \label{PNnorm}
\end{eqnarray}%
where we have used the identity,%
\begin{equation}
\Sigma _{x}I_{-}\Sigma _{x}=-I_{-}.
\end{equation}%
In sum, there are totally $m$ dynamic mode pairs for the eigenenergy pair ($%
\omega ,-\omega $) where $\omega \neq 0$. Equation (\ref{PNnorm}) shows that
each mode pair has two linearly independent eigenvectors with opposite norms.

As pointed out after the lemma \ref{Lmm7}, the two conventions of Eqs. (\ref%
{VP1}) and (\ref{VP2}) must be obeyed in constructing a derivative BV
transformation. The discussions above demonstrate that for a pair of nonzero
eigenvalues, the orthonormaliztion is compatible with the convention of Eq. (%
\ref{VP1}). If zero is an eigenvalue of the dynamic matrix, is the
orthonormaliztion compatible with the convention of Eq. (\ref{VP2})? i.e.,
does there exist such a basis for the eigenspace of zero eigenvalue that is
orthonormal and satisfies the convention of Eq. (\ref{VP2}) simultaneously?
The answer is yes.

\begin{lemma}
\label{ZeroSpace} If the dynamic matrix $D$ is physically diagonalizable and
has zero eigenvalue, there exists such an orthonormal basis for the
eigenspace of zero eigenvalue that meets the requirement of Eq. (\ref{VP2}).
\end{lemma}

\begin{proof}
According to the lemma \ref{Lmm6}, there always exists such a basis for the
eigenspace $V_{0}$ of zero eigenvalue that satisfies the requirement of Eq. (%
\ref{VP2}),
\begin{equation}
v_{m+k}(0)=\Sigma_{x}v_{k}^{\ast}(0),\text{ \ }k=1,2,\cdots,m,
\end{equation}
where $2m$ ($m\in%
\mathbb{N}
$) is the dimension of $V_{0}$, i.e., $\dim(V_{0})=2m$.

When $m=1$,
\begin{equation}
v_{1}^{\dag }(0)I_{-}v_{1}(0)\neq 0,
\end{equation}%
i.e., $v_{1}(0)$ can not be isotropic. Otherwise, one has%
\begin{equation}
v_{1}^{\dag }(0)I_{-}v_{1}(0)=0,\text{ \ }v_{2}^{\dag }(0)I_{-}v_{2}(0)=0.
\label{VV001}
\end{equation}%
In addition,%
\begin{eqnarray}
v_{1}^{\dag }(0)I_{-}v_{2}(0) &=&v_{1}^{\dag }(0)I_{-}\Sigma _{x}v_{1}^{\ast
}(0)  \notag \\
&=&\widetilde{v_{1}^{\ast }}(0)Jv_{1}^{\ast }(0)  \notag \\
&=&0,
\end{eqnarray}%
where $J$ is the unit symplectic matrix,%
\begin{equation}
J=\left[
\begin{array}{cc}
0 & I \\
-I & 0%
\end{array}%
\right] .  \label{JSy}
\end{equation}%
That is to say, the two eigenvectors $v_{1}(0)$ and $v_{2}(0)$ are
orthogonal to each other,%
\begin{equation}
v_{1}^{\dag }(0)I_{-}v_{2}(0)=0,\text{ \ }v_{2}^{\dag }(0)I_{-}v_{1}(0)=0.
\label{VV002}
\end{equation}%
Equations (\ref{VV001}) and (\ref{VV002}) show that the matrix of the
sesquilinear form $I_{-}$ vanishes identically on $V_{0}$. That is in
contradiction with the fact that $I_{-}$ is a nonsingular metric on $V_{0}$.
In a word, the norm of $v_{1}(0)$ can not vanish. So we can normalize it,%
\begin{equation}
v_{1}^{\dag }(0)I_{-}v_{1}(0)=1\text{ or }-1.
\end{equation}%
Accordingly,
\begin{equation}
v_{2}^{\dag }(0)I_{-}v_{2}(0)=-1\text{ or }1.
\end{equation}%
Together with Eq. (\ref{VV002}), one has%
\begin{equation}
v_{i}^{\dag }(0)I_{-}v_{j}(0)=-\lambda _{i}\delta _{ij},\text{ \ }\lambda
_{i}=\pm 1,\text{ }i,j=1,2.
\end{equation}%
This implies that the lemma holds when $m=1$.

Suppose that the lemma holds when $m=l$ ($l\in%
\mathbb{N}
$), we shall show that it also holds when $m=l+1$.

If there exists at least one of the eigenvectors that is nonisotropic,
assume without loss of generality that $v_{1}(0)$ is such a vector, i.e.,%
\begin{equation}
v_{1}^{\dag }(0)I_{-}v_{1}(0)\neq 0,
\end{equation}%
then we have%
\begin{equation}
v_{l+2}^{\dag }(0)I_{-}v_{l+2}(0)\neq 0,
\end{equation}%
namely, $v_{l+2}(0)$ is nonisotropic, too. That is because%
\begin{equation}
v_{l+2}(0)=\Sigma _{x}v_{1}^{\ast }(0).  \label{ZV1}
\end{equation}%
Consider the two-dimensional subspace $W$ spanned by the linearly
independent set $\{v_{1}(0),v_{l+2}(0)\}$, i.e.,
\begin{equation}
W=\mathrm{span}(v_{1}(0),v_{l+2}(0)).
\end{equation}%
Analogous to the case of $m=1$, we can obtain such an orthonormal basis for $%
W$,%
\begin{equation}
v_{i}^{\dag }(0)I_{-}v_{j}(0)=-\lambda _{i}\delta _{ij},  \label{ZV2}
\end{equation}%
where $\lambda _{i}=\pm 1$, and $i,j=1,l+2$. It satisfies the requirement of
Eq. (\ref{VP2}).

If all the eigenvectors are isotropic, there must exist at least two
eigenvectors such that their inner product is nonvanishing. Otherwise, the
form $I_{-}$ will vanish identically on $V_{0}$,%
\begin{equation}
v_{i}^{\dag }(0)I_{-}v_{j}(0)=0,\text{ \ }\forall i,j\in \{1,2,\cdots
,2l+2\},
\end{equation}%
which is obviously impossible because the form $I_{-}$ is nonsingular on $%
V_{0}$. Without loss of generality, let us suppose that
\begin{equation}
v_{1}^{\dag }(0)I_{-}v_{2}(0)\neq 0.  \label{NonReal0}
\end{equation}%
We can always adjust the phase of $v_{1}(0)$ or $v_{2}(0)$ so that $%
v_{1}^{\dag }(0)I_{-}v_{2}(0)$ is purely imaginary,%
\begin{equation}
v_{1}^{\dag }(0)I_{-}v_{2}(0)\notin
\mathbb{R}
.  \label{NonReal}
\end{equation}%
Now, set%
\begin{eqnarray}
w_{1} &=&v_{1}(0)+iv_{2}(0), \\
w_{2} &=&v_{1}(0)-iv_{2}(0),
\end{eqnarray}%
and%
\begin{eqnarray}
w_{l+2} &=&\Sigma _{x}w_{1}^{\ast }=v_{l+2}(0)-iv_{l+3}(0), \\
w_{l+3} &=&\Sigma _{x}w_{2}^{\ast }=v_{l+2}(0)+iv_{l+3}(0).
\end{eqnarray}%
From Eqs. (\ref{NonReal0}) and (\ref{NonReal}), we obtain%
\begin{eqnarray}
w_{1}^{\dag }I_{-}w_{1} &=&i\left[ v_{1}^{\dag }(0)I_{-}v_{2}(0)-v_{2}^{\dag
}(0)I_{-}v_{1}(0)\right]  \notag \\
&\neq &0.
\end{eqnarray}%
Obviously,
\begin{equation}
w_{1}(0),\text{ }w_{2}(0),\text{ }w_{l+2}(0),\text{ and }w_{l+3}(0)\in V_{0}.
\end{equation}%
Besides, they are linearly independent. For convenience, let us reset%
\begin{eqnarray}
v_{1}(0) &=&w_{1},\text{ \ }v_{2}(0)=w_{2}, \\
v_{l+2}(0) &=&w_{l+2},\text{ \ }v_{l+3}(0)=w_{l+3},
\end{eqnarray}%
and consider the new set,
\begin{widetext}%
\begin{equation}
\{v_{1}(0),v_{2}(0),\cdots ,v_{l+1}(0),v_{l+2}(0),v_{l+3}(0),\cdots
,v_{2l+2}(0)\}.
\end{equation}%
\end{widetext}
It is evident that this set forms a new basis for $V_{0}$, and satisfies the
requirement of Eq. (\ref{VP2}). In particular, $v_{1}(0)$ is nonisotropic,
\begin{equation}
v_{1}^{\dag }(0)I_{-}v_{1}(0)\neq 0.
\end{equation}%
Therefore, the new basis returns to the case discussed just above. All in
all, we can always obtain a two-dimensional subspace $W$ as given in Eqs. (%
\ref{ZV1})--(\ref{ZV2}) whether the basis vectors of $V_{0}$ are isotropic
or not.

Using the two basis vectors of $W$, we can put%
\begin{eqnarray}
\xi _{i}(0) &=&v_{i+1}(0)-\frac{v_{1}^{\dag }(0)I_{-}v_{i+1}(0)}{v_{1}^{\dag
}(0)I_{-}v_{1}(0)}v_{1}(0)  \notag \\
&&-\frac{v_{l+2}^{\dag }(0)I_{-}v_{i+1}(0)}{v_{l+2}^{\dag }(0)I_{-}v_{l+2}(0)%
}v_{l+2}(0), \\
\xi _{l+i}(0) &=&v_{l+i+2}(0)-\frac{v_{1}^{\dag }(0)I_{-}v_{l+i+2}(0)}{%
v_{1}^{\dag }(0)I_{-}v_{1}(0)}v_{1}(0)  \notag \\
&&-\frac{v_{l+2}^{\dag }(0)I_{-}v_{l+i+2}(0)}{v_{l+2}^{\dag
}(0)I_{-}v_{l+2}(0)}v_{l+2}(0),
\end{eqnarray}%
where $i=1,2,\cdots ,l$. They are all orthogonal to $v_{1}(0)$ and $%
v_{l+2}(0)$,%
\begin{equation}
v_{i}^{\dag }(0)I_{-}\xi _{j}(0)=0,\text{ \ }i=1,l+2;\text{ }j=1,2,\cdots
,2l.
\end{equation}%
Evidently, all those vectors $\xi _{i}(0)$ are still linearly independent
and eigenvectors of zero eigenvalue, i.e.,
\begin{equation}
\xi _{i}(0)\in V_{0},\text{ \ }i=1,2,\cdots ,2l.
\end{equation}%
Now, consider the space $W^{\prime }$,
\begin{equation}
W^{\prime }=\mathrm{span}(\{\xi _{i}(0)|\text{ }i=1,2,\cdots ,2l\}).
\end{equation}%
It is a proper subspace of $V_{0}$, $\dim (W^{\prime })=2l$. Obviously,
\begin{equation}
V_{0}=W\oplus W^{\prime }.
\end{equation}%
As shown above, $v_{1}(0)$ and $v_{l+2}(0)$ are both orthogonal to the set $%
\{\xi _{i}(0)|$ $i=1,2,\cdots ,2l\}$, therefore,
\begin{equation}
W^{\prime }=W^{\perp }.
\end{equation}%
This implies that $V_{0}$ is the orthogonal direct sum of $W$ and $W^{\prime
}$,%
\begin{equation}
V_{0}=W\odot W^{\prime }.  \label{WPW}
\end{equation}

It is easy to show that%
\begin{gather}
\left[ v_{1}^{\dag }(0)I_{-}v_{i+1}(0)\right] ^{\ast }=-v_{l+2}^{\dag
}(0)I_{-}v_{l+i+2}(0), \\
\left[ v_{l+2}^{\dag }(0)I_{-}v_{i+1}(0)\right] ^{\ast }=-v_{1}^{\dag
}(0)I_{-}v_{l+i+2}(0).
\end{gather}%
As a result, we obtain%
\begin{equation}
\xi _{l+i}(0)=\Sigma _{x}\xi _{i}^{\ast }(0),\text{ \ }i=1,2,\cdots ,l.
\end{equation}%
This indicates that the basis for $W^{\prime }$ is exacly in accordance with
the convention of Eq. (\ref{VP2}). Since $\dim (W^{\prime })=2l$, by the
induction hypothesis, the space $W^{\prime }$ has an orthonormal basis that
satisfies the convention of Eq. (\ref{VP2}). Suppose the basis is the set:
\begin{equation}
\{\zeta _{1}(0),\zeta _{2}(0),\cdots ,\zeta _{2l}(0)\},
\end{equation}%
which satisfies%
\begin{equation}
\zeta _{l+i}(0)=\Sigma _{x}\zeta _{i}^{\ast }(0),  \label{ZV3}
\end{equation}%
where $i=1,2,\cdots ,l$, and%
\begin{equation}
\zeta _{i}^{\dag }(0)I_{-}\zeta _{j}(0)=-\lambda _{i}\delta _{ij},
\label{ZV4}
\end{equation}%
where $\lambda _{i}=\pm 1$ and $i,j=1,2,\cdots ,2l$. It is evident that the
set
\begin{equation}
\{v_{1}(0),v_{l+2}(0),\zeta _{1}(0),\zeta _{2}(0),\cdots ,\zeta _{2l}(0)\}
\end{equation}%
is a basis for $V_{0}$. Upon ordering them as follows,%
\begin{eqnarray}
\eta _{1}(0) &=&v_{1}(0), \\
\eta _{l+2}(0) &=&v_{l+2}(0), \\
\eta _{i+1}(0) &=&\zeta _{i}(0),\text{ \ }i=1,2,\cdots ,l, \\
\eta _{l+i+2}(0) &=&\zeta _{l+i}(0),\text{ \ }i=1,2,\cdots ,l,
\end{eqnarray}%
one has%
\begin{equation}
\eta _{l+i}(0)=\Sigma _{x}\eta _{i}^{\ast }(0),
\end{equation}%
where $i=1,2,\cdots ,l+1$, and%
\begin{equation}
\eta _{i}^{\dag }(0)I_{-}\eta _{j}(0)=-\lambda _{i}\delta _{ij},
\end{equation}%
where $\lambda _{i}=\pm 1$ and $i,j=1,2,\cdots ,2(l+1)$. Here Eqs. (\ref{ZV1}%
), (\ref{ZV2}), (\ref{WPW}), (\ref{ZV3}), and (\ref{ZV4}) have been used.
The two equations above show that the new set
\begin{equation}
\{\eta _{1}(0),\eta _{2}(0),\cdots ,\eta _{2(l+1)}(0)\}
\end{equation}%
is an orthonormal basis for $V_{0}$, and it satisfies the requirement of Eq.
(\ref{VP2}). In other words, the lemma holds for $m=l+1$.

Finally, by mathematical induction, the lemma is valid for any $m\in%
\mathbb{N}
$.
\end{proof}

In this proof, we present a modified version of the Gram-Schmidt
orthogonalization process, which can maintain the satisfiability of the
convention of Eq. (\ref{VP2}). Put it another way, if one starts form a
basis satisfying Eq. (\ref{VP2}), he can arrive finally at an orthonormal
basis which will still satisfy Eq. (\ref{VP2}) through this modified
Gram-Schmidt orthogonalization process.

This lemma shows that there are totally $m$ mode pairs for the zero
eigenvalue $\omega=0$, each mode pair has two linearly independent
eigenvectors with opposite norms.

The total $n$ dynamic mode pairs are thus selected, their eigenvectors form
an orthonormal basis for the whole space of $%
\mathbb{C}
^{2n}$,
\begin{equation}
v^{\dag }(\omega _{i})I_{-}v(\omega _{j})=\lambda _{i}\delta _{ij},
\end{equation}%
where $\lambda _{i}=\pm 1$ and $1\leq i,j\leq 2n$. Here each eigenvalue is
counted up to its multiplicity. Each dynamic mode pair has two linearly
independent eigenvectors with both opposite eigenenergies and opposite
norms. Therefore, a half of the basis vectors have the norm $1$, the other
half have the norm $-1$.

From Eq. (\ref{vMvNew}) and the equation above, we obtain the following
lemma.

\begin{lemma}
\label{BasisC} With the orthonormal basis of $%
\mathbb{C}
^{2n}$ chosen as above, the coefficient matrix $M$ is diagonalized,%
\begin{equation}
v^{\dag}(\omega_{i})Mv(\omega_{j})=\lambda_{i}\omega_{i}\delta_{ij},\text{ \
}1\leq i,j\leq2n,  \label{vMv1}
\end{equation}
where $v^{\dag}(\omega_{i})$ is an eigenvector of the eigenvalue $\omega_{i}$%
, with $\lambda_{i}=1$ or $-1$ being the corresponding norm.
\end{lemma}

For each mode pair, its two eigenvectors have opposite norms, one is $+1$,
the other is $-1$. So we can stipulate an order for every mode pair: The
first eigenvector has the norm of $+1$, and the second one has the norm of $%
-1$. Under this stipulation, the derivative BV transformation of Eqs. (\ref%
{PsiToPhi}) and (\ref{TPair}) becomes
\begin{widetext}%
\begin{gather}
\psi =T_{n}\varphi ,  \label{TN} \\
T_{n}=\left[
\begin{array}{cccccccc}
v(\omega _{1}), & v(\omega _{2}), & \cdots , & v(\omega _{n}), & v(-\omega
_{1}), & v(-\omega _{2}), & \cdots , & v(-\omega _{n})%
\end{array}%
\right] ,  \label{DBVTn}
\end{gather}%
\end{widetext}
where, for each mode pair of $(v(\omega _{i}),v(-\omega _{i}))$, the
eigenvectors are ordered as follows,
\begin{eqnarray}
v^{\dag }(\omega _{i})I_{-}v(\omega _{i}) &=&1,  \label{TN1} \\
v^{\dag }(-\omega _{i})I_{-}v(-\omega _{i}) &=&-1.  \label{TN2}
\end{eqnarray}%
That is to say, the left half of $T_{n}$ is filled with the eigenvectors
with the positive norms of $+1$; the right half of $T_{n}$ is filled with
the eigenvectors with the negative norms of $-1$. For convenience, we shall
call $T_{n}$ the normal derivative BV matrix and Eq. (\ref{TN}) the normal
derivative BV transformation. To sum up, a normal derivative BV
transformation can always be generated by the Heisenberg equation of motion
if the dynamic matrix of the system is physically diagonalizable.

According to the stipulation for the normal BV matrix $T_{n}$, one has%
\begin{equation}
v^{\dag }(\omega _{i})I_{-}v(\omega _{j})=\lambda _{i}\delta _{ij},\text{ \ }%
1\leq i,j\leq 2n,
\end{equation}%
where%
\begin{equation}
\lambda _{i}=\left\{
\begin{array}{ll}
1, & 1\leq i\leq n \\
-1, & n+1\leq i\leq 2n.%
\end{array}%
\right.
\end{equation}%
In terms of matrix, it can be expressed as%
\begin{equation}
T_{n}^{\dag }I_{-}T_{n}=I_{-}.
\end{equation}%
Thereby, we arrive at the following lemma.

\begin{lemma}
\label{BVTn} If the dynamic matrix $D$ is physically diagonalizable, the
normal derivative BV matrix $T_{n}$ satisfies the identity,
\begin{equation}
T_{n}^{\dag }I_{-}T_{n}=I_{-}.
\end{equation}
\end{lemma}

In other words, $T_{n}$ is a member of the group $U(n,n)$ \cite{Chen}. The
sesquilinear form $I_{-}$ remains invariant under the transformation of $%
T_{n}$.

This lemma implies that the new field $\varphi $ is a standard bosonic
field, i.e.,%
\begin{equation}
\varphi \cdot \varphi ^{\dag }=I_{-}.
\end{equation}%
Put it another way, the new component operators, $d_{i}$ and $d_{i}^{\dag }$
($i=1,2,\cdots ,n$) will satisfy the standard commutation rules for the
annihilation and creation operators of bosons,%
\begin{equation}
\lbrack d_{i},d_{j}^{\dag }]=\delta _{i,j},\text{ \ }[d_{i},d_{j}]=0,\text{
\ }[d_{i}^{\dag },d_{j}^{\dag }]=0.  \label{DComm}
\end{equation}

\begin{lemma}
\label{MDiag} If the dynamic matrix $D$ is physically diagonalizable, the
normal derivative BV matrix $T_{n}$ will diagonalize the coefficient matrix $%
M$ in the manner of Hermitian congruence. That is
\begin{equation}
T_{n}^{\dag }MT_{n}=\mathrm{diag}(\omega _{1},\cdots ,\omega _{n},\omega
_{1},\cdots ,\omega _{n}).
\end{equation}
\end{lemma}

\begin{proof}
It comes simply from the lemma \ref{BVTn}\ and Eq. (\ref{vMv1}).
\end{proof}

This lemma shows that the dynamic matrix $D$ and the coefficient matrix $M$
can be diagonalized simultaneously if $D$ is physically diagonalizable:%
\begin{gather}
T_{n}^{-1}DT_{n}=\mathrm{diag}(\omega _{1},\cdots ,\omega _{n},-\omega
_{1},\cdots ,-\omega _{n}), \\
T_{n}^{\dag }MT_{n}=\mathrm{diag}(\omega _{1},\cdots ,\omega _{n},\omega
_{1},\cdots ,\omega _{n}).
\end{gather}%
It is evident that the manners of diagonalization are different: The former
is diagonalized by a similar transformation, and the latter by a Hermitian
congruence transformation. In brief, the two different matrices have been
diagonalized by the two different manners of transformation simultaneously,
that solves the key problem occurring in the Bose system.

With the two lemmas above, we obtain the sufficient condition for the
diagonalization of the Bose system.

\begin{proposition}
\label{PPS2} A quadratic Hamiltonian of bosons is BV diagonalizable if its
dynamic matrix is physically diagonalizable.
\end{proposition}

\begin{proof}
By replacing $\psi$ with the $\varphi$ of Eq. (\ref{TN}), Eq. (\ref{Ham2})
can be reformulated as%
\begin{equation}
H=\frac{1}{2}\varphi^{\dag}T_{n}^{\dag}MT_{n}\varphi-\frac{1}{2}\mathrm{tr}%
(\alpha).
\end{equation}
Using the lemma \ref{MDiag}, we have
\begin{equation}
H=\sum_{i=1}^{n}\omega_{i}d_{i}^{\dag}d_{i}+\frac{1}{2}\sum_{i=1}^{n}%
\omega_{i}-\frac{1}{2}\mathrm{tr}(\alpha),  \label{B4}
\end{equation}
where, as shown in Eq. (\ref{DComm}), the $d_{i}$ and $d_{i}^{\dag}$ ($%
i=1,2,\cdots,n$) are the new annihilation and creation operators for the
bosons. Equation (\ref{B4}) shows that a quadratic Hamiltonian of bosons can
be BV diagonalized if its dynamic matrix is physically diagonalizable.
\end{proof}

So far, we have found that the Heisenberg equation of motion is a natural
generator of the normal BV transformation. This transformation can react on
the Hamiltonian itself and brings it into a diagonalized form automatically.
That is just what we expected.

\begin{corollary}
\label{Corollary2} The quadratic Hamiltonian of bosons of Eq. (\ref{Ham1})
is BV diagonalizable if the coefficient submatrix $\gamma$ vanishes
identically, i.e., $\gamma=0$.
\end{corollary}

\begin{proof}
When $\gamma=0$, the dynamic matrix of Eq. (\ref{DM}) reduces to
\begin{equation}
D=\left[
\begin{array}{cc}
\alpha & 0 \\
0 & -\widetilde{\alpha}%
\end{array}
\right] .
\end{equation}
Since $\alpha$ is Hermitian, the $\widetilde{\alpha}$ is Hermitian, $%
\widetilde{\alpha}^{\dag}=\widetilde{\alpha}$, too. As a result, $D=D^{\dag}$%
. It implies that $D$ is unitarily diagonalizable and all its eigenvalues
are real when $\gamma=0$. Of course, the dynamic matrix $D$ is BV
diagonalizable if $\gamma=0$, which proves the corollary.
\end{proof}

On one hand, the corollary \ref{Corollary2} shows that some quadratic
Hamiltonians of bosons are BV diagonalizable. On the other hand, the
corollary \ref{Corollary1} indicates that there are also some quadratic
Hamiltonians of bosons that can not be BV diagonalized. What is
simultaneously the necessary and sufficient condition for the BV
diagonalization? Evidently, the answer is just the combination of the two
prepositions \ref{PPS1} and \ref{PPS2}.

\begin{theorem}
\label{BThm} A quadratic Hamiltonian of bosons is BV diagonalizable if and
only if its dynamic matrix is physically diagonalizable.
\end{theorem}

This theorem is obviously consistent with our physical intuition: A system
will behavior as a collective of quasi-particles if and only if there exists
a complete set of linearly independent normal modes of motion in the system.

In particular, it converts the BV diagonalization into the eigenvalue
problem of the dynamic matrix, whose theory is very clear and simple
mathematically, and familiar to all of us. Therefore, this theorem makes it
easy for us to find BV transformation and realize BV diagonalization.

Now, let us return to the example \ref{Example1}. If $\left\vert
\alpha\right\vert <\left\vert \gamma\right\vert $, the dynamic matrix has
two imaginary eigenvalues; the Hamiltonian is not BV diagonalizable. If $%
\left\vert \alpha\right\vert =\left\vert \gamma\right\vert $, the dynamic
matrix is itself not diagonalizable; the Hamiltonian is not BV
diagonalizable. If $\left\vert \alpha\right\vert >\left\vert
\gamma\right\vert $, the dynamic matrix has two real eigenvalues, it is
hence physically diagonalizable, therefore, the Hamiltonian is BV
diagonalizable. To sum up, the Hamiltonian is BV diagonalizable only when $%
\left\vert \alpha\right\vert >\left\vert \gamma\right\vert $.

Further, let us find out the normal derivative BV transformation for the
example \ref{Example1} in the case of $\left\vert \alpha \right\vert
>\left\vert \gamma \right\vert $.

Set%
\begin{equation}
\omega =\sqrt{\alpha ^{2}-\left\vert \gamma \right\vert ^{2}}.
\end{equation}%
The normalized eigenvector for $\omega $ can be obtained from Eq. (\ref%
{Exm1EgEv}),
\begin{equation}
v(\omega )=\left\{
\begin{array}{ll}
\frac{1}{\sqrt{2\omega \left( \alpha -\omega \right) }}\left[
\begin{array}{c}
\gamma \\
\omega -\alpha%
\end{array}%
\right] , & \alpha >0 \\
\frac{1}{\sqrt{2\omega \left( \omega -\alpha \right) }}\left[
\begin{array}{c}
\gamma \\
\omega -\alpha%
\end{array}%
\right] , & \alpha <0,%
\end{array}%
\right.
\end{equation}%
with the norm being
\begin{equation}
v^{\ast }(\omega )I_{-}v(\omega )=\left\{
\begin{array}{ll}
1, & \alpha >0 \\
-1, & \alpha <0.%
\end{array}%
\right.
\end{equation}%
According to Eq. (\ref{DBVTn}), the normal BV matrix is%
\begin{equation}
T_{n}=\left\{
\begin{array}{ll}
\lbrack v(\omega ),v(-\omega )], & \alpha >0 \\
\lbrack v(-\omega ),v(\omega )], & \alpha <0,%
\end{array}%
\right.
\end{equation}%
where the normalized eigenvector for $-\omega $ can be given according to
the convention of Eq. (\ref{VP1}). In detail, it reads,
\begin{equation}
T_{n}=\left\{
\begin{array}{cc}
\frac{1}{\sqrt{2\omega \left( \alpha -\omega \right) }}\left[
\begin{array}{cc}
\gamma & \omega -\alpha \\
\omega -\alpha & \gamma ^{\ast }%
\end{array}%
\right] , & \alpha >0 \\
\frac{1}{\sqrt{2\omega \left( \omega -\alpha \right) }}\left[
\begin{array}{cc}
\omega -\alpha & \gamma \\
\gamma ^{\ast } & \omega -\alpha%
\end{array}%
\right] , & \alpha <0.%
\end{array}%
\right.
\end{equation}

The following results can be readily verified,
\begin{equation}
T_{n}^{\dag }I_{-}T_{n}=I_{-},
\end{equation}%
\begin{equation}
T_{n}^{-1}DT_{n}=\left\{
\begin{array}{ll}
\mathrm{diag}(\omega ,-\omega ), & \alpha >0 \\
\mathrm{diag}(-\omega ,\omega ), & \alpha <0,%
\end{array}%
\right.
\end{equation}%
\begin{equation}
T_{n}^{\dag }MT_{n}=\left\{
\begin{array}{ll}
\mathrm{diag}(\omega ,\omega ), & \alpha >0 \\
\mathrm{diag}(-\omega ,-\omega ), & \alpha <0,%
\end{array}%
\right.
\end{equation}%
\begin{equation}
H=\left\{
\begin{array}{ll}
\omega d^{\dag }d+\frac{1}{2}\omega -\frac{1}{2}\alpha , & \alpha >0 \\
-\omega d^{\dag }d-\frac{1}{2}\omega -\frac{1}{2}\alpha , & \alpha <0,%
\end{array}%
\right.
\end{equation}%
where
\begin{equation}
\lbrack d,d^{\dag }]=1.
\end{equation}

In this example, the eigenvalues are not degenerate. Let us look at a
degenerate case.

\begin{example}
\label{Ins2}%
\begin{equation}
H=\alpha (c_{1}^{\dag }c_{1}+c_{2}^{\dag }c_{2})+\gamma c_{1}^{\dag
}c_{2}^{\dag }+\gamma ^{\ast }c_{1}c_{2}.  \label{Exmp3}
\end{equation}
\end{example}

\begin{solution}
The dynamic matrix $D$ is a $4\times 4$ matrix,%
\begin{equation}
D=\left[
\begin{array}{cccc}
\alpha & 0 & 0 & \gamma \\
0 & \alpha & \gamma & 0 \\
0 & -\gamma ^{\ast } & -\alpha & 0 \\
-\gamma ^{\ast } & 0 & 0 & -\alpha%
\end{array}%
\right] .
\end{equation}%
The characteristic equation is
\begin{equation}
\left( \omega ^{2}-\alpha ^{2}+\left\vert \gamma \right\vert ^{2}\right)
^{2}=0.
\end{equation}%
The solutions are
\begin{equation}
\omega =\left\{
\begin{array}{ll}
\pm \sqrt{\alpha ^{2}-\left\vert \gamma \right\vert ^{2}}, & \left\vert
\alpha \right\vert >\left\vert \gamma \right\vert \\
0,\text{ } & \left\vert \alpha \right\vert =\left\vert \gamma \right\vert \\
\pm i\sqrt{\left\vert \gamma \right\vert ^{2}-\alpha ^{2}}, & \left\vert
\alpha \right\vert <\left\vert \gamma \right\vert .%
\end{array}%
\right.
\end{equation}

If $\left\vert \alpha\right\vert <\left\vert \gamma\right\vert $, the
dynamic matrix has imaginary eigenvalues. Of course, the Hamiltonian is not
BV diagonalizable.

It can be readily verified that $D$ has only two linearly independent
eigenvectors if $\left\vert \alpha\right\vert =\left\vert \gamma\right\vert $%
,
\begin{equation}
v_{1}(0)=\left[
\begin{array}{c}
1 \\
0 \\
0 \\
\mp\mathrm{e}^{-i\theta}%
\end{array}
\right] ,\text{ \ }v_{2}(0)=\left[
\begin{array}{c}
0 \\
\mp\mathrm{e}^{i\theta} \\
1 \\
0%
\end{array}
\right] ,
\end{equation}
where $\theta=\arg(\gamma)$, and $\mp$ correspond to $\alpha=\pm\left\vert
\gamma\right\vert $ respectively. It means that $D$ is itself not
diagonalizable. Needless to say, the Hamiltonian is not BV diagonalizable
when $\left\vert \alpha\right\vert =\left\vert \gamma\right\vert $.

If $\left\vert \alpha \right\vert >\left\vert \gamma \right\vert $, there
are a pair of real eigenvalues, i.e., $(\omega ,-\omega )$ where%
\begin{equation}
\omega =\sqrt{\alpha ^{2}-\left\vert \gamma \right\vert ^{2}}.
\end{equation}%
It is found that $\omega $ has two linearly independent eigenvectors, they
can be chosen and orthonormalized as follows,%
\begin{equation}
v_{1}(\omega )=\left\{
\begin{array}{cc}
\left[
\begin{array}{c}
\frac{\gamma }{\sqrt{2\omega \left( \alpha -\omega \right) }} \\
0 \\
0 \\
\frac{\omega -\alpha }{\sqrt{2\omega \left( \alpha -\omega \right) }}%
\end{array}%
\right] , & \alpha >0 \\
\left[
\begin{array}{c}
\frac{\gamma }{\sqrt{2\omega \left( \omega -\alpha \right) }} \\
0 \\
0 \\
\frac{\omega -\alpha }{\sqrt{2\omega \left( \omega -\alpha \right) }}%
\end{array}%
\right] , & \alpha <0,%
\end{array}%
\right.
\end{equation}%
\begin{equation}
v_{2}(\omega )=\left\{
\begin{array}{cc}
\left[
\begin{array}{c}
0 \\
\frac{\omega +\alpha }{\sqrt{2\omega \left( \omega +\alpha \right) }} \\
\frac{-\gamma ^{\ast }}{\sqrt{2\omega \left( \omega +\alpha \right) }} \\
0%
\end{array}%
\right] , & \alpha >0 \\
\left[
\begin{array}{c}
0 \\
\frac{\omega +\alpha }{\sqrt{-2\omega \left( \omega +\alpha \right) }} \\
\frac{\gamma ^{\ast }}{\sqrt{-2\omega \left( \omega +\alpha \right) }} \\
0%
\end{array}%
\right] , & \alpha <0,%
\end{array}%
\right.
\end{equation}%
their norms are
\begin{equation}
v_{1}^{\ast }(\omega )I_{-}v_{1}(\omega )=\left\{
\begin{array}{ll}
1, & \alpha >0 \\
-1, & \alpha <0,%
\end{array}%
\right.
\end{equation}%
\begin{equation}
v_{2}^{\ast }(\omega )I_{-}v_{2}(\omega )=\left\{
\begin{array}{ll}
1, & \alpha >0 \\
-1, & \alpha <0.%
\end{array}%
\right.
\end{equation}%
The orthonormal eigenvectors for $-\omega $, $v_{1}(-\omega )$ and $%
v_{2}(-\omega )$, can be obtained according to the convention of Eq. (\ref%
{VP1}).

Evidently, each eigenvalue is two-fold degenerate. The four eigenvectors, $%
v_{1}(\omega )$, $v_{2}(\omega )$, $v_{1}(-\omega )$ and $v_{2}(-\omega )$,
form an orthonormal basis for $%
\mathbb{C}
^{4}$, i.e.,
\begin{equation}
T_{n}^{\dag }I_{-}T_{n}=I_{-},
\end{equation}%
where $T_{n}$ is the normal derivative BV matrix,%
\begin{equation}
T_{n}=\left\{
\begin{array}{ll}
\lbrack v_{1}(\omega ),v_{2}(\omega ),v_{1}(-\omega ),v_{2}(-\omega )], &
\alpha >0 \\
\lbrack v_{1}(-\omega ),v_{2}(-\omega ),v_{1}(\omega ),v_{2}(\omega )], &
\alpha <0.%
\end{array}%
\right.  \label{TnExm}
\end{equation}%
Those facts demonstrate that the dynamic matrix $D$ is physically
diagonalizable if $\left\vert \alpha \right\vert >\left\vert \gamma
\right\vert $.

To sum up, the Hamiltonian of Eq. (\ref{Exmp3}) is BV diagonalizable only
when $\left\vert \alpha \right\vert >\left\vert \gamma \right\vert $.

When $\left\vert \alpha \right\vert >\left\vert \gamma \right\vert $, the
following results can be verified straightforwardly,%
\begin{equation}
T_{n}^{-}DT_{n}=\left\{
\begin{array}{ll}
\mathrm{diag}(\omega ,\omega ,-\omega ,-\omega ), & \alpha >0 \\
\mathrm{diag}(-\omega ,-\omega ,\omega ,\omega ), & \alpha <0,%
\end{array}%
\right. ,
\end{equation}%
\begin{equation}
T_{n}^{\dag }MT_{n}=\left\{
\begin{array}{ll}
\mathrm{diag}(\omega ,\omega ,\omega ,\omega ), & \alpha >0 \\
\mathrm{diag}(-\omega ,-\omega ,-\omega ,-\omega ), & \alpha <0,%
\end{array}%
\right.
\end{equation}%
\begin{equation}
H=\left\{
\begin{array}{ll}
\omega (d_{1}^{\dag }d_{1}+d_{2}^{\dag }d_{2})+\omega -\alpha , & \alpha >0
\\
-\omega (d_{1}^{\dag }d_{1}+d_{2}^{\dag }d_{2})-\omega -\alpha , & \alpha <0.%
\end{array}%
\right.
\end{equation}%
where
\begin{equation}
\lbrack d_{i},d_{j}^{\dag }]=\delta _{ij},\text{ \ }[d_{i},d_{j}]=0,\text{ \
}[d_{i}^{\dag },d_{j}^{\dag }]=0.
\end{equation}
\end{solution}

The two examples above do not have zero eigenvalue. The following is an
example with zero eigenvalue.

\begin{example}
\label{Ins3}%
\begin{equation}
H=c_{1}^{\dag }c_{1}+c_{2}^{\dag }c_{2}-c_{1}^{\dag }c_{2}-c_{2}^{\dag
}c_{1}.  \label{Exmp4}
\end{equation}
\end{example}

\begin{solution}
The dynamic matrix $D$ is
\begin{equation}
D=\left[
\begin{array}{cccc}
1 & -1 & 0 & 0 \\
-1 & 1 & 0 & 0 \\
0 & 0 & -1 & 1 \\
0 & 0 & 1 & -1%
\end{array}
\right] .
\end{equation}
It has three eigenvalues,
\begin{equation}
\omega=0,\text{ }2,\text{ }-2.
\end{equation}

For the pair of the nonzero eigenvalues, $(2,-2)$, they are nondegenerate,
\begin{equation}
v(2)=\frac{1}{\sqrt{2}}\left[
\begin{array}{c}
1 \\
-1 \\
0 \\
0%
\end{array}%
\right] ,\text{ \ }v(-2)=\frac{1}{\sqrt{2}}\left[
\begin{array}{c}
0 \\
0 \\
1 \\
-1%
\end{array}%
\right] ,
\end{equation}%
where the convention of Eq. (\ref{VP1}) has been used. Their norms are%
\begin{equation}
v^{\dag }(2)I_{-}v(2)=1,\text{ \ }v^{\dag }(-2)I_{-}v(-2)=-1.
\end{equation}

For the zero eigenvalue, it is two-fold degenerate, the orthonormal basis
for this eigenspace can be chosen as follows,%
\begin{equation}
v(0)=\frac{1}{\sqrt{2}}\left[
\begin{array}{c}
1 \\
1 \\
0 \\
0%
\end{array}%
\right] ,\text{ \ }v(-0)=\frac{1}{\sqrt{2}}\left[
\begin{array}{c}
0 \\
0 \\
1 \\
1%
\end{array}%
\right] ,
\end{equation}%
where the convention of Eq. (\ref{VP2}) has been used. Their norms are%
\begin{equation}
v^{\dag }(0)I_{-}v(0)=1,\text{ \ }v^{\dag }(-0)I_{-}v(-0)=-1.
\end{equation}

Therefore, the dynamic matrix $D$ has three real eigenvalues and four
linearly independent eigenvectors, so it is physically diagonalizable.
According to the theorem \ref{BThm}, the Hamiltonian of Eq. (\ref{Exmp4}) is
BV diagonalizable.

According to Eq. (\ref{DBVTn}), the normal derivative BV matrix $T_{n}$ has
the form,
\begin{equation}
T_{n}=\left[
\begin{array}{cccc}
v(2), & v(0), & v(-2), & v(-0)%
\end{array}%
\right] .
\end{equation}%
It is easy to show that%
\begin{equation}
T_{n}^{\dag }I_{-}T_{n}=I_{-},
\end{equation}%
\begin{equation}
T_{n}^{-}DT_{n}=\mathrm{diag}(2,0,-2,-0),
\end{equation}%
\begin{equation}
T_{n}^{\dag }MT_{n}=\mathrm{diag}(2,0,2,0),
\end{equation}%
\begin{equation}
H=2d_{1}^{\dag }d_{1}+0d_{2}^{\dag }d_{2},
\end{equation}%
where
\begin{equation}
\lbrack d_{i},d_{j}^{\dag }]=\delta _{ij},\text{ \ }[d_{i},d_{j}]=0,\text{ \
}[d_{i}^{\dag },d_{j}^{\dag }]=0.
\end{equation}
\end{solution}

For this example, the eigenspace of the zero eigenvalue is two dimensional,
it is the smallest and nondegenerate. Finally, we give an example whose
eigenspace of the zero eigenvalue is more than two dimensional.

\begin{example}
\label{Ins4}%
\begin{eqnarray}
H &=&2c_{1}^{\dag }c_{1}+c_{2}^{\dag }c_{2}+c_{3}^{\dag }c_{3}+\sqrt{2}%
(c_{1}^{\dag }c_{2}+c_{2}^{\dag }c_{1})  \notag \\
&&+\sqrt{2}(c_{1}^{\dag }c_{3}+c_{3}^{\dag }c_{1})+(c_{2}^{\dag
}c_{3}+c_{3}^{\dag }c_{2}).
\end{eqnarray}
\end{example}

\begin{solution}
The dynamic matrix $D$ is%
\begin{equation}
D=\left[
\begin{array}{cccccc}
2 & \sqrt{2} & \sqrt{2} & 0 & 0 & 0 \\
\sqrt{2} & 1 & 1 & 0 & 0 & 0 \\
\sqrt{2} & 1 & 1 & 0 & 0 & 0 \\
0 & 0 & 0 & -2 & -\sqrt{2} & -\sqrt{2} \\
0 & 0 & 0 & -\sqrt{2} & -1 & -1 \\
0 & 0 & 0 & -\sqrt{2} & -1 & -1%
\end{array}%
\right] .
\end{equation}%
There are three eigenvalues,
\begin{equation}
\omega =0,\text{ }4,\text{ }-4.
\end{equation}%
For the pair of the nonzero eigenvalues, $(4,-4)$, they are nondegenerate,
\begin{equation}
v(4)=\left[
\begin{array}{c}
\frac{\sqrt{2}}{2} \\
\frac{1}{2} \\
\frac{1}{2} \\
0 \\
0 \\
0%
\end{array}%
\right] ,\text{ \ }v(-4)=\left[
\begin{array}{c}
0 \\
0 \\
0 \\
\frac{\sqrt{2}}{2} \\
\frac{1}{2} \\
\frac{1}{2}%
\end{array}%
\right] ,
\end{equation}%
with the norms as follows%
\begin{equation}
v^{\dag }(4)I_{-}v(4)=1,\text{ \ }v^{\dag }(-4)I_{-}v(-4)=-1.
\end{equation}%
For the zero eigenvalue, it is four-fold degenerate, the orthonormal basis
for the eigenspace can be chosen as follows,%
\begin{equation}
v_{1}(0)=\left[
\begin{array}{c}
\sqrt{\frac{1}{3}} \\
-\sqrt{\frac{2}{3}} \\
0 \\
0 \\
0 \\
0%
\end{array}%
\right] ,\text{ \ }v_{2}(0)=\left[
\begin{array}{c}
\frac{1}{\sqrt{6}} \\
\frac{1}{2\sqrt{3}} \\
-\frac{\sqrt{3}}{2} \\
0 \\
0 \\
0%
\end{array}%
\right] ,
\end{equation}%
\begin{equation}
v_{1}(-0)=\left[
\begin{array}{c}
0 \\
0 \\
0 \\
\sqrt{\frac{1}{3}} \\
-\sqrt{\frac{2}{3}} \\
0%
\end{array}%
\right] ,\text{ \ }v_{2}(-0)=\left[
\begin{array}{c}
0 \\
0 \\
0 \\
\frac{1}{\sqrt{6}} \\
\frac{1}{2\sqrt{3}} \\
-\frac{\sqrt{3}}{2}%
\end{array}%
\right] ,
\end{equation}%
their norms are%
\begin{eqnarray}
v_{1}^{\dag }(0)I_{-}v_{1}(0) &=&1, \\
v_{2}^{\dag }(0)I_{-}v_{2}(0) &=&1, \\
v_{1}^{\dag }(-0)I_{-}v_{1}(-0) &=&-1, \\
v_{2}^{\dag }(-0)I_{-}v_{2}(-0) &=&-1.
\end{eqnarray}

In a word, the dynamic matrix $D$ has three real eigenvalues and six
linearly independent eigenvectors. The Hamiltonian is BV diagonalizable.

The normal derivative BV matrix $T_{n}$ can be constructed according to Eq. (%
\ref{DBVTn}),
\begin{equation}
T_{n}=\left[ v(4),v_{1}(0),v_{2}(0),v(-4),v_{1}(-0),v_{2}(-0)\right] .
\end{equation}%
It can be verified straightforwardly that%
\begin{equation}
T_{n}^{\dag }I_{-}T_{n}=I_{-},
\end{equation}%
\begin{equation}
T_{n}^{-}DT_{n}=\mathrm{diag}(4,0,0,-4,-0,-0),
\end{equation}%
\begin{equation}
T_{n}^{\dag }MT_{n}=\mathrm{diag}(4,0,0,4,0,0),
\end{equation}%
\begin{equation}
H=4d_{1}^{\dag }d_{1}+0d_{2}^{\dag }d_{2}+0d_{3}^{\dag }d_{3},
\end{equation}%
where
\begin{equation}
\lbrack d_{i},d_{j}^{\dag }]=\delta _{ij},\text{ \ }[d_{i},d_{j}]=0,\text{ \
}[d_{i}^{\dag },d_{j}^{\dag }]=0.
\end{equation}
\end{solution}

\subsection{Uniqueness}

The theorem \ref{BThm} asserts that a quadratic Hamiltonian of bosons can be
BV diagonalized if its dynamic matrix is physically diagonalizable.
Apparently, there may exist many different BV transformations for a certain
Hamiltonian that can all realize the diagonalization. This occurs especially
when some eigenvalues of the dynamic matrix are degenerate because there are
much orthonormal bases for the eigenspace of a degenerate eigenvalue, and
accordingly there are much various choices for the column vectors of the
normal derivative BV matrix. For instance, it is easy to show that, when $%
\alpha >0$, the following two vectors,%
\begin{equation}
v_{1}(\omega )=\left[
\begin{array}{c}
\frac{\gamma }{2\sqrt{2\omega \left( \alpha -\omega \right) }} \\
\frac{\omega +\alpha }{2\sqrt{2\omega \left( \omega +\alpha \right) }} \\
\frac{-\gamma ^{\ast }}{2\sqrt{2\omega \left( \omega +\alpha \right) }} \\
\frac{\omega -\alpha }{2\sqrt{2\omega \left( \alpha -\omega \right) }}%
\end{array}%
\right] ,
\end{equation}%
\begin{equation}
v_{2}(\omega )=\left[
\begin{array}{c}
\frac{\gamma }{2\sqrt{2\omega \left( \alpha -\omega \right) }} \\
\frac{-(\omega +\alpha )}{2\sqrt{2\omega \left( \omega +\alpha \right) }} \\
\frac{\gamma ^{\ast }}{2\sqrt{2\omega \left( \omega +\alpha \right) }} \\
\frac{\omega -\alpha }{2\sqrt{2\omega \left( \alpha -\omega \right) }}%
\end{array}%
\right] ,
\end{equation}%
also constitutes an orthonormal basis for the eigenspace of $\omega $ of the
example \ref{Ins2}. Substituting them into the BV matrix, one will find that
this normal derivative BV matrix can diagonalize the Hamiltonian as the $%
T_{n}$ given in Eq. (\ref{TnExm}). Even if all the eigenvalues are
nondegenerate, the eigenvectors can choose their phases freely. In sum, the
normal derivative BV transformation can never be unique.

That poses a natural question: Can different normal BV transformations give
rise to different diagonalized forms for a certain quadratic Hamiltonian?
Or, put it another way, is the diagonalized form of a quadratic Hamiltonian
unique? The answer will be yes if one does not care the order of the
quadratic terms present in a diagonal Hamiltonian.

\begin{theorem}
\label{UniqueB} If a quadratic Hamiltonian of bosons is BV diagonalizable,
its diagonalized form will be unique up to a permutation of the quadratic
terms.
\end{theorem}

\begin{proof}
Suppose that there are two diagonalized forms for the Hamiltonian of Eq. (%
\ref{Ham1}). According to Eq. (\ref{B4}), they can be written as%
\begin{eqnarray}
H_{1} &=&\sum_{i=1}^{n}\omega _{i}d_{i}^{\dag }d_{i}+\frac{1}{2}%
\sum_{i=1}^{n}\omega _{i}-\frac{1}{2}\mathrm{tr}(\alpha ), \\
H_{2} &=&\sum_{i=1}^{n}\omega _{i}^{\prime }d_{i}^{\dag }d_{i}+\frac{1}{2}%
\sum_{i=1}^{n}\omega _{i}^{\prime }-\frac{1}{2}\mathrm{tr}(\alpha ).
\end{eqnarray}%
We shall prove that the set $\{\omega _{1},\omega _{2},\cdots ,\omega _{n}\}$
is identical to the set $\{\omega _{1}^{\prime },\omega _{2}^{\prime
},\cdots ,\omega _{n}^{\prime }\}$. Here, if $\omega _{i}=\omega
_{j}^{\prime }$, they are both counted up to the same multiplicity.

According to the lemma \ref{Similar}, the dynamic matrices of $H_{1}$ and $%
H_{2}$ are both similar to that of the Hamiltonian $H$ of Eq. (\ref{Ham1}),%
\begin{eqnarray}
T_{n1}^{-1}DT_{n1} &=&D_{1}, \\
T_{n2}^{-1}DT_{n2} &=&D_{2},
\end{eqnarray}%
where $D_{1}$, $D_{2}$, and $D$ are, respectively, the dynamic matrices for $%
H_{1}$, $H_{2}$ and $H$, and $T_{n1}$ and $T_{n2}$ are both the normal
derivative BV matrices,%
\begin{eqnarray}
T_{n1}^{\dag }I_{-}T_{n1} &=&I_{-},  \label{NBV1} \\
T_{n2}^{\dag }I_{-}T_{n2} &=&I_{-}.  \label{NBV2}
\end{eqnarray}%
Thereby, $D_{1}$ and $D_{2}$ are similar to each other,%
\begin{equation}
\left( T_{n1}^{-1}T_{n2}\right) ^{-1}D_{1}\left( T_{n1}^{-1}T_{n2}\right)
=D_{2}.  \label{D1D2}
\end{equation}%
Observe%
\begin{eqnarray}
D_{1} &=&\mathrm{diag}(\omega _{1},\cdots ,\omega _{n},-\omega _{1},\cdots
,-\omega _{n}), \\
D_{2} &=&\mathrm{diag}(\omega _{1}^{\prime },\cdots ,\omega _{n}^{\prime
},-\omega _{2}^{\prime },\cdots ,-\omega _{n}^{\prime }).
\end{eqnarray}%
We have
\begin{widetext}%
\begin{equation}
\{\omega _{1},\omega _{2},\cdots ,\omega _{n},-\omega _{1},-\omega
_{2},\cdots ,-\omega _{n}\}=\{\omega _{1}^{\prime },\omega _{2}^{\prime
},\cdots ,\omega _{n}^{\prime },-\omega _{1}^{\prime },-\omega _{2}^{\prime
},\cdots ,-\omega _{n}^{\prime }\},
\end{equation}%
\end{widetext}
where each eigenenergy is counted up to its multiplicity. As a consequence,%
\begin{equation}
\omega _{1}\in \{\omega _{1}^{\prime },\omega _{2}^{\prime },\cdots ,\omega
_{n}^{\prime }\},
\end{equation}%
or
\begin{equation}
\omega _{1}\in \{-\omega _{1}^{\prime },-\omega _{2}^{\prime },\cdots
,-\omega _{n}^{\prime }\}.
\end{equation}%
If
\begin{equation}
\omega _{1}\in \{-\omega _{1}^{\prime },-\omega _{2}^{\prime },\cdots
,-\omega _{n}^{\prime }\},
\end{equation}%
let $\omega _{1}$ be $-\omega _{1}^{\prime }$, i.e., $\omega _{1}=-\omega
_{1}^{\prime }$, without loss of generality. Thus, one obtains from Eq. (\ref%
{D1D2})%
\begin{equation}
v(\omega _{1})=T_{n1}^{-1}T_{n2}v(-\omega _{1}^{\prime }).
\end{equation}%
This gives rise to%
\begin{equation}
v^{\dag }(\omega _{1})I_{-}v(\omega _{1})=v^{\dag }(-\omega _{1}^{\prime
})I_{-}v(-\omega _{1}^{\prime }),  \label{Contrad}
\end{equation}%
where Eqs. (\ref{NBV1}) and (\ref{NBV2}) have been used. However,
\begin{eqnarray}
v^{\dag }(\omega _{1})I_{-}v(\omega _{1}) &=&1, \\
v^{\dag }(-\omega _{1}^{\prime })I_{-}v(-\omega _{1}^{\prime }) &=&-1.
\end{eqnarray}%
Obviously, they contradict the equation (\ref{Contrad}). Therefore,
\begin{equation}
\omega _{1}\notin \{-\omega _{1}^{\prime },-\omega _{2}^{\prime },\cdots
,-\omega _{n}^{\prime }\},
\end{equation}%
it must belong to the set $\{\omega _{1}^{\prime },\omega _{2}^{\prime
},\cdots ,\omega _{n}^{\prime }\}$, i.e.,
\begin{equation}
\omega _{1}\in \{\omega _{1}^{\prime },\omega _{2}^{\prime },\cdots ,\omega
_{n}^{\prime }\}.
\end{equation}%
This implies that%
\begin{equation}
\{\omega _{1},\omega _{2},\cdots ,\omega _{n}\}\subset \{\omega _{1}^{\prime
},\omega _{2}^{\prime },\cdots ,\omega _{n}^{\prime }\}.
\end{equation}%
For the same reason,%
\begin{equation}
\{\omega _{1}^{\prime },\omega _{2}^{\prime },\cdots ,\omega _{n}^{\prime
}\}\subset \{\omega _{1},\omega _{2},\cdots ,\omega _{n}\}.
\end{equation}%
So
\begin{equation}
\{\omega _{1},\omega _{2},\cdots ,\omega _{n}\}=\{\omega _{1}^{\prime
},\omega _{2}^{\prime },\cdots ,\omega _{n}^{\prime }\}.
\end{equation}%
This means that $H_{1}$ and $H_{2}$ are identical up to a permutation of the
members of the set $\{\omega _{1},\omega _{2},\cdots ,\omega _{n}\}$. In
other words, the diagonalized form of a quadratic Hamiltonian is unique up
to a permutation of the quadratic terms.
\end{proof}

Up to now, the BV diagonalization is always meant to diagonalize a
Hamiltonian with respect to the normal bosons, which fulfill the standard
commutation rules. Sometimes, e.g., in quantum electrodynamics, the
so-called time-polarized bosons are needed, they satisfy the abnormal
commutation relations,%
\begin{equation}
\lbrack b_{i},b_{j}^{\dag }]=-\delta _{ij},\text{ \ }[b_{i},b_{j}]=0,\text{
\ }[b_{i}^{\dag },b_{j}^{\dag }]=0.  \label{AbCom}
\end{equation}%
If one exchanges the roles of $b_{i}$ and $b_{i}^{\dag }$, and interprets
them respectively as a creator and annihilator, i.e.,%
\begin{equation}
d_{i}=b_{i}^{\dag },\text{ \ }d_{i}^{\dag }=b_{i},
\end{equation}%
he has%
\begin{equation}
\lbrack d_{i},d_{j}^{\dag }]=\delta _{ij},\text{ \ }[d_{i},d_{j}]=0,\text{ \
}[d_{i}^{\dag },d_{j}^{\dag }]=0.
\end{equation}%
That is to say, the time-polarized bosons can always be transformed into the
normal bosons, and vice versa. Hence, a Hamiltonian can also be diagonalized
with respect to the time-polarized bosons. Accordingly, Eq. (\ref{B4}) will
become%
\begin{equation}
H=\sum_{i=1}^{n}-\omega _{i}\left( -b_{i}^{\dag }b_{i}\right) -\frac{1}{2}%
\sum_{i=1}^{n}\omega _{i}-\frac{1}{2}\mathrm{tr}(\alpha ),
\end{equation}%
where $n_{i}=-b_{i}^{\dag }b_{i}$ are the particle-number operators for the
time-polarized bosons. Actually, a Hamiltonian can be diagonalized with
respect to the normal bosons, or the time-polarized bosons, or the mixture
of both the normal and time-polarized bosons as you wish, e.g.,%
\begin{eqnarray}
H &=&\sum_{i=1}^{m}\omega _{i}d_{i}^{\dag }d_{i}+\sum_{i=m+1}^{n}\omega
_{i}b_{i}^{\dag }b_{i}  \notag \\
&&+\frac{1}{2}\sum_{i=1}^{m}\omega _{i}-\frac{1}{2}\sum_{i=m+1}^{n}\omega
_{i}-\frac{1}{2}\mathrm{tr}(\alpha ),
\end{eqnarray}%
where $0\leq m\leq n$. This fact will be used in Sec. \ref{TPP}. Anyway, the
diagonalization is unique with respect to the normal bosons. The normal
bosons will be used, by default, for the diagonalization of the Bose system
unless otherwise specified.

By the way, we would like to note that the conclusions up to now are also
valid for the Bose system whose Hamiltonian is represented quadratically
with regard to the time-polarized bosons, or the mixture of both the normal
and time-polarized bosons. That is because, as mentioned above, all the
time-polarized bosons can be transformed into the normal bosons.

To conclude, a quadratic Hamiltonian of bosons has BV diagonalization if and
only if its dynamic matrix is physically diagonalizable. If the
diagonalization exists, its form is unique.

Thus far, a whole theory of diagonalization has been achieved for the Bose
system.

\section{Diagonalization Theory of Fermi Systems \label{DTFS}}

In this section, we turn to the Fermi case. We shall study first the
existence and then the uniqueness of the BV diagonalization for the Fermi
system.

\subsection{Existence}

The Heisenberg equation for the fermionic field $\psi$ can be derived from
Eq. (\ref{Ham1}),
\begin{equation}
i\frac{\mathrm{d}}{\mathrm{d}t}\psi=D\psi,
\end{equation}
where $D$ is the dynamic matrix for the Fermi system,
\begin{equation}
D=\left[
\begin{array}{cc}
\alpha & \gamma \\
\gamma^{\dag} & -\widetilde{\alpha}%
\end{array}
\right] .
\end{equation}
In contrast to the Bose system where the dynamic matrix is distinct from the
coefficient matrix, the dynamic matrix $D$ is now identical to the
coefficient matrix $M$,
\begin{equation}
M=\left[
\begin{array}{cc}
\alpha & \gamma \\
\gamma^{\dag} & -\widetilde{\alpha}%
\end{array}
\right] .
\end{equation}
This demonstrates that the coefficient matrix $M$ will control the dynamic
behavior of the system, just as the dynamic matrix $D$. That is the radical
difference between the Fermi and Bose systems. For the latter, as we know,
the coefficient matrix does not control the dynamic behavior of the system.
Now that $D=M$ and $M$ is Hermitian, $D$ is Hermitian, too. That is another
feature of the Fermi system, which will bring us much convenience.

Similar to Eq. (\ref{DIM}), the relation between $D$ and $M$ can be formally
written as
\begin{equation}
D=I_{+}M.
\end{equation}
This relation is useful in the diagonalization of the Fermi system.

As before, let us consider the eigenvalue problem,\
\begin{equation}
\omega\psi=D\psi.
\end{equation}

\begin{lemma}
For a quadratic Hamiltonian of fermions, its dynamic matrix is always BV
diagonalizable.
\end{lemma}

\begin{proof}
As mentioned above, the dynamic matrix for a Fermi system is Hermitian. It
is well known that a Hermitian matrix is diagonalizable, and all its
eigenvalues are real. Therefore, the dynamic matrix for a Fermi system is
always BV diagonalizable.
\end{proof}

This property is basically different from the Bose system. There, the
dynamic matrix is not always diagonalizable. Needless to say, it is not
always BV diagonalizable.

As $D=M$ and both are Hermitian, they can be diagonalized by an exactly
identical unitary transformation. Mathematically, a unitary transformation
is always a similar transformation, the diagonalization manner of $D$ is not
inharmonious with that of $M$ any longer. The problem present in the Bose
system disappears spontaneously in the Fermi system.

Analogous to the Bose system, one can easily show that the lemmas \ref{Lmm1}%
, \ref{Lmm2}, \ref{Lmm3}, \ref{Lmm4}, and \ref{Similar} are all valid for
the Fermi system. Since the dynamic matrix is always BV diagonalizable now,
the necessary condition for the BV diagonalization will hold automatically
for a Fermi system. That guarantees further that the lemmas \ref{Lmm5}, \ref%
{Lmm6}, and \ref{Lmm7} also hold for the Fermi system. All those eight
lemmas stem from the Hermiticity of the Hamiltonian, and are irrespective of
the statistics and metric of the system.

By introducing a sesquilinear form with $I_{+}$ and substituting $I_{-}$
with $I_{+}$, the lemmas \ref{Ortho} and \ref{OrNormal} are both valid for
the Fermi system, the only difference lies in that the norm with respect to $%
I_{+}$ is positive definite whereas the norm with respect to $I_{-}$ is
indefinite. The lemma \ref{ZeroSpace} also holds for the Fermi case, but its
proof needs quite a lot of modification, which we give below.

\begin{proof}
The eigenspace $V_{0}$ of zero eigenvalue is even dimensional, let the
dimension be $2m$ ($m\in%
\mathbb{N}
$). According to the lemma \ref{Lmm6}, there always exists a basis for $%
V_{0} $ that satisfies the requirement of Eq. (\ref{VP2}), i.e.,%
\begin{equation}
v_{m+l}(0)=\Sigma_{x}v_{l}^{\ast}(0),\text{ \ }l=1,2,\cdots,m,
\end{equation}

When $m=1$, $\dim(V_{0})=2$, there are two basis vectors, i.e., $v_{1}(0)$
and $v_{2}(0)$, they are linearly indepentdent.

First of all, we would make $v_{1}(0)$ normalized,%
\begin{equation}
v_{1}^{\dag }(0)I_{+}v_{1}(0)=1.
\end{equation}%
And then we consider $v_{2}(0)$. In fact, it is also normalized,
\begin{equation}
v_{2}^{\dag }(0)I_{+}v_{2}(0)=1,
\end{equation}%
that is because%
\begin{equation}
v_{2}(0)=\Sigma _{x}v_{1}^{\ast }(0).
\end{equation}%
There are two possible cases for $v_{2}(0)$: (1) It is orthogonal to $%
v_{1}(0)$. (2) It is not orthogonal to $v_{1}(0)$.

If $v_{2}(0)\bot v_{1}(0)$, i.e.,
\begin{equation}
v_{1}^{\dag }(0)I_{+}v_{2}(0)=0,
\end{equation}%
we have%
\begin{eqnarray}
v_{1}^{\dag }(0)I_{+}v_{1}(0) &=&1,  \label{V2V1A} \\
v_{2}^{\dag }(0)I_{+}v_{2}(0) &=&1, \\
v_{1}^{\dag }(0)I_{+}v_{2}(0) &=&0.
\end{eqnarray}%
and%
\begin{equation}
v_{2}(0)=\Sigma _{x}v_{1}^{\ast }(0),  \label{V2V1B}
\end{equation}%
So the basis $\{v_{1}(0),v_{2}(0)\}$ is itself orthonormal and satisfies the
requirement of Eq. (\ref{VP2}). In other words, the lemma holds if $%
v_{2}(0)\bot v_{1}(0)$.

If $v_{2}(0)$ is not orthogonal to $v_{1}(0)$, i.e.,%
\begin{equation}
v_{1}^{\dag }(0)I_{+}v_{2}(0)\neq 0,
\end{equation}%
we shall introduce two vectors $w_{1}(0)$ and $w_{2}(0)$ as follows,%
\begin{eqnarray}
w_{1}(0) &=&av_{1}(0)+bv_{2}(0),  \label{V10} \\
w_{2}(0) &=&\Sigma _{x}w_{1}^{\ast }(0).  \label{V20}
\end{eqnarray}%
Here $a\in
\mathbb{C}
$ and $b\in
\mathbb{C}
$ are two coefficients, they will be determined by the orthonormal
conditions,%
\begin{eqnarray}
w_{1}^{\dag }(0)I_{+}w_{1}(0) &=&1,  \label{V1V11} \\
w_{1}^{\dag }(0)I_{+}w_{2}(0) &=&0.  \label{V1V20}
\end{eqnarray}%
Equations (\ref{V10}) and (\ref{V20}) show that both $w_{1}(0)$ and $%
w_{2}(0) $ are linear combinations of $v_{1}(0)$ and $v_{2}(0)$.
Consequently, $w_{1}(0)$ and $w_{2}(0)$ are also the eigenvectors of zero
eigenvalue, i.e., $w_{1}(0)\in V_{0}$ and $w_{2}(0)\in V_{0}$.

Observe
\begin{equation}
v_{1}^{\dag }(0)I_{+}v_{2}(0)=\left[ \widetilde{v}_{1}(0)I_{+}\Sigma
_{x}v_{1}(0)\right] ^{\ast }.
\end{equation}%
We can adjust the phase of $v_{1}(0)$ anew so that
\begin{equation}
v_{1}^{\dag }(0)I_{+}v_{2}(0)>0.
\end{equation}%
By use of Cauchy inequality, we obtain
\begin{equation}
v_{1}^{\dag }(0)I_{+}v_{2}(0)<\sqrt{v_{1}^{\dag }(0)I_{+}v_{1}(0)}\sqrt{%
v_{2}^{\dag }(0)I_{+}v_{2}(0)},
\end{equation}%
where we have used the fact that $v_{1}(0)$ and $v_{2}(0)$ are linearly
independent. Since%
\begin{equation}
\sqrt{v_{1}^{\dag }(0)I_{+}v_{1}(0)}\sqrt{v_{2}^{\dag }(0)I_{+}v_{2}(0)}=1,
\end{equation}%
we have%
\begin{equation}
v_{1}^{\dag }(0)I_{+}v_{2}(0)<1.
\end{equation}%
In brief, we can alway have%
\begin{equation}
0<v_{1}^{\dag }(0)I_{+}v_{2}(0)<1,
\end{equation}%
when $v_{2}(0)$ is not orthogonal to $v_{1}(0)$.

Under such choice, Eqs. (\ref{V1V11}) and (\ref{V1V20}) become%
\begin{eqnarray}
a^{\ast }a+b^{\ast }b+\left( a^{\ast }b+b^{\ast }a\right) v_{1}^{\dag
}(0)I_{+}v_{2}(0) &=&1, \\
\left( a^{\ast }a^{\ast }+b^{\ast }b^{\ast }\right) v_{1}^{\dag
}(0)I_{+}v_{2}(0)+2a^{\ast }b^{\ast } &=&0.
\end{eqnarray}%
It can be readily confirmed that there exists at least the following real
solution for the coefficients $a$ and $b$,%
\begin{eqnarray}
a &=&\frac{1}{2\sqrt{1+v_{1}^{\dag }(0)I_{+}v_{2}(0)}}  \notag \\
&&+\frac{1}{2\sqrt{1-v_{1}^{\dag }(0)I_{+}v_{2}(0)}}, \\
b &=&\frac{1}{2\sqrt{1+v_{1}^{\dag }(0)I_{+}v_{2}(0)}}  \notag \\
&&-\frac{1}{2\sqrt{1-v_{1}^{\dag }(0)I_{+}v_{2}(0)}}.
\end{eqnarray}%
With this solution, we have%
\begin{gather}
w_{1}^{\dag }(0)I_{+}w_{1}(0)=1, \\
w_{2}^{\dag }(0)I_{+}w_{2}(0)=1, \\
w_{1}^{\dag }(0)I_{+}w_{2}(0)=0, \\
w_{2}(0)=\Sigma _{x}w_{1}^{\ast }(0).
\end{gather}%
That is to say, the set $\{w_{1}(0),w_{2}(0)\}$ will form an orthonormal
basis for $V_{0}$, and satisfy the requirement of Eq. (\ref{VP2}). This
implies that the lemma also holds if $v_{2}(0)$ is not orthogonal to $%
v_{1}(0)$.

To sum up, the lemma will always hold when $m=1$.

Suppose that the lemma holds when $m=l$ ($l\in
\mathbb{N}
$). We consider then the case where $m=l+1$. Obviously, it has a proper
subspace $W$ spanned by the linearly independent set $\{v_{1}(0),v_{l+2}(0)%
\} $, i.e.,
\begin{equation}
W=\mathrm{span}(v_{1}(0),v_{l+2}(0)).
\end{equation}%
Taking notice of%
\begin{equation}
v_{l+2}(0)=\Sigma _{x}v_{1}^{\ast }(0),
\end{equation}%
and following the same arguments as those for the case of $m=1$, we can
obtain an orthonormal basis for $W$,%
\begin{equation}
v_{i}^{\dag }(0)I_{-}v_{j}(0)=-\lambda _{i}\delta _{ij},
\end{equation}%
where $\lambda _{i}=\pm 1$ and $i,j=1,l+2$. It is evident that this basis
satisfies the requirement of Eq. (\ref{VP2}).

The rest steps of mathematical induction are simply similar to those for the
Bose case. The convention of Eq. (\ref{VP2}) can be kept by the modified
Gram-Schmidt orthogonalization process.
\end{proof}

The lemma \ref{BasisC} still holds for the Fermi system, with $\lambda
_{i}\equiv1$ for $i=1,2,\cdots,2n$.

Since all the eigenvectors are normalized to $+1$ now, one can not use the
sign of the norm to stipulate an order within a mode pair. Here, we shall
resort to the sign of the eigenvalue: The first eigenvalue in a pair is
positive, and the second one negative; it is arbitrary if both the
eigenvalues in a pair are equal to zero. Under such stipulation, the normal
derivative BV transformation has the form,
\begin{widetext}%
\begin{gather}
\psi =T_{n}\varphi ,  \label{Stipu1} \\
T_{n}=\left[
\begin{array}{cccccccc}
v(\omega _{1}), & v(\omega _{2}), & \cdots , & v(\omega _{n}), & v(-\omega
_{1}), & v(-\omega _{2}), & \cdots , & v(-\omega _{n})%
\end{array}%
\right] ,  \label{Stipu2}
\end{gather}%
\end{widetext}
where
\begin{equation}
\omega _{i}\geq 0,\text{ \ }i=1,2,\cdots ,n.  \label{Stipu3}
\end{equation}%
That is to say, the left half of $T_{n}$ is filled with the eigenvectors
with nonnegative eigenvalues; the right half of $T_{n}$ is filled with the
eigenvectors with nonpositive eigenvalues.

With $T_{n}$ ordered as above, the lemma \ref{BVTn} holds for the Fermi case,%
\begin{equation}
T_{n}^{\dag }I_{+}T_{n}=I_{+},
\end{equation}%
i.e., $T_{n}$ is a member of the $U(2n)$ group \cite{Chen}. This lemma
asserts that the new filed is a standard fermionic field.

The lemma \ref{MDiag} must be modified as follows,%
\begin{equation}
T_{n}^{\dag }MT_{n}=\mathrm{diag}(\omega _{1},\cdots ,\omega _{n},-\omega
_{1},\cdots ,-\omega _{n}).
\end{equation}%
That is because%
\begin{eqnarray}
T_{n}^{\dag }MT_{n} &=&T_{n}^{\dag }I_{+}DT_{n}  \notag \\
&=&T_{n}^{\dag }I_{+}T_{n}T_{n}^{-1}DT_{n}  \notag \\
&=&T_{n}^{-1}DT_{n},
\end{eqnarray}%
where $T_{n}^{\dag }I_{+}T_{n}=I_{+}$ has been used.

At last, we arrive at the diagonalization theorem for the Fermi system.

\begin{theorem}
\label{FThm} Any quadratic Hamiltonian of fermions is BV diagonalizable.
\end{theorem}

Apparently, the diagonalized form for the Hamiltonian is%
\begin{equation}
H=\sum_{i=1}^{n}\omega_{i}d_{i}^{\dag}d_{i}-\frac{1}{2}\sum_{i=1}^{n}%
\omega_{i}+\frac{1}{2}\mathrm{tr}(\alpha),  \label{F4}
\end{equation}
where all the eigenenergies are nonnegative,
\begin{equation}
\omega_{i}\geq0,\text{ \ }i=1,2,\cdots,n.
\end{equation}

Here, it is worth emphasizing that the BV diagonalization for a Fermi system
is itself of unitary diagonalization, that is because $T_{n}$ is, in fact, a
unitary matrix, $T_{n}^{\dag}T_{n}=I_{+}$.

All in all, the BV diagonalization for a quadratic Hamiltonian of fermions
is much simpler than that for a quadratic Hamiltonian of bosons.

\begin{example}
\label{Insf1}%
\begin{equation}
H=\alpha(c_{1}^{\dag}c_{1}+c_{2}^{\dag}c_{2})+\gamma(c_{1}^{\dag}c_{2}^{\dag
}-c_{1}c_{2}).
\end{equation}
\end{example}

\begin{solution}
The dynamic matrix $D$ is
\begin{equation}
D=\left[
\begin{array}{cccc}
\alpha & 0 & 0 & \gamma \\
0 & \alpha & -\gamma & 0 \\
0 & -\gamma & -\alpha & 0 \\
\gamma & 0 & 0 & -\alpha%
\end{array}%
\right] .
\end{equation}%
It has only a pair of eigenvalues, $(\omega ,-\omega )$ where%
\begin{equation}
\omega =\sqrt{\alpha ^{2}+\gamma ^{2}},
\end{equation}%
they are both two-fold degenerate. The normal BV matrix can be chosen as
follows,%
\begin{eqnarray}
T_{n} &=&\left[
\begin{array}{cccc}
v_{1}(\omega ), & v_{2}(\omega ), & v_{1}(-\omega ), & v_{2}(-\omega )%
\end{array}%
\right]  \notag \\
&=&\frac{1}{\sqrt{\left( \omega -\alpha \right) ^{2}+\gamma ^{2}}}  \notag \\
&&\times
\begin{bmatrix}
\gamma & 0 & 0 & \omega -\alpha \\
0 & -\gamma & \omega -\alpha & 0 \\
0 & \omega -\alpha & \gamma & 0 \\
\omega -\alpha & 0 & 0 & -\gamma%
\end{bmatrix}%
,
\end{eqnarray}%
where the convention of Eq. (\ref{VP1}) has been used for $v_{1}(-\omega )$
and $v_{2}(-\omega )$. The diagonalized Hamiltonian has the form,%
\begin{equation}
H=\omega (d_{1}^{\dag }d_{1}+d_{2}^{\dag }d_{2})-\omega +\alpha .
\end{equation}
\end{solution}

\begin{example}
\begin{equation}
H=\mu(c_{1}^{\dag}c_{2}+c_{2}^{\dag}c_{1})+\nu(c_{1}^{\dag}c_{2}^{%
\dag}-c_{1}c_{2}),
\end{equation}
where $\nu>0$.
\end{example}

\begin{solution}
The dynamic matrix $D$ is%
\begin{equation}
D=%
\begin{bmatrix}
0 & \mu & 0 & \nu \\
\mu & 0 & -\nu & 0 \\
0 & -\nu & 0 & -\mu \\
\nu & 0 & -\mu & 0%
\end{bmatrix}%
.
\end{equation}%
There are two pairs of eigenvalues, $(\omega _{1},-\omega _{1})$ and $%
(\omega _{2},-\omega _{2})$ where%
\begin{equation}
\omega _{1}=\mu +\nu ,\text{ \ }\omega _{2}=\mu -\nu .
\end{equation}%
The eigenvectors $v(\omega _{1})$ and $v(\omega _{2})$ can be chosen as
follows,%
\begin{equation}
v(\omega _{1})=\frac{1}{2}%
\begin{bmatrix}
1 \\
1 \\
-1 \\
1%
\end{bmatrix}%
,\text{ \ }v(\omega _{2})=\frac{1}{2}%
\begin{bmatrix}
1 \\
1 \\
1 \\
-1%
\end{bmatrix}%
.
\end{equation}%
Correspondingly, $v(-\omega _{1})$ and $v(-\omega _{2})$ can be obtained
from the convention of Eq. (\ref{VP1}).

The normal BV matrix and the form of the diagonalized Hamiltonian are listed
as follows.

1. If $\mu >\nu $,%
\begin{eqnarray}
T_{n} &=&\left[
\begin{array}{cccc}
v(\omega _{1}), & v(\omega _{2}), & v(-\omega _{1}), & v(-\omega _{2})%
\end{array}%
\right] , \\
H &=&\omega _{1}d_{1}^{\dag }d_{1}+\omega _{2}d_{2}^{\dag }d_{2}-\frac{1}{2}%
\left( \omega _{1}+\omega _{2}\right) .
\end{eqnarray}

2. If $\mu <-\nu $,%
\begin{eqnarray}
T_{n} &=&\left[
\begin{array}{cccc}
v(-\omega _{1}), & v(-\omega _{2}), & v(\omega _{1}), & v(\omega _{2})%
\end{array}%
\right] , \\
H &=&-\omega _{1}d_{1}^{\dag }d_{1}-\omega _{2}d_{2}^{\dag }d_{2}+\frac{1}{2}%
\left( \omega _{1}+\omega _{2}\right) .
\end{eqnarray}

3. If $-\nu \leq \mu \leq \nu $,%
\begin{eqnarray}
T_{n} &=&\left[
\begin{array}{cccc}
v(\omega _{1}), & v(-\omega _{2}), & v(-\omega _{1}), & v(\omega _{2})%
\end{array}%
\right] , \\
H &=&\omega _{1}d_{1}^{\dag }d_{1}-\omega _{2}d_{2}^{\dag }d_{2}-\frac{1}{2}%
\left( \omega _{1}-\omega _{2}\right) .
\end{eqnarray}
\end{solution}

\subsection{Uniqueness}

Following the same arguments as those for the Bose system, one can readily
know that there exist much different BV transformations that can all realize
BV diagonalization to the same Hamiltonian of fermions. Nevertheless, the
diagonalized form will be unique up to a permutation of the quadratic terms.

\begin{theorem}
The diagonalized form for a quadratic Hamiltonian of fermions is unique up
to a permutation of the quadratic terms.
\end{theorem}

\begin{proof}
Suppose that there are two diagonalized forms for the Hamiltonian of Eq. (%
\ref{Ham1}). According to Eq. (\ref{F4}), they can be written as%
\begin{eqnarray}
H_{1} &=&\sum_{i=1}^{n}\omega _{i}d_{i}^{\dag }d_{i}-\frac{1}{2}%
\sum_{i=1}^{n}\omega _{i}+\frac{1}{2}\mathrm{tr}(\alpha ), \\
H_{2} &=&\sum_{i=1}^{n}\omega _{i}^{\prime }d_{i}^{\dag }d_{i}-\frac{1}{2}%
\sum_{i=1}^{n}\omega _{i}^{\prime }+\frac{1}{2}\mathrm{tr}(\alpha ),
\end{eqnarray}%
where%
\begin{eqnarray}
\omega _{i} &\geq &0,\text{ \ }i=1,2,\cdots ,n,  \label{W1} \\
\omega _{i}^{\prime } &\geq &0,\text{ \ }i=1,2,\cdots ,n.  \label{W2}
\end{eqnarray}%
Let $D_{1}$ and $D_{2}$ be the dynamic matrices for $H_{1}$ and $H_{2}$,
respectively. Obviously, they are both diagonal,%
\begin{eqnarray}
D_{1} &=&\mathrm{diag}(\omega _{1},\cdots ,\omega _{n},-\omega _{1},\cdots
,-\omega _{n}), \\
D_{2} &=&\mathrm{diag}(\omega _{1}^{\prime },\cdots ,\omega _{n}^{\prime
},-\omega _{2}^{\prime },\cdots ,-\omega _{n}^{\prime }).
\end{eqnarray}%
As $D_{1}$ and $D_{2}$ are similar to each other, one has
\begin{widetext}
\begin{equation}
\{\omega _{1},\omega _{2},\cdots ,\omega _{n},-\omega _{1},-\omega
_{2},\cdots ,-\omega _{n}\}=\{\omega _{1}^{\prime },\omega _{2}^{\prime
},\cdots ,\omega _{n}^{\prime },-\omega _{1}^{\prime },-\omega _{2}^{\prime
},\cdots ,-\omega _{n}^{\prime }\}.
\end{equation}%
\end{widetext}
Paying attention to Eqs. (\ref{W1}) and (\ref{W2}), he can further obtain%
\begin{equation}
\{\omega _{1},\omega _{2},\cdots ,\omega _{n}\}=\{\omega _{1}^{\prime
},\omega _{2}^{\prime },\cdots ,\omega _{n}^{\prime }\}.
\end{equation}%
This demonstrates that the diagonalized form of a quadratic Hamiltonian is
unique up to a permutation of the quadratic terms.
\end{proof}

Since all the energies of quasiparticles are nonnegative in the Hamiltonian
of Eq. (\ref{F4}), there will be no quasiparticle in the ground state of the
system, viz., the ground state is exactly the vacuum of quasiparticles. If
one does not obey the stipulation given in Eqs. (\ref{Stipu1}), (\ref{Stipu2}%
), and (\ref{Stipu3}), he can obtain other diagonalized forms for the
Hamiltonian. Nevertheless, the ground states of the system will not be the
vacuum of quasiparticles. For example, instead of Eqs. (\ref{Stipu1}), (\ref%
{Stipu2}), and (\ref{Stipu3}), let us stipulate anew that
\begin{widetext}%
\begin{gather}
\psi =T_{n}\varphi , \\
T_{n}=\left[
\begin{array}{cccccccc}
v(-\omega _{1}), & v(-\omega _{2}), & \cdots , & v(-\omega _{n}), & v(\omega
_{1}), & v(\omega _{2}), & \cdots , & v(\omega _{n})%
\end{array}%
\right] ,
\end{gather}%
\end{widetext}
where%
\begin{equation}
\omega _{i}\geq 0,\text{ \ }i=1,2,\cdots ,n.
\end{equation}%
It is easy to show that the Hamiltonian has a new diagonalized form,\newpage
\begin{equation}
H=-\sum_{i=1}^{n}\omega _{i}d_{i}^{\dag }d_{i}+\frac{1}{2}%
\sum_{i=1}^{n}\omega _{i}+\frac{1}{2}\mathrm{tr}(\alpha ).
\end{equation}%
Now, all the energies of quasiparticles are nonpositive, the ground state of
the system will be the Fermi sea which is occupied fully by quasiparticles.
Here, the elementary excitations of the system would be quasiholes rather
than quasiparticles. Upon the particle-hole transformation,
\begin{equation}
f_{i}=d_{i}^{\dag },\text{ \ }f_{i}^{\dag }=d_{i},\text{ \ }i=1,2,\cdots ,n,
\end{equation}%
the new diagonalized form can be transformed into the old one,%
\begin{equation}
H=\sum_{i=1}^{n}\omega _{i}f_{i}^{\dag }f_{i}-\frac{1}{2}\sum_{i=1}^{n}%
\omega _{i}+\frac{1}{2}\mathrm{tr}(\alpha ).
\end{equation}%
and vice versa. Hence, both are essentially equivalent. Of course, you can
also use the mixing picture if you like, e.g.,%
\begin{eqnarray}
H &=&\sum_{i=1}^{m}\omega _{i}f_{i}^{\dag }f_{i}-\sum_{i=m+1}^{n}\omega
_{i}d_{i}^{\dag }d_{i}  \notag \\
&&-\frac{1}{2}\sum_{i=1}^{m}\omega _{i}+\frac{1}{2}\sum_{i=m+1}^{n}\omega
_{i}+\frac{1}{2}\mathrm{tr}(\alpha ),
\end{eqnarray}%
where $0<m<n$. The elementary excitations of the system include now both the
quasiparticles and quasiholes. Usually, the hole and mixing pictures are
less convenient than the particle picture. That is the reason why we
stipulate the BV transformation as in Eqs. (\ref{Stipu1}), (\ref{Stipu2}),
and (\ref{Stipu3}). By default, the particle picture will be used for the
diagonalization of the Fermi system unless otherwise specified. The
diagonalization is unique in the this picture.

To conclude, the BV diagonalization exists and is unique for every quadratic
Hamiltonian of fermions.

\section{Application to Bose Systems \label{ABS}}

Now, we apply the diagonalization theory of bosons to real systems. We shall
concentrate ourselves on the two typical Hamiltonians: the normal
Hamiltonian and the pairing Hamiltonian. As a matter of fact, they are the
prototypes of many practical models, and represent almost all the problems
which we encounter frequently in practice.

\subsection{The normal Hamiltonian}

First, let us consider the normal Hamiltonian,
\begin{equation}
H=\sum_{i,j=1}^{n}\alpha _{ij}c_{i}^{\dag }c_{j}.  \label{HamType1}
\end{equation}%
In this Hamiltonian, there are only the normal terms such as $c_{i}^{\dag
}c_{j}$. The abnormal terms, such as $c_{i}^{\dag }c_{j}^{\dag }$ and $%
c_{i}c_{j}$, disappear completely. This kind of Hamiltonian has been
discussed in the corollary \ref{Corollary2}, according to it, such a
Hamiltonian is always BV diagonalizable. In fact, the result can be
strengthened further as follows.

\begin{proposition}
\label{BPPType1} The normal Hamiltonian of Eq. (\ref{HamType1}) can be BV
diagonalized by the unitary transformation generated by the coefficient
matrix $\alpha$.
\end{proposition}

\begin{proof}
The eigenvalue equation is%
\begin{equation}
Dv(\omega )=\omega v(\omega ),
\end{equation}%
where $D$ is the dynamic matrix,%
\begin{equation}
D=\left[
\begin{array}{cc}
\alpha & 0 \\
0 & -\widetilde{\alpha }%
\end{array}%
\right] .
\end{equation}%
Let the eigenvector $v(\omega )$ be%
\begin{equation}
v(\omega )=\left[
\begin{array}{c}
x(\omega ) \\
y(\omega )%
\end{array}%
\right] ,
\end{equation}%
where $x(\omega )$ and $y(\omega )$ are the two subvectors of size $n$. The
eigenvalue equation of $D$ becomes
\begin{equation}
\left[
\begin{array}{cc}
\alpha & 0 \\
0 & -\widetilde{\alpha }%
\end{array}%
\right] \left[
\begin{array}{c}
x(\omega ) \\
y(\omega )%
\end{array}%
\right] =\omega \left[
\begin{array}{c}
x(\omega ) \\
y(\omega )%
\end{array}%
\right] .  \label{VW}
\end{equation}%
It reduces to
\begin{eqnarray}
\alpha x(\omega ) &=&\omega x(\omega ),  \label{XW1} \\
\widetilde{\alpha }y(\omega ) &=&-\omega y(\omega ).  \label{XW2}
\end{eqnarray}%
The first equation is exactly the eigenvalue equation of the coefficient
matrix $\alpha $. Since $\alpha $ is Hermitian, it can be unitarily
diagonalized,%
\begin{equation}
U^{\dag }\alpha U=\mathrm{diag}\left( \omega _{1},\omega _{2},\cdots ,\omega
_{n}\right) ,  \label{UAU}
\end{equation}%
where $\omega _{i}\in
\mathbb{R}
$ ($i=1,2,\cdots ,n$) are the eigenvalues of $\alpha $, and $U$ the unitary
matrix which consists of the eigenvectors of $\alpha $,%
\begin{gather}
U^{\dag }U=UU^{\dag }=I,  \label{UU} \\
U=\left[
\begin{array}{cccc}
x(\omega _{1}), & x(\omega _{2}), & \cdots , & x(\omega _{n})%
\end{array}%
\right] ,
\end{gather}%
the $x(\omega _{i})$ standing for the eigenvector of the eigenvalue $\omega
_{i}$, respectively. Obviously, Eqs. (\ref{XW1}) and (\ref{XW2}) have the
solutions,%
\begin{equation}
v(\omega _{i})=\left[
\begin{array}{c}
x(\omega _{i}) \\
0%
\end{array}%
\right] ,\text{ \ }i=1,2,\cdots ,n,
\end{equation}%
they are all the eigenvectors of Eq. (\ref{VW}), and orthonormalized as
follows,%
\begin{equation}
v^{\dag }(\omega _{i})I_{-}v(\omega _{j})=\delta _{ij},\text{ \ }%
i,j=1,2,\cdots ,n,
\end{equation}%
which comes directly from Eq. (\ref{UU}). According to the convention for $%
v(-\omega _{i})$, the rest half of the eigenvectors of Eq. (\ref{VW}) is
\begin{equation}
v(-\omega _{i})=\left[
\begin{array}{c}
0 \\
x^{\ast }(\omega _{i})%
\end{array}%
\right] ,\text{ \ }i=1,2,\cdots ,n,
\end{equation}%
they are also orthonormalized,%
\begin{equation}
v^{\dag }(-\omega _{i})I_{-}v(-\omega _{j})=-\delta _{ij},\text{ \ }%
i,j=1,2,\cdots ,n.
\end{equation}%
In fact, $x^{\ast }(\omega _{i})$ is exactly the eigenvector of $\widetilde{%
\alpha }$,%
\begin{equation}
\widetilde{\alpha }x^{\ast }(\omega _{i})=-(-\omega _{i})x^{\ast }(\omega
_{i}),
\end{equation}%
with the eigenvalue being $-\omega _{i}$. Substituting them into Eq. (\ref%
{DBVTn}), we obtain the normal BV matrix,%
\begin{equation}
T_{n}=\left[
\begin{array}{cc}
U & 0 \\
0 & U^{\ast }%
\end{array}%
\right] .  \label{UUT}
\end{equation}%
It can be readily confirmed that
\begin{equation}
T_{n}^{\dag }I_{-}T_{n}=I_{-},
\end{equation}%
\begin{equation}
T_{n}^{-}DT_{n}=\mathrm{diag}(\omega _{1},\cdots ,\omega _{n},-\omega
_{1},\cdots ,-\omega _{n}),
\end{equation}%
\begin{equation}
T_{n}^{\dag }MT_{n}=\mathrm{diag}(\omega _{1},\cdots ,\omega _{n},\omega
_{1},\cdots ,\omega _{n}),
\end{equation}%
\begin{equation}
H=\sum_{i=1}^{n}\omega _{i}d_{i}^{\dag }d_{j},  \label{Hdd}
\end{equation}%
where
\begin{equation}
\lbrack d_{i},d_{j}^{\dag }]=\delta _{i,j},\text{ \ }[d_{i},d_{j}]=0,\text{
\ }[d_{i}^{\dag },d_{j}^{\dag }]=0.
\end{equation}

Equation (\ref{UUT}) shows that the normal BV matrix $T_{n}$ for the
diagonalization of the Hamiltonian of (\ref{HamType1}) can be constructed
using the unitary matrix $U$ generated by the coefficient matrix $\alpha$.
Put it another way, the quadratic Hamiltonian of Eq. (\ref{HamType1}) can be
BV diagonalized by the unitary transformation generated by the coefficient
matrix $\alpha$.
\end{proof}

The proposition shows that, to diagonalize the Hamiltonian of Eq. (\ref%
{HamType1}), one should first find the unitary matrix $U$ from the
coefficient matrix $\alpha$, and then construct the normal BV matrix $T_{n}$
according to Eq. (\ref{UUT}).

The above procedure is feasible but somewhat redundant, it can be further
simplified. As a matter of fact, the BV transformation corresponding to $%
T_{n}$ can be reduced to a simple unitary transformation,%
\begin{equation}
c=Ud,\text{ \ }c^{\dag}=d^{\dag}U^{\dag},  \label{Udc}
\end{equation}
which can be verified readily from the substitution of $T_{n}$ into Eq. (\ref%
{TN}). As a consequence, we obtain
\begin{equation}
d\cdot d^{\dag}=I,\text{ \ }d\cdot d=0,\text{ \ }d^{\dag}\cdot d^{\dag}=0,
\label{DDI}
\end{equation}
where the standard relations,
\begin{equation}
c\cdot c^{\dag}=I,\text{ \ }c\cdot c=0,\text{ \ }c^{\dag}\cdot c^{\dag}=0,
\end{equation}
have been used. Through the unitary transformation of Eq. (\ref{Udc}), the
Hamiltonian of Eq. (\ref{HamType1}) can be straightforwardly diagonalized as
follows,%
\begin{equation}
H=c^{\dag}\alpha c=d^{\dag}U^{\dag}\alpha
Ud=\sum_{i=1}^{n}\omega_{i}d_{i}^{\dag}d_{j}.  \label{CAC}
\end{equation}
It is the same as Eq. (\ref{Hdd}).

Physically, this simple version arises directly from the fact that the
Heisenberg equation of the field $c$,
\begin{equation}
i\frac{\mathrm{d}}{\mathrm{d}t}c=\alpha c,  \label{RHeisenberg}
\end{equation}%
does not couple with its Hermitian field $c^{\dag }$. If a unitary
transformation $U$ for the field $c$ is generated by this equation of motion,%
\begin{equation}
c=Ud,
\end{equation}%
an adjoint transformation $U^{\dag }$ will be yielded meanwhile for the
Hermitian field $c^{\dag }$,%
\begin{equation}
c^{\dag }=d^{\dag }U^{\dag },
\end{equation}%
by the equation of motion,
\begin{equation}
-i\frac{\mathrm{d}}{\mathrm{d}t}c^{\dag }=c^{\dag }\alpha .
\end{equation}%
As shown by Eqs. (\ref{UAU}), (\ref{DDI}) and (\ref{CAC}), they diagonalize
the Hamiltonian exactly. The simple version makes it much easier to
diagonalize the Hamiltonian of Eq. (\ref{HamType1}). It is unnecessary to
solve the eigenvalue problem of the dynamic matrix $D$ and construct the BV
matrix $T_{n}$, but sufficient for us to find out the unitary matrix $U$
from the Hermitian matrix $\alpha $. In short, it reduces the eigenvalue
problem of $D$, which is of size $2n$, to the eigenvalue problem of $\alpha $%
, which is of size $n$.

\begin{example}
\begin{equation}
H=\varepsilon_{1}c_{1}^{\dag}c_{1}+\varepsilon_{2}c_{2}^{\dag}c_{2}+\mu
(c_{1}^{\dag}c_{2}+c_{2}^{\dag}c_{1}).
\end{equation}
\end{example}

\begin{solution}
The coefficient matrix is%
\begin{equation}
\alpha =%
\begin{bmatrix}
\varepsilon _{1} & \mu \\
\mu & \varepsilon _{2}%
\end{bmatrix}%
.
\end{equation}%
It has two eigenvalues,%
\begin{eqnarray}
\omega _{1} &=&\frac{1}{2}\left[ \varepsilon _{1}+\varepsilon _{2}+\sqrt{%
\left( \varepsilon _{1}-\varepsilon _{2}\right) ^{2}+4\mu ^{2}}\right] , \\
\omega _{2} &=&\frac{1}{2}\left[ \varepsilon _{1}+\varepsilon _{2}-\sqrt{%
\left( \varepsilon _{1}-\varepsilon _{2}\right) ^{2}+4\mu ^{2}}\right] .
\end{eqnarray}%
The unitary matrix can be found as follows,%
\begin{eqnarray}
U &=&\left[
\begin{array}{cc}
v(\omega _{1}), & v(\omega _{2})%
\end{array}%
\right]  \notag \\
&=&\frac{1}{\sqrt{\left( \omega _{1}-\varepsilon _{1}\right) ^{2}+\mu ^{2}}}
\notag \\
&&\times
\begin{bmatrix}
\mu & \omega _{2}-\varepsilon _{2} \\
\omega _{1}-\varepsilon _{1} & \mu%
\end{bmatrix}%
,
\end{eqnarray}%
The diagonalized Hamiltonian is%
\begin{equation}
H=\omega _{1}d_{1}^{\dag }d_{1}+\omega _{2}d_{2}^{\dag }d_{2}.
\end{equation}
\end{solution}

\begin{example}
\begin{eqnarray}
H &=&\varepsilon (c_{1}^{\dag }c_{1}+c_{2}^{\dag }c_{2}+c_{3}^{\dag }c_{3})
\notag \\
&&+\mu (c_{1}^{\dag }c_{2}+c_{2}^{\dag }c_{1}+c_{2}^{\dag }c_{3}+c_{3}^{\dag
}c_{2}).
\end{eqnarray}
\end{example}

\begin{solution}
The coefficient matrix is%
\begin{equation}
\alpha =%
\begin{bmatrix}
\varepsilon & \mu & 0 \\
\mu & \varepsilon & \mu \\
0 & \mu & \varepsilon%
\end{bmatrix}%
.
\end{equation}%
It has three eigenvalues,%
\begin{equation}
\omega _{1}=\varepsilon ,\text{ \ }\omega _{2}=\varepsilon +\sqrt{2}\mu ,%
\text{ \ }\omega _{3}=\varepsilon -\sqrt{2}\mu .
\end{equation}%
The corresponding unitary matrix is
\begin{eqnarray}
U &=&\left[
\begin{array}{ccc}
v(\omega _{1}), & v(\omega _{2}), & v(\omega _{3})%
\end{array}%
\right]  \notag \\
&=&\frac{1}{2}%
\begin{bmatrix}
\sqrt{2} & 1 & 1 \\
0 & \sqrt{2} & -\sqrt{2} \\
-\sqrt{2} & 1 & 1%
\end{bmatrix}%
,
\end{eqnarray}%
It is easy to show that
\begin{equation}
H=\omega _{1}d_{1}^{\dag }d_{1}+\omega _{2}d_{2}^{\dag }d_{2}+\omega
_{3}d_{3}^{\dag }d_{3}.
\end{equation}
\end{solution}

By the way, we note that the examples \ref{Ins3} and \ref{Ins4} can also be
diagonalized using the present method.

\subsection{The pairing Hamiltonian}

As shown above, the Heisenberg equation and dynamic matrix are reducible for
a normal Hamiltonian. There is another reducible case, which we are going to
handle below.

Consider the so-called pairing Hamiltonian,
\begin{equation}
H=\sum_{i,j=1}^{n}(\alpha _{ij}a_{i}^{\dag }a_{j}+\varepsilon
_{ij}b_{i}^{\dag }b_{j}+\gamma _{ij}a_{i}^{\dag }b_{j}^{\dag }+\gamma
_{ji}^{\ast }a_{i}b_{j}),  \label{HamType2}
\end{equation}%
where $a_{i}$ ($a_{i}^{\dag }$) and $b_{i}$ ($b_{i}^{\dag }$) are both the
annihilation (creation) operators of bosons, and
\begin{equation}
\alpha ^{\dag }=\alpha ,\text{ \ }\varepsilon ^{\dag }=\varepsilon ,\text{ \
}\widetilde{\gamma }=\gamma .
\end{equation}%
In this Hamiltonian, the abnormal terms, such as $a_{i}b_{j}$ and $%
a_{i}^{\dag }b_{j}^{\dag }$, appear exactly in pairs, each pair has one $a$%
-boson and one $b$-boson, there are totally $n$ pairs between $a$- and $b$%
-bosons. Simply speaking, the particles of the system are formed perfectly
into boson pairs.

According to Eq. (\ref{Heqb3}), the Heisenberg equation for the Hamiltonian
above has the variables of $a_{i}$ $a_{i}^{\dag }$, $b_{i}$, and $%
b_{i}^{\dag }$ where $i=1,2,\cdots ,n$, hence it has the multiplicity of $4n$%
. It is easy to show that the Heisenberg equation can be reduced to the
equations of motion of the variables of $a_{i}$ and $b_{i}^{\dag }$ ($%
i=1,2,\cdots ,n$),%
\begin{eqnarray}
i\frac{\mathrm{d}}{\mathrm{d}t}a_{i} &=&\alpha _{ij}a_{j}+\gamma
_{ij}b_{j}^{\dag }, \\
i\frac{\mathrm{d}}{\mathrm{d}t}b_{i}^{\dag } &=&-\varepsilon
_{ji}b_{j}-\gamma _{ij}^{\ast }a_{j}.
\end{eqnarray}%
Clearly, that is just of multiplicity $2n$. We shall utilize those $2n$
multiple equations straightforwardly to study the diagonalization problem of
the Hamiltonian of Eq. (\ref{HamType2}). It is equivalent to but will be
simpler than from the primitive equation (\ref{Heqb3}) and the theorem \ref%
{BThm}, as has already been seen from the discussions on the normal
Hamiltonian. In particular, it will bring us a fairly simple algorithm for
the diagonalization of the paring Hamiltonian.

For the sake of convenience, we introduce the new operators $c_{i}$ and $%
c_{i}^{\dag }$ as follows,%
\begin{equation}
c_{i}=b_{i}^{\dag },\text{ \ }c_{i}^{\dag }=b_{i},\text{ \ }i=1,2,\cdots ,n.
\end{equation}%
Accordingly, the commutators will be
\begin{gather}
a\cdot a^{\dag }=I,\text{ \ }a\cdot a=0,\text{ \ }a^{\dag }\cdot a^{\dag }=0,
\\
c\cdot c^{\dag }=-I,\text{ \ }c\cdot c=0,\text{ \ }c^{\dag }\cdot c^{\dag
}=0, \\
a\cdot c=0,\text{ \ }a\cdot c^{\dag }=0,\text{ \ }a^{\dag }\cdot c=0,\text{
\ }a^{\dag }\cdot c^{\dag }=0.
\end{gather}%
In fact, the $c$-particles are just the so-called time-polarized bosons.
Using those new operators, the Hamiltonian of Eq. (\ref{HamType2}) can be
written as
\begin{equation}
H=\psi ^{\dag }M\psi -\mathrm{tr}(\varepsilon ),  \label{HamPairb}
\end{equation}%
where $M$ is the coefficient matrix,
\begin{equation}
M=\left[
\begin{array}{cc}
\alpha & \gamma \\
\gamma ^{\dag } & \varepsilon%
\end{array}%
\right] ,
\end{equation}%
and $\psi $ the field operator,
\begin{equation}
\psi =\left[
\begin{array}{c}
a \\
c%
\end{array}%
\right] ,\text{ \ }\psi ^{\dag }=\left[
\begin{array}{cc}
a^{\dag }, & c^{\dag }%
\end{array}%
\right] .
\end{equation}%
The commutator for $\psi $ is
\begin{equation}
\psi \cdot \psi ^{\dag }=I_{-}.  \label{CommP}
\end{equation}%
Here, it is worth pointing out that there is no involution symmetry for the
field $\psi $ now, which is quite different from that of Eq. (\ref{Psi}). As
can be seen latter, this property will make the diagonalization much easier:
One need not ensure the involution symmetry for the new field any more.

The Heisenberg equation of motion for the field $\psi$ can be derived from
Eqs. (\ref{HamPairb}) and (\ref{CommP}),
\begin{equation}
i\frac{\mathrm{d}}{\mathrm{d}t}\psi=D\psi,
\end{equation}
where $D$ is the dynamic matrix,
\begin{equation}
D=\left[
\begin{array}{cc}
\alpha & \gamma \\
-\gamma^{\dag} & -\varepsilon%
\end{array}
\right] .
\end{equation}
As regards $M$ and $D$, one has
\begin{equation}
D=I_{-}M,  \label{DIM1}
\end{equation}
which is identical to Eq. (\ref{DIM}). We note that the coefficient matrix $%
M $ is Hermitian.

Different from the normal Hamiltonian, the pairing Hamiltonian is not always
BV diagonalizable.

\begin{proposition}
The boson pairing Hamiltonian of Eq. (\ref{HamType2}) is BV diagonalizable
if and only if the dynamic matrix $D$ is physically diagonalizable.
\end{proposition}

\begin{proof}
The sufficiency can be proved as follows.

First, if the dynamic matrix $D$ is BV diagonalizable, then its eigenspaces
will be orthogonal to each other with respect to the metric $I_{-}$. The
proof is the same as that for the lemma \ref{Ortho}, which can be easily
seen by comparing Eq. (\ref{DIM1}) with Eq. (\ref{DIM}).

Second, for every eigenspace of the dynamic matrix $D$, there exists an
orthonormal basis with respect to the metric $I_{-}$. The proof is
completely the same as that for the lemma \ref{OrNormal}.

Now, summing up all the orthonormal bases chosen as above, we obtain an
orthonormal basis for the whole space $%
\mathbb{C}
^{2n}$,
\begin{equation}
v^{\dag}(\omega_{i})I_{-}v(\omega_{j})=\lambda_{i}\delta_{ij},
\end{equation}
where $v(\omega_{i})$ ($1\leq i\leq2n$) are the eigenvectors with $\lambda
_{i}=\pm1$ being the corresponding norms. It follows from Eq. (\ref{DIM1})
that
\begin{equation}
v^{\dag}(\omega_{i})Mv(\omega_{j})=\lambda_{i}\omega_{i}\delta_{ij}.
\end{equation}

By introducing the matrix,
\begin{equation}
U=\left[
\begin{array}{cccc}
v(\omega _{1}), & v(\omega _{2}), & \cdots , & v(\omega _{2n})%
\end{array}%
\right] ,
\end{equation}%
the two equations above can be formulated as
\begin{gather}
U^{\dag }I_{-}U=\mathrm{diag}(\lambda _{1},\lambda _{2},\cdots ,\lambda
_{2n}), \\
U^{\dag }MU=\mathrm{diag}(\lambda _{1}\omega _{1},\lambda _{2}\omega
_{2},\cdots ,\lambda _{2n}\omega _{2n}).
\end{gather}%
Here, it is enough for us to take only into account of the
orthonormalization of the eigenvectors because, as pointed out above, the
field $\psi $ has no involution symmetry. One need not take care of both the
orthonormalization and involution symmetry simultaneously as before,
particularly as in the lemma \ref{ZeroSpace}. Obviously, that brings us much
convenience.

Applying Sylveter's law of inertia \cite{Roman} to the first equation above,
we find that, of the total $2n$ norms ($\lambda _{i}$ with $i=1,2,\cdots ,2n$%
), there must be $n$ positive norms and $n$ negative norms. Upon rearranging
the order of the eigenvectors, $v(\omega _{i}),$ $1\leq i\leq 2n$, within
the matrix $U$, the two equations above can be reformulated as
\begin{gather}
U^{\dag }I_{-}U=I_{-}, \\
U^{\dag }MU=\mathrm{diag}(\omega _{1},\cdots ,\omega _{n},-\omega
_{n+1},\cdots ,-\omega _{2n}).
\end{gather}

Now, defining a new field $\varphi $,%
\begin{equation}
\varphi =U^{-1}\psi ,
\end{equation}%
we have from Eq. (\ref{CommP})
\begin{equation}
\varphi \cdot \varphi ^{\dag }=I_{-}.
\end{equation}%
Accordingly, the Hamiltonian of Eq. (\ref{HamPairb}) can be written as
\begin{equation}
H=\varphi ^{\dag }U^{\dag }MU\varphi -\mathrm{tr}(\varepsilon ).
\end{equation}%
If one expands $\varphi $ as
\begin{equation}
\varphi =\left[
\begin{array}{c}
d \\
e%
\end{array}%
\right] ,\text{ \ }\varphi ^{\dag }=\left[
\begin{array}{cc}
d^{\dag }, & e^{\dag }%
\end{array}%
\right] ,
\end{equation}%
he obtains the commuatation realtions,
\begin{gather}
d\cdot d^{\dag }=I,\text{ \ }d\cdot d=0,\text{ \ }d^{\dag }\cdot d^{\dag }=0,
\\
e\cdot e^{\dag }=-I,\text{ \ }e\cdot e=0,\text{ \ }e^{\dag }\cdot e^{\dag
}=0, \\
d\cdot e=0,\text{ \ }d\cdot e^{\dag }=0,\text{ \ }d^{\dag }\cdot e=0,\text{
\ }d^{\dag }\cdot e^{\dag }=0,
\end{gather}%
and the corresponding Hamiltonian,
\begin{equation}
H=\sum_{i=1}^{n}(\omega _{i}d_{i}^{\dag }d_{i}-\omega _{n+i}e_{i}^{\dag
}e_{i})-\mathrm{tr}\left( \varepsilon \right) .
\end{equation}

Heeding that the $e$-particles are the time-polarized bosons, we need
perform the transformation,
\begin{equation}
f_{i}=e_{i}^{\dag },\text{ \ }f_{i}^{\dag }=e_{i},\text{ \ }i=1,2,\cdots ,n,
\end{equation}%
where the $f$-particles return to the normal bosons. At last, we obtain the
diagonalized Hamiltonian,
\begin{equation}
H=\sum_{i=1}^{n}(\omega _{i}d_{i}^{\dag }d_{i}-\omega _{n+i}f_{i}^{\dag
}f_{i})-\sum_{i=1}^{n}\omega _{n+i}-\mathrm{tr}\left( \varepsilon \right) ,
\end{equation}%
where the operators $d_{i}$ ($d_{i}^{\dag }$) and $f_{i}$ ($f_{i}^{\dag }$) (%
$i=1,2,\cdots ,n$) satisfy the standard commutation relations for bosons,
\begin{gather}
d\cdot d^{\dag }=I,\text{ \ }d\cdot d=0,\text{ \ }d^{\dag }\cdot d^{\dag }=0,
\\
f\cdot f^{\dag }=I,\text{ \ }f\cdot f=0,\text{ \ }f^{\dag }\cdot f^{\dag }=0,
\\
d\cdot f=0,\text{ \ }d\cdot f^{\dag }=0,\text{ \ }d^{\dag }\cdot f=0,\text{
\ }d^{\dag }\cdot f^{\dag }=0.
\end{gather}

The proof for necessity is simply similar to that for the proposition \ref%
{PPS1}.
\end{proof}

\begin{example}
\begin{equation}
H=\varepsilon_{1}c_{1}^{\dag}c_{1}+\varepsilon_{2}c_{2}^{\dag}c_{2}+%
\gamma(c_{1}^{\dag}c_{2}^{\dag}+c_{2}c_{1}),
\end{equation}
where $\gamma>0$.
\end{example}

\begin{solution}
The dynamic matrix is%
\begin{equation}
D=%
\begin{bmatrix}
\varepsilon _{1} & \gamma \\
-\gamma & -\varepsilon _{2}%
\end{bmatrix}%
.
\end{equation}%
Obviously, its characteristic equation is%
\begin{equation}
\omega ^{2}+(\varepsilon _{2}-\varepsilon _{1})\omega +(\gamma
^{2}-\varepsilon _{1}\varepsilon _{2})=0.
\end{equation}

1. If $\left\vert \varepsilon _{1}+\varepsilon _{2}\right\vert <2\gamma $,
there are two imaginary eigenvalues,%
\begin{equation}
\omega =\frac{1}{2}\left[ \varepsilon _{1}-\varepsilon _{2}\pm i\sqrt{%
4\gamma ^{2}-\left( \varepsilon _{1}+\varepsilon _{2}\right) ^{2}}\right] .
\end{equation}%
The dynamic matrix $D$ is not physically diagonalizable.

2. If $\left\vert \varepsilon _{1}+\varepsilon _{2}\right\vert =2\gamma $,
there is only one real eigenvalue,%
\begin{equation}
\omega =\frac{1}{2}\left( \varepsilon _{1}-\varepsilon _{2}\right) .
\end{equation}%
It is easy to show that $\omega $ has only one eigenvector. The dynamic
matrix $D$ is not physically diagonalizable.

3. If $\left\vert \varepsilon _{1}+\varepsilon _{2}\right\vert >2\gamma $,
there are two real eigenvalues,%
\begin{eqnarray}
\omega _{1} &=&\frac{1}{2}\left[ \varepsilon _{1}-\varepsilon _{2}+\sqrt{%
\left( \varepsilon _{1}+\varepsilon _{2}\right) ^{2}-4\gamma ^{2}}\right] ,
\\
\omega _{2} &=&\frac{1}{2}\left[ \varepsilon _{1}-\varepsilon _{2}-\sqrt{%
\left( \varepsilon _{1}+\varepsilon _{2}\right) ^{2}-4\gamma ^{2}}\right] .
\end{eqnarray}%
The dynamic matrix $D$ is thus physically diagonalizable. Meanwhile, the
Hamiltonian $H$ can be BV diagonalized, the corresponding transformation
matrix is%
\begin{equation}
U=\left\{
\begin{array}{ll}
\left[ v(\omega _{1}),v(\omega _{2})\right] , & \varepsilon _{1}+\varepsilon
_{2}>2\gamma \\
\left[ v(\omega _{2}),v(\omega _{1})\right] , & \varepsilon _{1}+\varepsilon
_{2}<-2\gamma ,%
\end{array}%
\right.
\end{equation}%
where%
\begin{eqnarray}
v(\omega _{1}) &=&\frac{1}{\sqrt{\gamma ^{2}-\left( \omega _{1}-\varepsilon
_{1}\right) ^{2}}}%
\begin{bmatrix}
\gamma \\
\omega _{1}-\varepsilon _{1}%
\end{bmatrix}%
, \\
v(\omega _{2}) &=&\frac{1}{\sqrt{\gamma ^{2}-\left( \omega _{1}-\varepsilon
_{1}\right) ^{2}}}%
\begin{bmatrix}
\omega _{2}+\varepsilon _{2} \\
-\gamma%
\end{bmatrix}%
.
\end{eqnarray}%
It is easy to show%
\begin{equation}
U^{\dag }I_{-}U=I_{-},
\end{equation}%
\begin{equation}
U^{\dag }MU=\left\{
\begin{array}{ll}
\mathrm{diag}(\omega _{1},-\omega _{2}), & \varepsilon _{1}+\varepsilon
_{2}>2\gamma \\
\mathrm{diag}(\omega _{2},-\omega _{1}), & \varepsilon _{1}+\varepsilon
_{2}<-2\gamma ,%
\end{array}%
\right.
\end{equation}%
\begin{equation}
H=\left\{
\begin{array}{ll}
\omega _{1}d_{1}^{\dag }d_{1}-\omega _{2}d_{2}^{\dag }d_{2}-\omega
_{2}-\varepsilon _{2}, & \varepsilon _{1}+\varepsilon _{2}>2\gamma \\
\omega _{2}d_{1}^{\dag }d_{1}-\omega _{1}d_{2}^{\dag }d_{2}-\omega
_{1}-\varepsilon _{2}, & \varepsilon _{1}+\varepsilon _{2}<-2\gamma .%
\end{array}%
\right.
\end{equation}
\end{solution}

Obviously, this example is an extension of the example \ref{Ins2}. Besides,
one can see that the present method is much more convenient than that
adopted by the example \ref{Ins2}.

\begin{example}
\begin{equation}
H=\sum_{i=1}^{2}\varepsilon(a_{i}^{\dag}a_{i}+b_{i}^{\dag}b_{i})+\sum
_{i,j=1}^{2}\gamma(a_{i}^{\dag}b_{j}^{\dag}+a_{i}b_{j}),
\end{equation}
where $\varepsilon>0$, and $\gamma>0$.
\end{example}

\begin{solution}
The dynamic matrix is%
\begin{equation}
D=%
\begin{bmatrix}
\varepsilon & 0 & \gamma & \gamma \\
0 & \varepsilon & \gamma & \gamma \\
-\gamma & -\gamma & -\varepsilon & 0 \\
-\gamma & -\gamma & 0 & -\varepsilon%
\end{bmatrix}%
.
\end{equation}%
There are four eigenvalues,%
\begin{equation}
\omega _{1}=\varepsilon ,\text{ \ }\omega _{2}=-\varepsilon ,\text{ \ }%
\omega _{3}=\omega _{+},\text{ \ }\omega _{4}=\omega _{-},
\end{equation}%
where%
\begin{equation}
\omega _{\pm }=\left\{
\begin{array}{ll}
\pm i\sqrt{4\gamma ^{2}-\varepsilon ^{2}}, & \varepsilon <2\gamma \\
0,\text{ } & \varepsilon =2\gamma \\
\pm \sqrt{\varepsilon ^{2}-4\gamma ^{2}}, & \varepsilon >2\gamma .%
\end{array}%
\right.
\end{equation}

If $\varepsilon<2\gamma$, $D$ has two imaginary eigenvalues, $H$ is not BV
diagonalizable.

If $\varepsilon=2\gamma$, the zero eigenvalue has only one eigenvector, $H$
is not BV diagonalizable.

If $\varepsilon >2\gamma $, $D$ is physically diagonalizable, $H$ can be BV
diagonalized. The corresponding transformation matrix is
\begin{equation}
U=\left[
\begin{array}{cccc}
v(\omega _{1}), & v(\omega _{3}), & v(\omega _{2}), & v(\omega _{4})%
\end{array}%
\right] ,
\end{equation}%
where%
\begin{equation}
v(\omega _{1})=\left[
\begin{array}{c}
\frac{1}{\sqrt{2}} \\
-\frac{1}{\sqrt{2}} \\
0 \\
0%
\end{array}%
\right] ,\text{ \ }v(\omega _{3})=\left[
\begin{array}{c}
\frac{2\gamma }{\sqrt{8\gamma ^{2}-2\left( \omega _{+}-\varepsilon \right)
^{2}}} \\
\frac{2\gamma }{\sqrt{8\gamma ^{2}-2\left( \omega _{+}-\varepsilon \right)
^{2}}} \\
\frac{\omega _{+}-\varepsilon }{\sqrt{8\gamma ^{2}-2\left( \omega
_{+}-\varepsilon \right) ^{2}}} \\
\frac{\omega _{+}-\varepsilon }{\sqrt{8\gamma ^{2}-2\left( \omega
_{+}-\varepsilon \right) ^{2}}}%
\end{array}%
\right]
\end{equation}%
\begin{equation}
v(\omega _{2})=\left[
\begin{array}{c}
0 \\
0 \\
\frac{1}{\sqrt{2}} \\
-\frac{1}{\sqrt{2}}%
\end{array}%
\right] ,\text{ \ }v(\omega _{4})=\left[
\begin{array}{c}
\frac{\omega _{-}+\varepsilon }{\sqrt{8\gamma ^{2}-2\left( \omega
_{-}+\varepsilon \right) ^{2}}} \\
\frac{\omega _{-}+\varepsilon }{\sqrt{8\gamma ^{2}-2\left( \omega
_{-}+\varepsilon \right) ^{2}}} \\
-\frac{2\gamma }{\sqrt{8\gamma ^{2}-2\left( \omega _{-}+\varepsilon \right)
^{2}}} \\
-\frac{2\gamma }{\sqrt{8\gamma ^{2}-2\left( \omega _{-}+\varepsilon \right)
^{2}}}%
\end{array}%
\right]
\end{equation}%
The diagonalized Hamiltonian is%
\begin{eqnarray}
H &=&\omega _{1}d_{1}^{\dag }d_{1}+\omega _{3}d_{2}^{\dag }d_{2}-\omega
_{2}d_{3}^{\dag }d_{3}-\omega _{4}d_{4}^{\dag }d_{4}  \notag \\
&&-\omega _{2}-\omega _{4}-2\varepsilon .
\end{eqnarray}
\end{solution}

The propositions and algorithms developed in this section can be applied to
statistical as well as condensed-matter physics \cite%
{Bogoliubov1,Bogoliubov4,Fetter,Wagner,Leggett,Madelung,White}.

\section{Application to Fermi Systems \label{AFS}}

As in the preceding subsection, we shall also concentrate ourselves on the
normal and pairing Hamiltonians. They represent the problems which we
encounter most frequently in practice.

\subsection{The normal Hamiltonian \label{NHF}}

The normal Hamiltonian reads
\begin{equation}
H=\sum_{i,j=1}^{n}\alpha_{ij}c_{i}^{\dag}c_{j}.  \label{FHamType1}
\end{equation}
The proposition \ref{BPPType1} can be easily transplanted to the present
case.

\begin{proposition}
\label{PPSNHF} A normal Hamiltonian of fermions can be BV diagonalized by
the unitary transformation generated by its coefficient matrix.
\end{proposition}

That is also because the Heisenberg equation for the Hamiltonian of Eq. (\ref%
{FHamType1}) is reducible. It can be reduced as%
\begin{equation}
i\frac{\mathrm{d}}{\mathrm{d}t}c=\alpha c.
\end{equation}%
Since $\alpha $ is Hermitian, this equation of motion can generate a unitary
transformation $U$ for the field $c$,%
\begin{equation}
c=Ud,
\end{equation}%
where%
\begin{gather}
U^{\dag }U=UU^{\dag }=I, \\
U^{\dag }\alpha U=\mathrm{diag}(\omega _{1},\omega _{2},\cdots ,\omega _{n}).
\end{gather}%
The $d$ represents the new field, it is easy to show that $d$ is a standard
fermionic field,%
\begin{equation}
d\cdot d^{\dag }=I,\text{ \ }d\cdot d=0,\text{ \ }d^{\dag }\cdot d^{\dag }=0.
\end{equation}%
Accordingly,%
\begin{equation}
H=d^{\dag }U^{\dag }\alpha Ud=\sum_{i=1}^{n}\omega _{i}d_{i}^{\dag }d_{j}.
\end{equation}%
Besides, a particle-hole transformation will be needed if some eigenenergies
are negative.

To sum up, a normal Hamiltonian can be BV diagonalized by the unitary
transformation generated by its coefficient matrix no matter whether the
system is bosonic or fermionic.

\begin{example}
\begin{equation}
H=\varepsilon(c_{1}^{\dag}c_{1}+c_{2}^{\dag}c_{2})+\mu(c_{1}^{%
\dag}c_{2}+c_{2}^{\dag}c_{1}),
\end{equation}
where $\mu>0$.
\end{example}

\begin{solution}
The coefficient matrix is%
\begin{equation}
\alpha =%
\begin{bmatrix}
\varepsilon & \mu \\
\mu & \varepsilon%
\end{bmatrix}%
.
\end{equation}%
It has two eigenvalues,%
\begin{equation}
\omega _{1}=\varepsilon +\mu ,\text{ \ }\omega _{2}=\varepsilon -\mu .
\end{equation}%
The unitary matrix can be easily found,%
\begin{equation}
U=\left[
\begin{array}{cc}
v(\omega _{1}), & v(\omega _{2})%
\end{array}%
\right] =\frac{1}{\sqrt{2}}%
\begin{bmatrix}
1 & 1 \\
1 & -1%
\end{bmatrix}%
.
\end{equation}

1. If $\varepsilon \geq \mu $, $\omega _{1}>0$ and\ $\omega _{2}\geq 0$,%
\begin{gather}
\left[
\begin{array}{c}
c_{1} \\
c_{2}%
\end{array}%
\right] =\frac{1}{\sqrt{2}}%
\begin{bmatrix}
1 & 1 \\
1 & -1%
\end{bmatrix}%
\left[
\begin{array}{c}
d_{1} \\
d_{2}%
\end{array}%
\right] , \\[0.05in]
H=\omega _{1}d_{1}^{\dag }d_{1}+\omega _{2}d_{2}^{\dag }d_{2}.
\end{gather}

2. If $-\mu \leq \varepsilon <\mu $, $\omega _{1}\geq 0$ and\ $\omega _{2}<0$%
,
\begin{gather}
\left[
\begin{array}{c}
c_{1} \\
c_{2}%
\end{array}%
\right] =\frac{1}{\sqrt{2}}\left[
\begin{array}{cc}
1 & 1 \\
1 & -1%
\end{array}%
\right] \left[
\begin{array}{c}
d_{1} \\
d_{2}^{\dag }%
\end{array}%
\right] , \\[0.05in]
H=\omega _{1}d_{1}^{\dag }d_{1}-\omega _{2}d_{2}^{\dag }d_{2}+\omega _{2}.
\end{gather}

3. If $\varepsilon <-\mu $, $\omega _{1}<0$ and\ $\omega _{2}<0$,
\begin{gather}
\left[
\begin{array}{c}
c_{1} \\
c_{2}%
\end{array}%
\right] =\frac{1}{\sqrt{2}}\left[
\begin{array}{cc}
1 & 1 \\
1 & -1%
\end{array}%
\right] \left[
\begin{array}{c}
d_{1}^{\dag } \\
d_{2}^{\dag }%
\end{array}%
\right] , \\[0.05in]
H=-\omega _{1}d_{1}^{\dag }d_{1}-\omega _{2}d_{2}^{\dag }d_{2}+\omega
_{1}+\omega _{2}.
\end{gather}%
Here, a particle-hole transformation is performed to the $d_{2}$- particles
if $-\mu \leq \varepsilon <\mu $, and to both the $d_{1}$- and $d_{2}$-
particles if $\varepsilon <-\mu $.
\end{solution}

\subsection{The pairing Hamiltonian}

The pairing Hamiltonian reads
\begin{equation}
H=\sum_{i,j=1}^{n}(\alpha _{ij}a_{i}^{\dag }a_{j}+\varepsilon
_{ij}b_{i}^{\dag }b_{j}+\gamma _{ij}a_{i}b_{j}+\gamma _{ji}^{\ast
}a_{i}^{\dag }b_{j}^{\dag }),  \label{FHamType2}
\end{equation}%
where%
\begin{equation}
\alpha ^{\dag }=\alpha ,\text{ \ }\varepsilon ^{\dag }=\varepsilon ,\text{ \
}\widetilde{\gamma }=-\gamma .
\end{equation}

Following the Bose case, let us introduce the new operators $c_{i}$ and $%
c_{i}^{\dag }$ as%
\begin{equation}
c_{i}=b_{i}^{\dag },\text{ \ }c_{i}^{\dag }=b_{i},\text{ \ }i=1,2,\cdots ,n.
\label{PtoH}
\end{equation}%
The new anticommutators will be
\begin{gather}
a\cdot a^{\dag }=I,\text{ \ }a\cdot a=0,\text{ \ }a^{\dag }\cdot a^{\dag }=0,
\\
c\cdot c^{\dag }=I,\text{ \ }c\cdot c=0,\text{ \ }c^{\dag }\cdot c^{\dag }=0,
\\
a\cdot c=0,\text{ \ }a\cdot c^{\dag }=0,\text{ \ }a^{\dag }\cdot c=0,\text{
\ }a^{\dag }\cdot c^{\dag }=0.
\end{gather}%
Obviously, they are still standard, which is rather different from the Bose
case. In terms of these new operators, Eq. (\ref{FHamType2}) can be
expressed as
\begin{equation}
H=\psi ^{\dag }M\psi +\mathrm{tr}(\varepsilon ),  \label{HamPairf}
\end{equation}%
where $M$ is the coefficient matrix,
\begin{equation}
M=\left[
\begin{array}{cc}
\alpha & \gamma \\
\gamma ^{\dag } & -\varepsilon%
\end{array}%
\right] ,
\end{equation}%
and $\psi $ the field operator,
\begin{equation}
\psi =\left[
\begin{array}{c}
a \\
c%
\end{array}%
\right] ,\text{ \ }\psi ^{\dag }=\left[
\begin{array}{cc}
a^{\dag }, & c^{\dag }%
\end{array}%
\right] .
\end{equation}%
It is evident that%
\begin{gather}
M=M^{\dag }, \\
\psi \cdot \psi ^{\dag }=I_{+}.
\end{gather}%
Namely, $M$ is Hermitian, and $\psi $ is a standard fermionic field. This
means that Eq. (\ref{HamPairf}) is, in fact, a normal Hamiltonian. Hence, we
obtain the proposition.

\begin{proposition}
A pairing Hamiltonian of fermions can be first transformed into a normal
Hamiltonian, and then BV diagonalized by the unitary transformation
generated by the corresponding coefficient matrix.
\end{proposition}

As well known, Eq. (\ref{PtoH}) represents a particle-hole transformation in
physics. So, the proposition says actually that a pairing Hamiltonian of
fermions can be transformed into a normal Hamiltonian by a particle-hole
transformation.

\begin{example}
\begin{equation}
H=\varepsilon_{1}c_{1}^{\dag}c_{1}+\varepsilon_{2}c_{2}^{\dag}c_{2}+%
\gamma(c_{1}^{\dag}c_{2}^{\dag}+c_{2}c_{1}),
\end{equation}
where $\gamma>0$.
\end{example}

\begin{solution}
\begin{description}
\item The coefficient matrix is%
\begin{equation}
M=%
\begin{bmatrix}
\varepsilon _{1} & \gamma \\
\gamma & -\varepsilon _{2}%
\end{bmatrix}%
.
\end{equation}%
It has two eigenvalues,%
\begin{eqnarray}
\omega _{1} &=&\frac{1}{2}\left[ \varepsilon _{1}-\varepsilon _{2}+\sqrt{%
\left( \varepsilon _{1}+\varepsilon _{2}\right) ^{2}+4\gamma ^{2}}\right] ,
\\
\omega _{2} &=&\frac{1}{2}\left[ \varepsilon _{1}-\varepsilon _{2}-\sqrt{%
\left( \varepsilon _{1}+\varepsilon _{2}\right) ^{2}+4\gamma ^{2}}\right] ,
\end{eqnarray}%
and generates a unitary transformation matrix,%
\begin{eqnarray}
U &=&\left[
\begin{array}{cc}
v(\omega _{1}), & v(\omega _{2})%
\end{array}%
\right]  \notag \\
&=&\frac{1}{\sqrt{\left( \omega _{1}-\varepsilon _{1}\right) ^{2}+\gamma ^{2}%
}}  \notag \\
&&\times
\begin{bmatrix}
\gamma & \omega _{2}+\varepsilon _{2} \\
\omega _{1}-\varepsilon _{1} & \gamma%
\end{bmatrix}%
.
\end{eqnarray}
\end{description}

1. If $\gamma ^{2}\leq -\varepsilon _{1}\varepsilon $ and $\varepsilon
_{1}>0 $, the diagonalized Hamiltonian has the form,
\begin{equation}
H=\omega _{1}d_{1}^{\dag }d_{1}+\omega _{2}d_{2}^{\dag }d_{2}+\varepsilon
_{2},
\end{equation}%
where%
\begin{equation}
\left[
\begin{array}{c}
d_{1} \\
d_{2}%
\end{array}%
\right] =U^{-1}\left[
\begin{array}{c}
c_{1} \\
c_{2}^{\dag }%
\end{array}%
\right] .
\end{equation}

2. If $\gamma ^{2}\leq -\varepsilon _{1}\varepsilon _{2}$ and $\varepsilon
_{1}<0$, the diagonalized Hamiltonian has the form,
\begin{equation}
H=-\omega _{1}d_{1}^{\dag }d_{1}-\omega _{2}d_{2}^{\dag }d_{2}+\omega
_{1}+\omega _{2}+\varepsilon _{2},
\end{equation}%
where%
\begin{equation}
\left[
\begin{array}{c}
d_{1}^{\dag } \\
d_{2}^{\dag }%
\end{array}%
\right] =U^{-1}\left[
\begin{array}{c}
c_{1} \\
c_{2}^{\dag }%
\end{array}%
\right] .
\end{equation}

3. If $\gamma ^{2}>-\varepsilon _{1}\varepsilon _{2}$, the diagonalized
Hamiltonian has the form,
\begin{equation}
H=\omega _{1}d_{1}^{\dag }d_{1}-\omega _{2}d_{2}^{\dag }d_{2}+\omega
_{2}+\varepsilon _{2},
\end{equation}%
where%
\begin{equation}
\left[
\begin{array}{c}
d_{1} \\
d_{2}^{\dag }%
\end{array}%
\right] =U^{-1}\left[
\begin{array}{c}
c_{1} \\
c_{2}^{\dag }%
\end{array}%
\right] .
\end{equation}
\end{solution}

Obviously, this example is an extension of the example \ref{Insf1}. One can
see that the present method is much simpler than that used by the example %
\ref{Insf1}.

\begin{example}
\begin{equation}
H=\sum_{i=1}^{2}\varepsilon(a_{i}^{\dag}a_{i}+b_{i}^{\dag}b_{i})+\sum
_{i,j=1}^{2}\gamma(a_{i}^{\dag}b_{j}^{\dag}+a_{i}b_{j}),
\end{equation}
where $\varepsilon>0$, and $\gamma>0$.
\end{example}

\begin{solution}
The coefficient matrix is
\begin{equation}
M=%
\begin{bmatrix}
\varepsilon & 0 & \gamma & \gamma \\
0 & \varepsilon & \gamma & \gamma \\
\gamma & \gamma & -\varepsilon & 0 \\
\gamma & \gamma & 0 & -\varepsilon%
\end{bmatrix}%
.
\end{equation}%
It has four eigenvalues,%
\begin{gather}
\omega _{1}=\varepsilon ,\text{ \ }\omega _{2}=-\varepsilon , \\
\omega _{3}=\sqrt{\varepsilon ^{2}+4\gamma ^{2}},\text{ \ }\omega _{4}=-%
\sqrt{\varepsilon ^{2}+4\gamma ^{2}}.
\end{gather}%
The diagonalized Hamiltonian is%
\begin{eqnarray}
H &=&\omega _{1}d_{1}^{\dag }d_{1}-\omega _{2}d_{2}^{\dag }d_{2}+\omega
_{3}d_{3}^{\dag }d_{3}-\omega _{4}d_{4}^{\dag }d_{4}  \notag \\
&&+\omega _{2}+\omega _{4}+2\varepsilon .
\end{eqnarray}%
The corresponding BV transformation has the form,
\begin{widetext}%
\begin{equation}
\left[
\begin{array}{c}
a_{1}\vspace*{5pt} \\
a_{2}\vspace*{5pt} \\
b_{1}^{\dag }\vspace*{5pt} \\
b_{2}^{\dag }\vspace*{5pt}%
\end{array}%
\right] =\left[
\begin{array}{cccc}
\frac{1}{\sqrt{2}} & 0 & \frac{2\gamma }{\sqrt{8\gamma ^{2}+2\left( \omega
_{3}-\varepsilon \right) ^{2}}} & \frac{\omega _{4}+\varepsilon }{\sqrt{%
8\gamma ^{2}+2\left( \omega _{4}+\varepsilon \right) ^{2}}} \\
-\frac{1}{\sqrt{2}} & 0 & \frac{2\gamma }{\sqrt{8\gamma ^{2}+2\left( \omega
_{3}-\varepsilon \right) ^{2}}} & \frac{\omega _{4}+\varepsilon }{\sqrt{%
8\gamma ^{2}+2\left( \omega _{4}+\varepsilon \right) ^{2}}} \\
0 & \frac{1}{\sqrt{2}} & \frac{\omega _{3}-\varepsilon }{\sqrt{8\gamma
^{2}+2\left( \omega _{3}-\varepsilon \right) ^{2}}} & \frac{2\gamma }{\sqrt{%
8\gamma ^{2}+2\left( \omega _{4}+\varepsilon \right) ^{2}}} \\
0 & -\frac{1}{\sqrt{2}} & \frac{\omega _{3}-\varepsilon }{\sqrt{8\gamma
^{2}+2\left( \omega _{3}-\varepsilon \right) ^{2}}} & \frac{2\gamma }{\sqrt{%
8\gamma ^{2}+2\left( \omega _{4}+\varepsilon \right) ^{2}}}%
\end{array}%
\right] \left[
\begin{array}{c}
d_{1}\vspace*{5pt} \\
d_{2}^{\dag }\vspace*{5pt} \\
d_{3}\vspace*{5pt} \\
d_{4}^{\dag }\vspace*{5pt}%
\end{array}%
\right] .
\end{equation}%
\end{widetext}%
\end{solution}

The propositions and algorithms developed in this section can be applied to
statistical physics \cite%
{Bogoliubov2,Bogoliubov3,Bogoliubov4,Valatin1,Valatin2,Fetter},
condensed-matter physics \cite{Wagner,Mahan,Altland,Doniach}, and nuclear
physics \cite{Fetter,Ring}.

\section{Generalized\ Bogoliubov-Valatin Transformation \label{GBVT}}

Historically, Dirac \cite{Dirac} and later Shr\"{o}dinger \cite%
{Schrodinger1,Schrodinger2,Schrodinger3,Infeld} found that a linear harmonic
oscillator can be diagonalized with regard to the bosonic creation and
annihilation operators,%
\begin{eqnarray}
H &=&\frac{p^{2}}{2m}+\frac{1}{2}m\omega ^{2}q^{2}  \notag \\
&=&\omega a^{\dag }a+\frac{1}{2}\omega ,  \label{LHH}
\end{eqnarray}%
where $m$ and $\omega $ are the mass and (angular) frequency of the
oscillator, respectively. The $q$ and $p$ are the coordinate and momentum
operators. They are both Hermitian, and satisfy the canonical commutation
rule,%
\begin{equation}
\lbrack p,q]=-i.
\end{equation}%
The $a$ and $a^{\dag }$ are the annihilation and creation operators. They
are Hermitian conjugates of each other, and satisfy the bosonic commutation
rule,%
\begin{equation}
\lbrack a,a^{\dag }]=1.
\end{equation}%
Both $a$ and $a^{\dag }$ are the linear functions of $q$ and $p$,
\begin{eqnarray}
a &=&\sqrt{\frac{m\omega }{2}}\left( q+i\frac{p}{m\omega }\right) ,
\label{DT1} \\
a^{\dag } &=&\sqrt{\frac{m\omega }{2}}\left( q-i\frac{p}{m\omega }\right) ,
\label{DT2}
\end{eqnarray}%
and vice versa. For convenience, we shall call such a transformation Dirac
transformation, and call the corresponding diagonalization Dirac
diagonalization.

In this section, we shall study the necessary and sufficient conditions for
Dirac diagonalization. We find that Dirac diagonalization is actually a
generalization of BV diagonalization.

\subsection{Theory of Dirac diagonalization}

Consider the Hamiltonian that is quadratic in coordinates and momenta,%
\begin{eqnarray}
H &=&\sum_{i,j=1}^{n}(\frac{1}{2}\mu _{ij}p_{i}p_{j}+\frac{1}{2}\kappa
_{ij}q_{i}q_{j}+\frac{1}{2}\gamma _{ij}p_{i}q_{j}+\frac{1}{2}\gamma
_{ji}q_{i}p_{j})  \notag \\
&=&\frac{1}{2}p^{\dag }\mu p+\frac{1}{2}q^{\dag }\kappa q+\frac{1}{2}p^{\dag
}\gamma q+\frac{1}{2}q^{\dag }\gamma ^{\dag }p,  \label{HamC0}
\end{eqnarray}%
where the coefficients are all real, $\kappa _{il}\in
\mathbb{R}
$, $\mu _{ij}\in
\mathbb{R}
$, $\gamma _{ij}\in
\mathbb{R}
$, and%
\begin{equation}
\mu =\mu ^{\dag },\text{ \ }\kappa =\kappa ^{\dag }.
\end{equation}%
The coordinates and momenta satisfy the canonical commutation relations,%
\begin{equation}
\lbrack p_{i},q_{j}]=-i\delta _{ij},\text{ \ }[p_{i},p_{j}]=0,\text{ \ }%
[q_{i},q_{j}]=0,
\end{equation}%
or equivalently,%
\begin{equation}
p\cdot q^{\dag }=-iI,\text{ \ }p\cdot p^{\dag }=0,\text{ \ }q\cdot q^{\dag
}=0.  \label{CCR}
\end{equation}

Let us put%
\begin{equation}
\phi=\left[
\begin{array}{c}
p \\
q%
\end{array}
\right] ,\text{ \ }\phi^{\dag}=\left[
\begin{array}{cc}
p^{\dag}, & q^{\dag}%
\end{array}
\right] .
\end{equation}
The Hamiltonian can written more compactly as%
\begin{equation}
H=\frac{1}{2}\phi^{\dag}M\phi,  \label{CHam}
\end{equation}
where\ $M$ is the coefficient matrix,%
\begin{equation}
M=\left[
\begin{array}{cc}
\mu & \gamma \\
\gamma^{\dag} & \kappa%
\end{array}
\right] ,
\end{equation}
which is Hermitian. Equation (\ref{CCR}) becomes%
\begin{equation}
\phi\cdot\phi^{\dag}=\Sigma_{y},  \label{NCCR}
\end{equation}
where%
\begin{equation}
\Sigma_{y}=\left[
\begin{array}{cc}
0 & -iI \\
iI & 0%
\end{array}
\right] .  \label{Sigmay}
\end{equation}
Obviously, $\Sigma_{y}$ is a Hermitian matrix.

The Heisenberg equation of motion for the field $\phi$ can be derived from
Eqs. (\ref{CHam}) and (\ref{NCCR}),%
\begin{equation}
i\frac{\mathrm{d}}{\mathrm{d}t}\phi=D\phi,  \label{lphi}
\end{equation}
where $D$ is the dynamic matrix,
\begin{equation}
D=\Sigma_{y}M.  \label{DSyM}
\end{equation}
Generally, $D$ is not Hermitian. The above equation gives the relationship
between the dynamic matrx $D$ and the coefficient matrix $M$, it is the
counterpart of Eq. (\ref{DIM}).

The eigenvalue equation for Eq. (\ref{lphi}) is,%
\begin{equation}
\omega v(\omega)=Dv(\omega).
\end{equation}
It can be rewritten as%
\begin{equation}
i\omega v(\omega)=JMv(\omega),  \label{JMV}
\end{equation}
where $J$ is the unit symplectic matrix given in Eq. (\ref{JSy}). Paying
attention to the fact that $JM$ is a real matrix, one realizes immediatly
that the lemma \ref{Lmm1} also holds now.

\begin{lemma}
If $\omega$ is an eigenvalue of the dynamic matrix $D$, then $-\omega^{\ast}$
will also be an eigenvalue of $D$.
\end{lemma}

Nevertheless, the lemma \ref{Lmm2} should be modified as follows.

\begin{lemma}
If $v(\omega)$ is an eigenvector belonging to the eigenvalue $\omega$ of the
dynamic matrix $D$, then its complex conjugate $v^{\ast}(\omega)$ will be an
eigenvector belonging to the eigenvalue $-\omega^{\ast}$.
\end{lemma}

These two lemmas show that the dynamic mode pair appears now in the form of $%
\{\omega,v(\omega)\}$ and $\{-\omega^{\ast},v^{\ast}(\omega)\}$ where the
two eigenvectors are complex conjugates of each other.

Now, consider the Dirac transformation,%
\begin{equation}
\phi=T\psi,
\end{equation}
where $\psi$ represents the bosonic field given in Eq. (\ref{Psi}).
Evidently, the Heisenberg equation of $\psi$ is still linear and homogeneous,%
\begin{equation}
i\frac{\mathrm{d}}{\mathrm{d}t}\psi=D_{1}\psi,
\end{equation}
where $D_{1}$ is the dynamic matrix for $\psi$. Following the proof for the
lemma \ref{Similar}, we obtain the lemma below.

\begin{lemma}
\label{SimilarD} Under a Dirac transformation, the two dynamic matrices
respectively for the old and new fields will be similar to each other.
\end{lemma}

It is evident that $D_{1}$ is a real diagonal matrix if the Hamiltonian of
Eq. (\ref{CHam}) has been diagonalized Diracianly. As a result, the
proposition \ref{PPS1} holds for Dirac diagonalization, too.

\begin{proposition}
\label{PPD} If a Hamiltonian quadratic in coordinates and momenta can be
Diracianly diagonalized, its dynamic matrix is physically diagonalizable.
\end{proposition}

Also, it can be readily confirmed that the lemma \ref{Lmm5} holds for the
present case.

\begin{lemma}
\label{Lmm5D} If the dynamic matrix $D$ is physically diagonalizable, then,
for each pair of nonzero eigenvalues, i.e., $(\omega,-\omega)$ with $%
\omega\neq0$, they have the same degeneracy. Namely, their eigenspace have
the same dimension. Especially, their bases can be chosen as%
\begin{equation}
v_{l}(-\omega)=v_{l}^{\ast}(\omega),\text{ \ }l=1,2,\cdots,m,  \label{VPD1}
\end{equation}
where $m$ ($m\in%
\mathbb{N}
$) is the dimension of the eigenspace of $\omega$.
\end{lemma}

But the lemma \ref{Lmm6} needs a few modifications.

\begin{lemma}
\label{Lmm6D} If the dynamic matrix $D$ is physically diagonalizable and has
zero eigenvalue, the eigenspace of zero eigenvalue is even dimensional. In
particular, its basis vectors can be chosen and grouped as%
\begin{equation}
v_{m+l}(0)=v_{l}^{\ast}(0),\text{ \ }l=1,2,\cdots,m,  \label{VPD2}
\end{equation}
where $2m$ ($m\in%
\mathbb{N}
$) is the dimension of the eigenspace of zero eigenvalue.
\end{lemma}

\begin{proof}
It just needs to prove the second point.

When $\omega =0$, Eq. (\ref{JMV}) becomes%
\begin{equation}
JMv(0)=0.
\end{equation}%
As pointed out above, $JM$ is a real matrix, its eigenvectors can be chosen
naturally as real vectors. In other words, there exists a real basis for the
eigenspace of zero eigenvalue,
\begin{equation}
w_{l}(0)=w_{l}^{\ast }(0),\text{ \ }l=1,2,\cdots ,2m.
\end{equation}%
Let us put%
\begin{eqnarray}
v_{l}(0) &=&w_{l}(0)+iw_{m+l}(0), \\
v_{m+l}(0) &=&w_{l}(0)-iw_{m+l}(0),
\end{eqnarray}%
where $l=1,2,\cdots ,m$. They are obviously a basis that is in accordance
with Eq. (\ref{VPD2}).
\end{proof}

Instead of $I_{-}$, we can here introduce a sesquilinear form using $%
\Sigma_{y}$, it is also a nonsingular metric. Notice the similarity between
Eqs. (\ref{DSyM}) and (\ref{DIM}). We can transplant the lemmas \ref{Ortho}
and \ref{OrNormal} to the present case.

\begin{lemma}
\label{Lmm7D} If the dynamic matrix $D$ is physically diagonalizable, its
eigenspaces will be orthogonal to each other with respect to the metric $%
\Sigma_{y}$.
\end{lemma}

\begin{lemma}
\label{Lmm8D} If the dynamic matrix $D$ is physically diagonalizable, then,
for each eigenspace of $D$, there exists an orthonormal basis with respect
to the metric $\Sigma _{y}$.
\end{lemma}

As to the lemma \ref{ZeroSpace}, it needs some modifications.

\begin{lemma}
\label{Lmm9D} If the dynamic matrix $D$ is physically diagonalizable and has
zero eigenvalue, there exists such an orthonormal basis for the eigenspace
of zero eigenvalue that%
\begin{equation}
v_{m+l}(0)=v_{l}^{\ast}(0),\text{ \ }l=1,2,\cdots,m,
\end{equation}
where $2m$ ($m\in%
\mathbb{N}
$) is the dimension of the eigenspace of zero eigenvalue.
\end{lemma}

\begin{proof}
According to the lemma \ref{Lmm6D}, the eigenspace $V_{0}$ of zero
eigenvalue has a basis,%
\begin{equation}
v_{m+l}(0)=v_{l}^{\ast}(0),\text{ \ }l=1,2,\cdots,m,
\end{equation}
where $\dim(V_{0})=2m$ ($m\in%
\mathbb{N}
$).

When $m=1$, $v_{1}(0)$ must be nonisotropic,
\begin{equation}
v_{1}^{\dag }(0)\Sigma _{y}v_{1}(0)\neq 0.
\end{equation}%
Otherwise, one has%
\begin{equation}
v_{1}^{\dag }(0)\Sigma _{y}v_{1}(0)=0,\text{ \ }v_{2}^{\dag }(0)\Sigma
_{y}v_{2}(0)=0.  \label{VV001C}
\end{equation}%
In addition,%
\begin{eqnarray}
v_{1}^{\dag }(0)\Sigma _{y}v_{2}(0) &=&v_{1}^{\dag }(0)\Sigma
_{y}v_{1}^{\ast }(0)  \notag \\
&=&-i\widetilde{v_{1}^{\ast }}(0)Jv_{1}^{\ast }(0)  \notag \\
&=&0.
\end{eqnarray}%
That is to say,
\begin{equation}
v_{1}^{\dag }(0)\Sigma _{y}v_{2}(0)=0,\text{ \ }v_{2}^{\dag }(0)\Sigma
_{y}v_{1}(0)=0.  \label{VV002C}
\end{equation}%
Equations (\ref{VV001C}) and (\ref{VV002C}) contradict the fact that $\Sigma
_{y}$ is a nonsingular metric on $V_{0}$. That is to say, $v_{1}(0)$ can not
be isotropic. It can thus be normalized,%
\begin{equation}
v_{1}^{\dag }(0)\Sigma _{y}v_{1}(0)=1\text{ or }-1.
\end{equation}%
Correspondingly,
\begin{equation}
v_{2}^{\dag }(0)\Sigma _{y}v_{2}(0)=-1\text{ or }1.
\end{equation}%
Those together with Eq. (\ref{VV002C}) show that%
\begin{equation}
v_{i}^{\dag }(0)\Sigma _{y}v_{j}(0)=-\lambda _{i}\delta _{ij},\text{ \ }%
\lambda _{i}=\pm 1,\text{ \ }i,j=1,2.
\end{equation}%
Thereby, the lemma holds when $m=1$.

The rest steps of mathematical induction are similar to the lemma \ref%
{ZeroSpace}.
\end{proof}

The combination of the lemmas (\ref{Lmm5D})--(\ref{Lmm9D}) shows that there
are totally $2n$ dynamic mode pairs. Each mode pair takes the form of $%
\{\omega ,v(\omega )\}$ and $\{-\omega ,v^{\ast }(\omega )\}$, viz., it
contains two opposite eigenenergies, and two complex conjugate eigenvectors.
Most importantly, the two complex conjugate eigenvectors have opposite
norms, one is $+1$, the other is $-1$:
\begin{equation}
v^{\dag }(-\omega _{i})\Sigma _{y}v(-\omega _{i})=-\left[ v^{\dag }(\omega
_{i})\Sigma _{y}v(\omega _{i})\right] ^{\ast }=\pm 1.
\end{equation}%
Thereby, we can stipulate an order for every mode pair as in Eqs. (\ref{TN}%
)--(\ref{TN2}): The first eigenvector has the norm of $+1$, the second one
has the norm of $-1$. So, we obtain a derivative transformation $T_{d}$ as
follows,
\begin{widetext}%
\begin{gather}
\phi =T_{d}\psi ,  \label{TNC} \\
T_{d}=\left[
\begin{array}{cccccccc}
v(\omega _{1}), & v(\omega _{2}), & \cdots , & v(\omega _{n}), & v(-\omega
_{1}), & v(-\omega _{2}), & \cdots , & v(-\omega _{n})%
\end{array}%
\right] ,
\end{gather}%
\end{widetext}
where
\begin{eqnarray}
v^{\dag }(\omega _{i})\Sigma _{y}v(\omega _{i}) &=&1, \\
v^{\dag }(-\omega _{i})\Sigma _{y}v(-\omega _{i}) &=&-1.
\end{eqnarray}

Consequently, we obtain a new version of the lemma \ref{BVTn}.

\begin{lemma}
If the dynamic matrix $D$ is physically diagonalizable, its derivative
transformation $T_{d}$ satisfies the identity,
\begin{equation}
T_{d}^{\dag }\Sigma _{y}T_{d}=I_{-}.
\end{equation}
\end{lemma}

The lemma shows that%
\begin{equation}
\psi\cdot\psi^{\dag}=I_{-}.  \label{PsiPsiI}
\end{equation}

The new field $\psi$ has another important property.

\begin{lemma}
The new field $\psi$ difined in Eq. (\ref{TNC}) has the involution symmetry,
\begin{equation}
\psi=\left( \widetilde{\Sigma_{x}\psi}\right) ^{\dag}.  \label{InvSymm}
\end{equation}
\end{lemma}

\begin{proof}
As $\phi$ is a real field, one has%
\begin{equation}
\phi=\widetilde{\phi^{\dag}}.
\end{equation}
From Eq. (\ref{TNC}), it follows that%
\begin{equation}
\widetilde{\phi^{\dag}}=T_{d}^{\ast}\widetilde{\psi^{\dag}},
\end{equation}
which means%
\begin{equation}
T_{d}\psi=T_{d}^{\ast}\widetilde{\psi^{\dag}}.  \label{TnTnStar}
\end{equation}
It is easy to show that%
\begin{equation}
T_{d}^{\ast}=T_{d}\Sigma_{x}.
\end{equation}
Substituting it into Eq. (\ref{TnTnStar}), one obtains%
\begin{equation}
T_{d}\psi=T_{d}\Sigma_{x}\widetilde{\psi^{\dag}}.
\end{equation}
Since $T_{d}$ is invertible, and $\Sigma_{x}^{-1}=\Sigma_{x}$, he arrives
finally at%
\begin{equation}
\Sigma_{x}\psi=\widetilde{\psi^{\dag}}.
\end{equation}
It is equivalent to Eq. (\ref{InvSymm}).
\end{proof}

This lemma shows that the new field $\psi$ takes exactly the form of Eq. (%
\ref{Psi}). Together with Eq. (\ref{PsiPsiI}), it means that the new field $%
\psi$ is a standard bosonic field.

\begin{lemma}
If the dynamic matrix $D$ is physically diagonalizable, its derivative
transformation $T_{d}$ will diagonalize the coefficient matrix $M$ in the
manner of Hermitian congruence,
\begin{equation}
T_{d}^{\dag }MT_{d}=\mathrm{diag}(\omega _{1},\cdots ,\omega _{n},\omega
_{1},\cdots ,\omega _{n}).  \label{TMTC}
\end{equation}
\end{lemma}

\begin{proof}
From Eq. (\ref{DSyM}), it follows that%
\begin{equation}
M=\Sigma _{y}D.
\end{equation}%
Therefore,%
\begin{eqnarray}
T_{d}^{\dag }MT_{d} &=&T_{d}^{\dag }\Sigma _{y}DT_{d}  \notag \\
&=&T_{d}^{\dag }\Sigma _{y}T_{d}T_{d}^{-1}DT_{d}  \notag \\
&=&I_{-}T_{d}^{-1}DT_{d}.
\end{eqnarray}%
As $D$ is physically diagonalizable, we have%
\begin{equation}
T_{d}^{-1}DT_{d}=\mathrm{diag}(\omega _{1},\cdots ,\omega _{n},-\omega
_{1},\cdots ,-\omega _{n}).
\end{equation}%
The combination of the two equations above proves the lemma.
\end{proof}

\begin{theorem}
A Hamiltonian quadratic in coordinates and momenta is Diracianly
diagonalizable if and only if its dynamic matrix is physically
diagonalizable.
\end{theorem}

\begin{proof}
The necessary condition has been proved by the propostion \ref{PPD}.

The sufficient condition can be proved as follows.

Consider the quadratic Hamiltonian of Eq. (\ref{CHam}). If its dynamic
matrix is physically diagonalizable, it can generates a derivative
transformation as given in Eq. (\ref{TNC}). Under this transformation, Eq. (%
\ref{CHam}) becomes%
\begin{eqnarray}
H &=&\frac{1}{2}\psi ^{\dag }T_{d}^{\dag }MT_{d}\psi  \notag \\
&=&\sum_{i=1}^{n}\omega _{i}d_{i}^{\dag }d_{i}+\frac{1}{2}%
\sum_{i=1}^{n}\omega _{i},  \label{DDH}
\end{eqnarray}%
where Eqs. (\ref{Psi}), (\ref{PsiPsiI}), and (\ref{TMTC}) have been used. It
is a Hamiltonian that is diagonal with respect to the bosonic operators $%
d_{i}^{\dag }$ and $d_{i}$\ ($i=1,2,\cdots ,n$).
\end{proof}

This theorem shows that the derivative transformation defined in Eq. (\ref%
{TNC}) is exactly a Dirac transformation, it brings a Hamiltonian quadratic
in coordinates and momenta into the form diagonalized with respect to
bosons. Just as BV transformation, Dirac transformation can be generated by
the equation of motion of the system automatically.

Now, let us return to the linear harmonic oscillator, the dynamic matrix is%
\begin{equation}
D=\left[
\begin{array}{cc}
0 & -i \\
i & 0%
\end{array}%
\right] \left[
\begin{array}{cc}
\frac{1}{m} & 0 \\
0 & m\omega ^{2}%
\end{array}%
\right] .
\end{equation}%
It has a pair of eigenenergies,%
\begin{equation}
\varepsilon =\pm \omega .
\end{equation}%
If $\omega =0$, there exists only one eigenvector,%
\begin{equation}
v(0)=\left[
\begin{array}{c}
0 \\
1%
\end{array}%
\right] .
\end{equation}%
It means that the Hamiltonian of a free particle,%
\begin{equation}
H=\frac{p^{2}}{2m},
\end{equation}%
is not Diracianly diagonalizable. If $\omega >0$, $D$ is physically
diagonalizable, there are two linearly independent orthonormalized
eigenvectors%
\begin{eqnarray}
v(\omega ) &=&\frac{1}{\sqrt{2m\omega }}\left[
\begin{array}{c}
-im\omega \\
1%
\end{array}%
\right] , \\
v(-\omega ) &=&\frac{1}{\sqrt{2m\omega }}\left[
\begin{array}{c}
im\omega \\
1%
\end{array}%
\right] ,
\end{eqnarray}%
where the convention of Eq. (\ref{VPD1}) has been used for $v(-\omega )$.
They generate a Dirac transformation,%
\begin{equation}
\left[
\begin{array}{c}
p \\
q%
\end{array}%
\right] =\frac{1}{\sqrt{2m\omega }}\left[
\begin{array}{cc}
-im\omega & im\omega \\
1 & 1%
\end{array}%
\right] \left[
\begin{array}{c}
a \\
a^{\dag }%
\end{array}%
\right] .
\end{equation}%
As a result, we have%
\begin{eqnarray}
p &=&-i\sqrt{\frac{m\omega }{2}}(a-a^{\dag }),  \label{pDirac} \\
q &=&\frac{1}{\sqrt{2m\omega }}(a+a^{\dag }).  \label{qDirac}
\end{eqnarray}%
They are just the inverse of Eqs. (\ref{DT1}) and (\ref{DT2}). According to
Eq. (\ref{DDH}), the diagonalized Hamiltonian has the form,%
\begin{equation}
H=\omega a^{\dag }a+\frac{1}{2}\omega ,
\end{equation}%
which is the same as Eq. (\ref{LHH}).

\begin{theorem}
If a Hamiltonian quadratic in coordinates and momenta is Diracianly
diagonalizable, its diagonalized form will be unique up to a permutation of
the quadratic terms.
\end{theorem}

\begin{proof}
The proof is similar to that for the theorem \ref{UniqueB}.
\end{proof}

Set%
\begin{eqnarray}
c_{i} &=&\frac{1}{\sqrt{2}}(q_{i}+ip_{i}), \\
c_{i}^{\dag } &=&\frac{1}{\sqrt{2}}(q_{i}-ip_{i}),
\end{eqnarray}%
where $q_{i}$ and $p_{i}$ ($i=1,2,\cdots ,n$) are the coordinates and
momenta which satisfy the canonical commutation rules of Eq. (\ref{CCR}). It
is easy to show that $c_{i}$ and $c_{i}^{\dag }$ are the annihilation and
creation operators that satisfy the bosonic commutation rules of Eqs. (\ref%
{Com1})--(\ref{Com3}). The inverse is also true. Upon such an invertible
linear substitution, the two Hamiltonians of Eqs. (\ref{Ham1}) and (\ref%
{HamC0}) can be transformed into each other, up to a real constant. Thereby,
the BV diagonalization of Eq. (\ref{Ham1}) is mathematically equivalent to
the Dirac diagonalization of Eq. (\ref{HamC0}). However, Dirac
diagonalization does not require, in physics, the initial field to be a
standard bosonic field. Correspondingly, the transformation matrix does not
need to ensure the invariance of the metric, and is no longer limited to be
a member of the group $U(n,n)$, cf.,%
\begin{eqnarray}
T_{n}^{\dag }I_{-}T_{n} &=&I_{-}, \\
T_{d}^{\dag }\Sigma _{y}T_{d} &=&I_{-}.
\end{eqnarray}%
In this sense, Dirac diagonalization is more general than BV
diagonalization. In particular, it can be generalized easily to complex
collective coordinates and momenta, so it is rather useful in the
quantizations of Bose fields, which will be discussed in the next section.

Here and now, we would give an interesting example of Dirac diagonalization.

\subsection{Landau quantization}

Consider a charged particle moving in a uniform magnetic field. The
Hamiltonian is%
\begin{equation}
H=\frac{1}{2m}(\mathbf{p-}q\mathbf{A)}^{2},
\end{equation}%
where $q$, $m$ and $\mathbf{p}$ denote the charge, mass and momentum of the
particle, respectively. As to $\mathbf{A}$, it is the vector potential of
the magnetic field, which can be expressed in the form \cite{Landau},%
\begin{equation}
\mathbf{A}=\frac{1}{2}\mathbf{B}\times \mathbf{r,}
\end{equation}%
where $\mathbf{B}$ represents the magnetic field, and $\mathbf{r}$ the
coordinates of the particle. Assume without loss of generality that the
magnetic field $\mathbf{B}$ is set along the $z$-axis. Thus%
\begin{eqnarray}
H &=&\frac{1}{2m}\left( p_{x}^{2}+p_{y}^{2}\right) +\frac{1}{2}m\omega
_{L}^{2}\left( x^{2}+y^{2}\right)  \notag \\
&&+\omega _{L}(xp_{y}-yp_{x})+\frac{1}{2m}p_{z}^{2},
\end{eqnarray}%
where $\omega _{L}=qB/2m$ is the Larmor frequency. Obviously, $p_{z}$
commutes with $H$, it will be conserved. This means that $p_{z}$ can be
replaced by a constant, which brings a contantant energy, $p_{z}^{2}/2m$, to
the Hamiltonian. Since a constant energy is unimportant to a Hamiltonian, we
shall concern ourselves with the simplified version,%
\begin{eqnarray}
H &=&\frac{1}{2m}\left( p_{x}^{2}+p_{y}^{2}\right) +\frac{1}{2}m\omega
_{L}^{2}\left( x^{2}+y^{2}\right)  \notag \\
&&+\omega _{L}(xp_{y}-yp_{x}).
\end{eqnarray}%
As shown in the following, this Hamiltonian can be Diracianly diagonalized
and yield the so-called Landau levels.

\begin{solution}
The dynamic matrix is%
\begin{equation}
D=\Sigma _{y}M,
\end{equation}%
where $M$ is the coefficient matrix,%
\begin{equation}
M=\left[
\begin{array}{cccc}
\frac{1}{m} & 0 & 0 & -\omega _{L} \\
0 & \frac{1}{m} & \omega _{L} & 0 \\
0 & \omega _{L} & m\omega _{L}^{2} & 0 \\
-\omega _{L} & 0 & 0 & m\omega _{L}^{2}%
\end{array}%
\right] .
\end{equation}%
The dynamic matrix $D$ has three eigenvalues%
\begin{equation}
\omega _{1}=2\omega _{L},\text{ \ }\omega _{2}=-2\omega _{L},\text{ \ \ }%
\omega _{3}=0.
\end{equation}%
The first two constitute a dynamic mode, their eigenvectors are complex
conjugate to each other,%
\begin{equation}
v(2\omega )=\frac{1}{2\sqrt{m\omega _{L}}}\left[
\begin{array}{c}
-m\omega _{L} \\
-im\omega _{L} \\
-i \\
1%
\end{array}%
\right] ,
\end{equation}%
\begin{equation}
v(-2\omega )=\frac{1}{2\sqrt{m\omega _{L}}}\left[
\begin{array}{c}
-m\omega _{L} \\
im\omega _{L} \\
i \\
1%
\end{array}%
\right] ,
\end{equation}%
the corresponding norms are%
\begin{eqnarray}
v^{\dag }(2\omega _{L})\Sigma _{y}v(2\omega _{L}) &=&1, \\
v^{\dag }(-2\omega _{L})\Sigma _{y}v(-2\omega _{L}) &=&-1.
\end{eqnarray}%
The third is zero, it is two-fold degenerate, and has two linearly
independent eigenvectors, e.g.,%
\begin{equation}
v_{1}(0)=\left[
\begin{array}{c}
m\omega _{L} \\
0 \\
0 \\
1%
\end{array}%
\right] ,\text{ \ }v_{2}(0)=\left[
\begin{array}{c}
0 \\
-m\omega _{L} \\
1 \\
0%
\end{array}%
\right] .
\end{equation}%
They can be linearly combined and orthonormalized into a dynamic mode,%
\begin{equation}
v(+0)=\frac{1}{2\sqrt{m\omega _{L}}}\left[
\begin{array}{c}
m\omega _{L} \\
-im\omega _{L} \\
i \\
1%
\end{array}%
\right] ,
\end{equation}%
\begin{equation}
v(-0)=\frac{1}{2\sqrt{m\omega _{L}}}\left[
\begin{array}{c}
m\omega _{L} \\
im\omega _{L} \\
-i \\
1%
\end{array}%
\right] ,
\end{equation}%
with the norms being%
\begin{eqnarray}
v^{\dag }(+0)\Sigma _{y}v(+0) &=&1, \\
v^{\dag }(-0)\Sigma _{y}v(-0) &=&-1.
\end{eqnarray}%
The Dirac transformation is therefore obtained as follows,%
\begin{equation}
\left[
\begin{array}{c}
p_{x} \\
p_{y} \\
x \\
y%
\end{array}%
\right] =\left[ v(2\omega _{L}),v(+0),v(-2\omega _{L}),v(-0)\right] \left[
\begin{array}{c}
a_{1} \\
a_{2} \\
a_{1}^{\dag } \\
a_{2}^{\dag }%
\end{array}%
\right] .
\end{equation}%
Finally, the diagonalized Hamiltonian is%
\begin{equation}
H=2\omega _{L}a_{1}^{\dag }a_{1}+\omega _{L}.
\end{equation}
\end{solution}

This result is rigorous and exactly the same as that due to Landau \cite%
{Landau}, with $\omega_{c}=2\omega_{L}$ being the cyclotron frequency.

\subsection{Partial diagonalization}

If $\gamma=0$,
\begin{equation}
H=\frac{1}{2}p^{\dag}\mu p+\frac{1}{2}q^{\dag}\kappa q,  \label{Harm5}
\end{equation}
it reduces to the Hamiltonian of Eq. (\ref{Ham3}). As already known from
Sec. \ref{Sec1B}, the latter is identical to the Hamiltonian of Eq. (\ref%
{Harm1}). Thereby, we obtain the proposition.

\begin{proposition}
\label{PPSPD} The Hamiltonian of Eq. (\ref{Harm5}) is Diracianly
diagonalizable when the matrices $\mu$ and $\kappa$ are both positive
definite. If $\mu$ is positive definite but $\kappa$ is nonnegative
definite, it is partially diagonalizable, i.e., except that the normal modes
with zero frequencies are conserved quantities and hence not Diracianly
diagonalizable, all the rest part can be Diracianly diagonalized.
\end{proposition}

The Hamiltonian of such kind is quite important because it is frequently
encountered in various physical problems, particularly, in small
oscillations and field quantization. To handle it, one can first transform
it into Eq. (\ref{Harm1}) and then perform a Dirac diagonalization to the
latter, the total Dirac transformation being the product of the two
successive transformations. That will be easier than handling Eq. (\ref%
{Harm5}) directly.

In the end, it is worth noting that all the conclusions of this section are
valid for the time-polarized commutation relations,%
\begin{equation}
\lbrack p_{i},q_{j}]=i\delta_{ij},\text{ \ }[p_{i},p_{j}]=0,\text{ \ }%
[q_{i},q_{j}]=0,
\end{equation}
which can be seen readily by exchanging the roles of the group of $p_{i}$ ($%
\forall i\in S\subset\{1,2,\cdots,n\}$) and the group of $q_{i}$ ($\forall
i\in S\subset\{1,2,\cdots,n\}$). This kind of abnormal commutation relations
occurs in the quantization of Maxwell field, and will be discussed in Sec. %
\ref{TPP}.

\section{Field Quanta \label{FQ}}

In this section, we intend to examine field quanta, including Klein-Gordon
field, phonon field, and Dirac field. In references, such problems are less
addressed by BV or Dirac transformation. We find that they are pretty good
tools for those problems.

\subsection{Klein-Gordon field \label{KGF}}

Let us begin with the neutral Klein-Gordon field $\phi (\mathbf{x})$ \cite%
{Greiner,Mandl,Berestetskii}. Its Hamiltonian reads as follows,
\begin{equation}
H=\int \mathrm{d}\mathbf{x}\frac{1}{2}\left\{ \pi ^{2}(\mathbf{x})+\left[
\nabla \phi (\mathbf{x})\right] ^{2}+m^{2}\phi ^{2}(\mathbf{x})\right\} ,
\end{equation}%
where $m>0$ is the mass of the field. The $\pi (\mathbf{x})$ is the momentum
density conjugate to the field $\phi (\mathbf{x})$, they satisfy the
canonical commutation rules,%
\begin{eqnarray}
\left[ \pi (\mathbf{x}),\phi (\mathbf{x}^{\prime })\right] &=&-i\delta (%
\mathbf{x-x}^{\prime }), \\
\left[ \pi (\mathbf{x}),\pi (\mathbf{x}^{\prime })\right] &=&0, \\
\left[ \phi (\mathbf{x}),\phi (\mathbf{x}^{\prime })\right] &=&0,
\end{eqnarray}%
where $\delta (\mathbf{x})$ denotes Dirac delta function.

As usual, we would expand $\phi (\mathbf{x})$ and $\pi (\mathbf{x})$ into
plane waves,%
\begin{eqnarray}
\phi (\mathbf{x}) &=&\frac{1}{\left( 2\pi \right) ^{3/2}}\int \mathrm{d}%
\mathbf{p\,}\phi (\mathbf{p})\mathrm{e}^{i\mathbf{p}\cdot \mathbf{x}},
\label{Phix} \\
\pi (\mathbf{x}) &=&\frac{1}{\left( 2\pi \right) ^{3/2}}\int \mathrm{d}%
\mathbf{p\,}\pi (\mathbf{p})\mathrm{e}^{-i\mathbf{p}\cdot \mathbf{x}},
\label{Pix}
\end{eqnarray}%
where we used the duality between $\phi (\mathbf{x})$ and $\pi (\mathbf{x})$%
. Physically, the $\phi (\mathbf{p})$ and $\pi (\mathbf{p})$ represent the
complex collective coordinates and momenta of the system, respectively. As $%
\phi (\mathbf{x})$ and $\pi (\mathbf{x})$ are both real-valued fields,
\begin{eqnarray}
\phi ^{\dag }(\mathbf{p}) &=&\phi (-\mathbf{p}), \\
\pi ^{\dag }(\mathbf{p}) &=&\pi (-\mathbf{p}).
\end{eqnarray}%
In terms of $\phi (\mathbf{p})$ and $\pi (\mathbf{p})$, the Hamiltonian and
canonical commutation rules can be expressed as follows,%
\begin{equation}
H=\int \mathrm{d}\mathbf{p}\frac{1}{2}\left[ \pi ^{\dag }(\mathbf{p})\pi (%
\mathbf{p})+\left( m^{2}+\mathbf{p}^{2}\right) \phi ^{\dag }(\mathbf{p})\phi
(\mathbf{p})\right] ,  \label{HamHarm1}
\end{equation}%
\begin{eqnarray}
\left[ \pi (\mathbf{p}),\phi (\mathbf{p}^{\prime })\right] &=&-i\delta (%
\mathbf{p-p}^{\prime }),  \label{CCom1} \\
\left[ \pi (\mathbf{p}),\pi (\mathbf{p}^{\prime })\right] &=&0,
\label{CCom2} \\
\left[ \phi (\mathbf{p}),\phi (\mathbf{p}^{\prime })\right] &=&0.
\label{CCom3}
\end{eqnarray}

First of all, let us take a look at the equations of motion for $\phi (%
\mathbf{p})$ and $\pi (\mathbf{p})$, which can be derived from Eq. (\ref%
{HamHarm1}) and Eqs. (\ref{CCom1})--(\ref{CCom3}),%
\begin{eqnarray}
i\frac{\mathrm{d}}{\mathrm{d}t}\phi (\mathbf{p}) &=&i\pi ^{\dag }(\mathbf{p}%
), \\
i\frac{\mathrm{d}}{\mathrm{d}t}\pi ^{\dag }(\mathbf{p}) &=&-i\left( m^{2}+%
\mathbf{p}^{2}\right) \phi (\mathbf{p}).
\end{eqnarray}%
They can be combined as%
\begin{equation}
i\frac{\mathrm{d}}{\mathrm{d}t}\varphi (\mathbf{p})=D(\mathbf{p})\varphi (%
\mathbf{p}),  \label{KGHEq}
\end{equation}%
where $\varphi (\mathbf{p})$ is the field operator,%
\begin{equation}
\varphi (\mathbf{p})=\left[
\begin{array}{c}
\pi ^{\dag }(\mathbf{p}) \\
\phi (\mathbf{p})%
\end{array}%
\right] ,
\end{equation}%
and $D(\mathbf{p})$ the the dynamic matrix,%
\begin{equation}
D(\mathbf{p})=\left[
\begin{array}{cc}
0 & -i\left( m^{2}+\mathbf{p}^{2}\right) \\
i & 0%
\end{array}%
\right] .
\end{equation}%
Obviously, the commutation rule for $\varphi (\mathbf{p})$ is%
\begin{equation}
\varphi (\mathbf{p})\cdot \varphi ^{\dag }(\mathbf{p}^{\prime })=\delta (%
\mathbf{p-p}^{\prime })\Sigma _{y},
\end{equation}%
where $\Sigma _{y}$ is defined in Eq. (\ref{Sigmay}).

The form of Eq. (\ref{KGHEq}) hints us that the Hamiltonian $H$\ should be
formulated with the field $\varphi(\mathbf{p})$,%
\begin{equation}
H=\int\mathrm{d}\mathbf{p}\frac{1}{2}\varphi^{\dag}(\mathbf{p})M(\mathbf{p}%
)\varphi(\mathbf{p}),
\end{equation}
where $M(\mathbf{p})$ is the coefficient matrix,%
\begin{equation}
M(\mathbf{p})=\left[
\begin{array}{cc}
1 & 0 \\
0 & m^{2}+\mathbf{p}^{2}%
\end{array}
\right] .
\end{equation}

It is evident that%
\begin{eqnarray}
\varphi (-\mathbf{p}) &=&\widetilde{\left[ \varphi ^{\dag }(\mathbf{p})%
\right] }, \\
D(\mathbf{p}) &=&\Sigma _{y}M(\mathbf{p}).
\end{eqnarray}%
With those relations on hand, the following conclusions become self-evident.

\begin{enumerate}
\item The Hamiltonian of Eq. (\ref{HamHarm1}) is Diracianly diagonalizable
if and only if $D(\mathbf{p})$ is physically diagonalizable.

\item If $D(\mathbf{p})$ is physically diagonalizable, the dynamic mode pair
takes the form of $(\omega,v(\omega))$ and $(\omega,v(-\omega))$ where $%
v(-\omega)=v^{\ast}(\omega)$.

\item If $D(\mathbf{p})$ is physically diagonalizable, there exists an
orthonormal basis with respect to the metric $\Sigma _{y}$. It generates a
Dirac transformation as follows,%
\begin{gather}
\varphi (\mathbf{p})=T_{d}\psi (\mathbf{p}), \\
T_{d}=\left[
\begin{array}{cc}
v(\omega ), & v(-\omega )%
\end{array}%
\right] ,
\end{gather}%
where
\begin{eqnarray}
v^{\dag }(\omega )\Sigma _{y}v(\omega ) &=&1, \\
v^{\dag }(-\omega )\Sigma _{y}v(-\omega ) &=&-1.
\end{eqnarray}%
This implies that%
\begin{equation}
T_{d}^{\dag }\Sigma _{y}T_{d}=I_{-},
\end{equation}%
and that%
\begin{equation}
\psi (\mathbf{p})\cdot \psi ^{\dag }(\mathbf{p}^{\prime })=\delta (\mathbf{%
p-p}^{\prime })I_{-}.  \label{PPDI}
\end{equation}

\item The new field $\psi (\mathbf{p})$ has the involution symmetry,
\begin{equation}
\psi (-\mathbf{p})=\left( \widetilde{\Sigma _{x}\psi (\mathbf{p})}\right)
^{\dag }.
\end{equation}%
This together with Eq. (\ref{PPDI}) means that $\psi (\mathbf{p})$ is a
standard bosonic field, and assumes the following form,%
\begin{equation}
\psi (\mathbf{p})=\left[
\begin{array}{c}
a(\mathbf{p}) \\
a^{\dag }(-\mathbf{p})%
\end{array}%
\right] ,
\end{equation}%
where%
\begin{eqnarray}
\lbrack a(\mathbf{p}),a^{\dag }(\mathbf{p}^{\prime })] &=&\delta (\mathbf{p-p%
}^{\prime }), \\
\lbrack a(\mathbf{p}),a(\mathbf{p}^{\prime })] &=&0, \\
\lbrack a^{\dag }(\mathbf{p}),a^{\dag }(\mathbf{p}^{\prime })] &=&0.
\end{eqnarray}
\end{enumerate}

Those discussions demonstrate that the theory of Dirac diagonalization is
also suitable for the complex collective coordinates and momenta.

It is easy to show that $D(\mathbf{p})$ has a pair of real eigenvalues,%
\begin{equation}
\omega =\pm \varepsilon (\mathbf{p})=\pm \sqrt{m^{2}+\mathbf{p}^{2}}.
\end{equation}%
Since $D(\mathbf{p})$ is a square matrix of size 2, this means that $D(%
\mathbf{p})$ is physically diagonalizable. According to the above statement
1, the neutral Klein-Gordon field is Diracianly diagonalizable.

According to the above statements 2 and 3, the corresponding orthonormal
eigenvectors can be chosen as%
\begin{eqnarray}
v(\varepsilon (\mathbf{p})) &=&\frac{1}{\sqrt{2\varepsilon (\mathbf{p})}}%
\left[
\begin{array}{c}
-i\varepsilon (\mathbf{p}) \\
1%
\end{array}%
\right] , \\
v(-\varepsilon (\mathbf{p})) &=&\frac{1}{\sqrt{2\varepsilon (\mathbf{p})}}%
\left[
\begin{array}{c}
i\varepsilon (\mathbf{p}) \\
1%
\end{array}%
\right] ,
\end{eqnarray}%
with the norms being%
\begin{eqnarray}
v^{\dag }(\varepsilon (\mathbf{p}))\Sigma _{y}v(\varepsilon (\mathbf{p}))
&=&1, \\
v^{\dag }(-\varepsilon (\mathbf{p}))\Sigma _{y}v(-\varepsilon (\mathbf{p}))
&=&-1.
\end{eqnarray}

Finally, according to the above statements 3 and 4, the Dirac transformation
can be constructed as follows,
\begin{equation}
\left[
\begin{array}{c}
\pi ^{\dag }(\mathbf{p}) \\
\phi (\mathbf{p})%
\end{array}%
\right] =\left[
\begin{array}{cc}
v_{1}(\varepsilon (\mathbf{p})), & v_{2}(-\varepsilon (\mathbf{p})%
\end{array}%
\right] \left[
\begin{array}{c}
a(\mathbf{p}) \\
a^{\dag }(-\mathbf{p})%
\end{array}%
\right] .
\end{equation}%
As a result, we have%
\begin{eqnarray}
\phi (\mathbf{p}) &=&\sqrt{\frac{1}{2\varepsilon (\mathbf{p})}}\left[ a(%
\mathbf{p})+a^{\dag }(-\mathbf{p})\right] ,  \label{Phip} \\
\pi (\mathbf{p}) &=&-i\sqrt{\frac{\varepsilon (\mathbf{p})}{2}}\left[ a(-%
\mathbf{p})-a^{\dag }(\mathbf{p})\right] ,  \label{Pip}
\end{eqnarray}%
and%
\begin{equation}
H=\int \mathrm{d}\mathbf{p}\left[ \varepsilon (\mathbf{p})a^{\dag }(\mathbf{p%
})a(\mathbf{p})+\frac{1}{2}\varepsilon (\mathbf{p})\right] .
\end{equation}

Substituting Eqs. (\ref{Phip}) and (\ref{Pip}) into Eqs. (\ref{Phix}) and (%
\ref{Pix}), and complementing them with the variable of time, we have%
\begin{eqnarray}
\phi (\mathbf{x},t) &=&\int \mathrm{d}\mathbf{p}\sqrt{\frac{1}{2\left( 2\pi
\right) ^{3}\varepsilon (\mathbf{p})}}%
\bigg\{%
a(\mathbf{p})\mathrm{e}^{i\left[ \mathbf{p}\cdot \mathbf{x-}\varepsilon (%
\mathbf{p})t\right] }  \notag \\
&&+a^{\dag }(\mathbf{p})\mathrm{e}^{-i\left[ \mathbf{p}\cdot \mathbf{x-}%
\varepsilon (\mathbf{p})t\right] }%
\bigg\}%
, \\
\pi (\mathbf{x},t) &=&-i\int \mathrm{d}\mathbf{p}\sqrt{\frac{\varepsilon (%
\mathbf{p})}{2\left( 2\pi \right) ^{3}}}%
\bigg\{%
a(\mathbf{p})\mathrm{e}^{i\left[ \mathbf{p}\cdot \mathbf{x-}\varepsilon (%
\mathbf{p})t\right] }  \notag \\
&&-a^{\dag }(\mathbf{p})\mathrm{e}^{-i\left[ \mathbf{p}\cdot \mathbf{x-}%
\varepsilon (\mathbf{p})t\right] }%
\bigg\}%
.
\end{eqnarray}%
They are explicitly Lorentz covariant.

Evidently, all those results are the same as Refs. \cite%
{Greiner,Mandl,Berestetskii}. The charged Klein-Gordon field can be handled
similarly.

Besides, we note that other complete sets of orthonormal functions, e.g.,
spherical waves, can be used instead of the plane waves to expand the fields
if necessary.

For the neutral Klein-Gordon field, one can deal with the real collective
coordinates and momenta as in Sec. \ref{GBVT} if he expands the fields $\phi
(\mathbf{x})$ and $\pi (\mathbf{x})$ with the real plane waves, i.e., $%
\{\sin (\mathbf{p}\cdot \mathbf{x})\mathbf{,}\cos (\mathbf{p}\cdot \mathbf{x}%
)\}$. He can transform back to the complex representation at the end of the
calculation. That is rather tedious. In quantum mechanics and quantum field
theory, complex waves and fields are unavoidable. That is the reason why we
generalize the diagonalization theory given in Sec. \ref{GBVT} to the case
of complex collective coordinates and momenta.

\subsection{Phonon field}

For the neutral Klein-Gordon field, Eq. (\ref{HamHarm1}) shows that it
belongs to the case that the matrices $\mu $ and $\kappa $ in Eq. (\ref%
{Harm5}) are both positive definite. Physically, that is because it is
massive, i.e., $m>0$. There are also fields that belong to the other case
where $\mu $ is positive definite but $\kappa $ is nonnegative definite. A
familiar example is the phonon field.

For the sake of brevity, we shall consider a simple three-dimensional
lattice. As usual, we take the harmonic approximation \cite%
{Born,Maradudin,Callaway}, under which the Hamiltonian of the system becomes%
\begin{eqnarray}
H &=&\frac{1}{2m}\sum_{\mathbf{l},\alpha }p^{\alpha }(\mathbf{l})p^{\alpha }(%
\mathbf{l})  \notag \\
&&+\frac{1}{2}\sum_{\mathbf{l},\alpha }\sum_{\mathbf{l}^{\prime },\beta
}\Phi _{\alpha \beta }(\mathbf{l}-\mathbf{l}^{\prime })u^{\alpha }(\mathbf{l}%
)u^{\beta }(\mathbf{l}^{\prime }),
\end{eqnarray}%
where $m$ is the mass of the atoms or ions, $\alpha $ and $\beta $ denotes
the $x$-, or $y$-, or $z$-component, $\mathbf{l}$ is shortened for the
lattice vector. The $\mathbf{u}(\mathbf{l})$ represents the phonon field,
and $\mathbf{p}(\mathbf{l})$ the conjugate momentum field. They satisfy the
canonical commutation rules,%
\begin{eqnarray}
\lbrack p^{\alpha }(\mathbf{l}),u^{\beta }(\mathbf{l}^{\prime })]
&=&-i\delta _{\mathbf{ll}^{\prime }}\delta _{\alpha \beta }, \\
\lbrack p^{\alpha }(\mathbf{l}),p^{\beta }(\mathbf{l}^{\prime })] &=&0, \\
\lbrack u^{\alpha }(\mathbf{l}),u^{\beta }(\mathbf{l}^{\prime })] &=&0.
\end{eqnarray}%
The matrix $\Phi _{\alpha \beta }(\mathbf{l}-\mathbf{l}^{\prime })$ stands
for the interaction between the atoms or ions,%
\begin{equation}
\Phi _{\alpha \beta }(\mathbf{l}-\mathbf{l}^{\prime })=\left( \frac{\partial
^{2}\Phi }{\partial u_{\alpha }(\mathbf{l})\partial u_{\beta }(\mathbf{l}%
^{\prime })}\right) _{0}=\Phi _{\beta \alpha }(\mathbf{l}^{\prime }-\mathbf{l%
}),
\end{equation}%
where $\Phi $ is the elastic potential of the lattice, and the subfix $0$
denotes the equilibrium configuration of the atoms or ions. Since the
equilibrium configuration corresponds to the minimum of the potential, the
matrix $\Phi _{\alpha \beta }(\mathbf{l}-\mathbf{l}^{\prime })$ is
nonnegative definite. Moreover, it fulfils the condition,%
\begin{equation}
\sum_{\mathbf{l}}\Phi _{\alpha \beta }(\mathbf{l}-\mathbf{l}^{\prime })=0,
\label{TranlSym}
\end{equation}%
due to the translational symmetry of the lattice.

Using the collective coordinates $u^{\alpha }(\mathbf{k})$ and momenta $%
p^{\alpha }(\mathbf{k})$,%
\begin{eqnarray}
u^{\alpha }(\mathbf{k}) &=&\frac{1}{\sqrt{N}}\sum_{\mathbf{l}}u^{\alpha }(%
\mathbf{l})e^{-i\mathbf{k}\cdot \mathbf{l}}, \\
p^{\alpha }(\mathbf{k}) &=&\frac{1}{\sqrt{N}}\sum_{\mathbf{l}}p^{\alpha }(%
\mathbf{l})e^{i\mathbf{k}\cdot \mathbf{l}},
\end{eqnarray}%
we obtain%
\begin{eqnarray}
H &=&\frac{1}{2m}\sum_{\mathbf{k},\alpha }\left[ p^{\alpha }(\mathbf{k})%
\right] ^{\dag }p^{\alpha }(\mathbf{k})  \notag \\
&&+\frac{1}{2}\sum_{\mathbf{k},\alpha ,\beta }\Phi _{\alpha \beta }(\mathbf{k%
})\left[ u^{\alpha }(\mathbf{k})\right] ^{\dag }u^{\beta }(\mathbf{k}),
\label{PhotonHarm1}
\end{eqnarray}%
where%
\begin{eqnarray}
\lbrack p^{\alpha }(\mathbf{k}),u^{\beta }(\mathbf{k}^{\prime })]
&=&-i\delta _{\mathbf{kk}^{\prime }}\delta _{\alpha \beta }, \\
\lbrack p^{\alpha }(\mathbf{k}),p^{\beta }(\mathbf{k}^{\prime })] &=&0, \\
\lbrack u^{\alpha }(\mathbf{k}),u^{\beta }(\mathbf{k}^{\prime })] &=&0,
\end{eqnarray}%
and%
\begin{equation}
\Phi _{\alpha \beta }(\mathbf{k})=\sum_{\mathbf{l}}\Phi _{\alpha \beta }(%
\mathbf{l})e^{-i\mathbf{k}\cdot \mathbf{l}}.  \label{TranlSym1}
\end{equation}%
Here and hereafter in this subsection, all the vectors $\mathbf{k}$ belong
to the first Brillouin zone.

Since a simple lattice always has inversion symmetry, $\Phi _{\alpha \beta }(%
\mathbf{k})$ gets the following properties,
\begin{eqnarray}
\Phi _{\alpha \beta }(\mathbf{k}) &=&\Phi _{\alpha \beta }^{\ast }(\mathbf{k}%
), \\
\Phi _{\alpha \beta }(\mathbf{k}) &=&\Phi _{\beta \alpha }(\mathbf{k}), \\
\Phi _{\alpha \beta }(-\mathbf{k}) &=&\Phi _{\alpha \beta }(\mathbf{k}).
\end{eqnarray}%
The first two properties indicate that $\Phi _{\alpha \beta }(\mathbf{k})$
is a real symmetric matrix. Besides, $\Phi _{\alpha \beta }(\mathbf{k})$ is
also a nonnegative matrix, just as $\Phi _{\alpha \beta }(\mathbf{l})$.

The Hamiltonian of Eq. (\ref{PhotonHarm1}) can be handled using the field,%
\begin{equation}
\varphi (\mathbf{k})=\left[
\begin{array}{c}
p^{1}(\mathbf{k}) \\
p^{2}(\mathbf{k}) \\
p^{3}(\mathbf{k}) \\
\left[ u^{1}(\mathbf{k})\right] ^{\dag } \\
\left[ u^{2}(\mathbf{k})\right] ^{\dag } \\
\left[ u^{3}(\mathbf{k})\right] ^{\dag }%
\end{array}%
\right] ,
\end{equation}%
as in the proceeding subsection. However, it will be more convenient to
follow the propositions \ref{PPS0} and \ref{PPSPD}. We therefore perform,
first, a linear transformation that will make $H$ diagonalized with regard
to the new collective coordinates and momenta. As in Sec. \ref{Sec1B}, this
transformation can be produced by the equations of motion of the coordinates
$u_{\alpha }(\mathbf{k})$, exactly speaking, the following eigenvalue
equation,%
\begin{equation}
\omega ^{2}u^{\alpha }(\mathbf{k})=\frac{1}{m}\sum_{\beta }\Phi _{\alpha
\beta }(\mathbf{k})u^{\beta }(\mathbf{k}).
\end{equation}%
As $\Phi _{\alpha \beta }(\mathbf{k})/m$ is a real, symmetric, and
nonnegative matrix, it has three eigenvalues,%
\begin{equation}
\omega _{\sigma }^{2}(\mathbf{k})\geq 0,\text{ \ }\sigma =1,2,3,
\end{equation}%
and a complete set of three orthonormal eigenvectors,%
\begin{eqnarray}
\sum_{\alpha }e_{\sigma }^{\alpha }(\mathbf{k})e_{\sigma ^{\prime }}^{\alpha
}(\mathbf{k}) &=&\delta _{\sigma \sigma ^{\prime }}, \\
\sum_{\sigma }e_{\sigma }^{\alpha }(\mathbf{k})e_{\sigma }^{\beta }(\mathbf{k%
}) &=&\delta _{\alpha \beta }.
\end{eqnarray}%
They are the so-called polarization vectors. One can further adjust these
eigenvectors such that%
\begin{equation}
\mathbf{e}_{\sigma }(\mathbf{k})=\mathbf{e}_{\sigma }(-\mathbf{k}),
\end{equation}%
that is because $\Phi _{\alpha \beta }(-\mathbf{k})=\Phi _{\alpha \beta }(%
\mathbf{k})$. According to Eqs. (\ref{Tf1}) and (\ref{Tf2}), $u^{\alpha }(%
\mathbf{k})$ and $p^{\alpha }(\mathbf{k})$ should be expanded as%
\begin{eqnarray}
u^{\alpha }(\mathbf{k}) &=&\sum_{\sigma }\phi _{\sigma }(\mathbf{k})\frac{1}{%
\sqrt{m}}e_{\sigma }^{\alpha }(\mathbf{k}), \\
p^{\alpha }(\mathbf{k}) &=&\sum_{\sigma }\pi _{\sigma }(\mathbf{k})\sqrt{m}%
e_{\sigma }^{\alpha }(\mathbf{k}),
\end{eqnarray}%
where $\phi _{\sigma }(\mathbf{k})$ and $\pi _{\sigma }(\mathbf{k})$ are the
new collective coordinates and momenta of the system. In terms of these new
collective coordinates and momenta, Eq. (\ref{PhotonHarm1}) can be expressed
as%
\begin{equation}
H=\sum_{\mathbf{k},\sigma }\left[ \frac{1}{2}\pi _{\sigma }^{\dag }(\mathbf{k%
})\pi _{\sigma }(\mathbf{k})+\frac{1}{2}\omega _{\sigma }^{2}(\mathbf{k}%
)\phi _{\sigma }^{\dag }(\mathbf{k})\phi _{\sigma }(\mathbf{k})\right] ,
\label{PhotonHarm2}
\end{equation}%
where%
\begin{eqnarray}
\lbrack \pi _{\sigma }(\mathbf{k}),\phi _{\sigma ^{\prime }}(\mathbf{k}%
^{\prime })] &=&-i\delta _{\mathbf{kk}^{\prime }}\delta _{\sigma \sigma
^{\prime }}, \\
\lbrack \pi _{\sigma }(\mathbf{k}),\pi _{\sigma ^{\prime }}(\mathbf{k}%
^{\prime })] &=&0, \\
\lbrack \phi _{\sigma }(\mathbf{k}),\phi _{\sigma ^{\prime }}(\mathbf{k}%
^{\prime })] &=&0,
\end{eqnarray}%
and%
\begin{equation}
\phi _{\sigma }^{\dag }(\mathbf{k})=\phi _{\sigma }(-\mathbf{k}),\text{ \ }%
\pi _{\sigma }^{\dag }(\mathbf{k})=\pi _{\sigma }(-\mathbf{k}).
\end{equation}

The Hamiltonian of Eq. (\ref{PhotonHarm2}) is the same in form as that of
Eq. (\ref{HamHarm1}), it can thus be treated as the latter. However, the
photon field is massless,%
\begin{equation}
\omega _{\sigma }(\mathbf{k})\rightarrow 0,\text{ \ }\mathbf{k}\rightarrow 0,
\end{equation}%
which can be readily seen from Eqs. (\ref{TranlSym}) and (\ref{TranlSym1}).
Therefore, the Hamiltonian $H$ is only partially diagonalizable, the
components of $\mathbf{k}=0$ can not be diagonalized. For $\mathbf{k}\neq 0$%
, the Dirac transformation is the same as Eqs. (\ref{Phip}) and (\ref{Pip}),%
\begin{eqnarray}
\phi _{\sigma }(\mathbf{k}) &=&\sqrt{\frac{1}{2\omega _{\sigma }(\mathbf{k})}%
}\left[ a_{\sigma }(\mathbf{k})+a_{\sigma }^{\dag }(-\mathbf{k})\right] , \\
\pi _{\sigma }(\mathbf{k}) &=&-i\sqrt{\frac{\omega _{\sigma }(\mathbf{k})}{2}%
}\left[ a_{\sigma }(-\mathbf{k})-a_{\sigma }^{\dag }(\mathbf{k})\right] .
\end{eqnarray}%
Under this transformation, the Hamiltonian becomes
\begin{eqnarray}
H &=&\frac{1}{2}\sum_{\sigma }\pi _{\sigma }^{2}(\mathbf{0})  \notag \\
&&+\sum_{\mathbf{k}\neq 0}\sum_{\sigma }\left[ \omega _{\sigma }(\mathbf{k}%
)a_{\mathbf{k}\sigma }^{\dag }a_{\mathbf{k}\sigma }+\frac{1}{2}\omega
_{\sigma }(\mathbf{k})\right] .
\end{eqnarray}%
That is the partially diagonalized form for the Hamiltonian of the photon
field.

It is easy to show that $\pi_{\sigma}(\mathbf{0})$ ($\sigma=1,2,3$)
represent physically the momenta of the center of mass of the system. The
partial diagonalization is thus not difficult to understand because the
center of mass of the system behaviors as a free particle.

\begin{remark}
\label{Remark0} Usually, the components of $\mathbf{k}=0$ are regarded to be
diagonalizable as those of $\mathbf{k}\neq 0$. That is to say, the Dirac
transformation,%
\begin{eqnarray}
\phi _{\sigma }(\mathbf{k}) &=&\sqrt{\frac{1}{2\omega _{\sigma }(\mathbf{k})}%
}\left[ a_{\sigma }(\mathbf{k})+a_{\sigma }^{\dag }(-\mathbf{k})\right] , \\
\pi _{\sigma }(\mathbf{k}) &=&-i\sqrt{\frac{\omega _{\sigma }(\mathbf{k})}{2}%
}\left[ a_{\sigma }(-\mathbf{k})-a_{\sigma }^{\dag }(\mathbf{k})\right] ,
\end{eqnarray}%
and the Hamiltonian%
\begin{equation}
H=\sum_{\mathbf{k},\sigma }\left[ \omega _{\sigma }(\mathbf{k})a_{\mathbf{k}%
\sigma }^{\dag }a_{\mathbf{k}\sigma }+\frac{1}{2}\omega _{\sigma }(\mathbf{k}%
)\right] .
\end{equation}%
are taken to be appropriate for all $\mathbf{k}$ \cite%
{Born,Maradudin,Callaway}. Strictly speaking, that is not right,
particularly in the discrete case (Notice that $\omega _{\sigma }(\mathbf{k}%
)|_{\mathbf{k}=0}=0$. The transformation becomes meaningless when $\mathbf{k}%
=0$.). Nevertheless, it will cause no problem in the thermodynamic limit.
That is because the state of $\mathbf{k}=0$ has zero measure and contributes
nothing to the integration over $\mathbf{k}$. You can change the values of
an integrand on a null set at your will, that imposes no influence on the
integration. Since one usually needs to take the thermodynamic limit finally
in his calculation, it will be convenient to think that the total
Hamiltonian is diagonalizable in this case. The same thing occurs in the
photon field \cite{Greiner,Mandl,Berestetskii} and in the magnons of
antiferromagnets \cite{Madelung,White}. In the language of mathematics,
those fields can be said to be diagonalizable almost everywhere.
\end{remark}

All the discussions in this subsection are also valid for the
diagonalization of the Maxwell field under Coulomb gauge. The
diagonalization of the Maxwell field under Lorentz gauge will be discussed
in Sec. \ref{TPP}.

\subsection{Dirac field}

Finally, let us consider the Dirac field \cite{Greiner,Mandl,Berestetskii}.
It is a Fermi field, different from the two cases above. The Hamiltonian is%
\begin{equation}
H=\int \mathrm{d}\mathbf{x\,}\psi ^{\dag }(\mathbf{x})\left( -i\mathbf{%
\alpha }\cdot \nabla +\beta m\right) \psi (\mathbf{x}),  \label{HDirac}
\end{equation}%
where $m$ is the mass of the field, and
\begin{equation}
\alpha _{i}=\left[
\begin{array}{cc}
0 & \sigma _{i} \\
\sigma _{i} & 0%
\end{array}%
\right] ,\text{ \ }\beta =\left[
\begin{array}{cc}
I & 0 \\
0 & -I%
\end{array}%
\right]
\end{equation}%
with $\sigma _{i}$ ($i=1,2,3$) being Pauli's $2\times 2$ matrices, and $I$
the $2\times 2$ unit matrix. The spinor fields, $\psi (\mathbf{x})$ and $%
\psi ^{\dag }(\mathbf{x})$, satisfy the anticommuation rules,%
\begin{eqnarray}
\lbrack \psi _{\mu }(\mathbf{x}),\psi _{\nu }^{\dag }(\mathbf{x}^{\prime })]
&=&\delta (\mathbf{x-x}^{\prime })\delta _{\mu \nu }, \\
\lbrack \psi _{\mu }(\mathbf{x}),\psi _{\nu }(\mathbf{x}^{\prime })] &=&0, \\
\lbrack \psi _{\mu }^{\dag }(\mathbf{x}),\psi _{\nu }^{\dag }(\mathbf{x}%
^{\prime })] &=&0.
\end{eqnarray}

As before, we can expand the fields $\psi (\mathbf{x})$ and $\psi ^{\dag }(%
\mathbf{x})$ with plane waves,%
\begin{eqnarray}
\psi (\mathbf{x}) &=&\frac{1}{\left( 2\pi \right) ^{3/2}}\int \mathrm{d}%
\mathbf{p\,}\psi (\mathbf{p})\mathrm{e}^{i\mathbf{p}\cdot \mathbf{x}}, \\
\psi ^{\dag }(\mathbf{x}) &=&\frac{1}{\left( 2\pi \right) ^{3/2}}\int
\mathrm{d}\mathbf{p\,}\psi ^{\dag }(\mathbf{p})\mathrm{e}^{-i\mathbf{p}\cdot
\mathbf{x}},
\end{eqnarray}%
where%
\begin{eqnarray}
\psi (\mathbf{p}) &=&\left[
\begin{array}{c}
c_{1}(\mathbf{p}) \\
c_{2}(\mathbf{p}) \\
c_{3}(\mathbf{p}) \\
c_{4}(\mathbf{p})%
\end{array}%
\right] , \\
\psi ^{\dag }(\mathbf{p}) &=&\left[
\begin{array}{cccc}
c_{1}^{\dag }(\mathbf{p}), & c_{2}^{\dag }(\mathbf{p}), & c_{3}^{\dag }(%
\mathbf{p}), & c_{4}^{\dag }(\mathbf{p})%
\end{array}%
\right] .
\end{eqnarray}%
Substituting them into Eq. (\ref{HDirac}), we obtain%
\begin{eqnarray}
H &=&\int \mathrm{d}\mathbf{p\,}\psi ^{\dag }(\mathbf{p})\left( \mathbf{p}%
\cdot \mathbf{\alpha }+m\beta \right) \psi (\mathbf{p})  \notag \\
&=&\int \mathrm{d}\mathbf{p}\sum_{\mu ,\nu }c_{\mu }^{\dag }(\mathbf{p}%
)\left( \mathbf{p}\cdot \mathbf{\alpha }+m\beta \right) _{\mu \nu }c_{\nu }(%
\mathbf{p}),
\end{eqnarray}%
where%
\begin{eqnarray}
\lbrack c_{\mu }(\mathbf{p}),c_{\nu }^{\dag }(\mathbf{p}^{\prime })]
&=&\delta (\mathbf{p-p}^{\prime })\delta _{\mu \nu }, \\
\lbrack c_{\mu }(\mathbf{p}),c_{\nu }(\mathbf{p}^{\prime })] &=&0, \\
\lbrack c_{\mu }^{\dag }(\mathbf{p}),c_{\nu }^{\dag }(\mathbf{p}^{\prime })]
&=&0.
\end{eqnarray}

Since $\alpha _{i}^{\dag }=\alpha _{i}$ and $\beta ^{\dag }=\beta $, $H$ is
a normal Hamiltonian that has been discussed in Sec. \ref{NHF}. According to
the proposition \ref{PPSNHF}, it can be diagonalized by the unitary
transformation,%
\begin{equation}
\psi (\mathbf{p})=T_{\mathbf{p}}\varphi (\mathbf{p}),
\end{equation}%
where $T_{\mathbf{p}}$ is the unitary matrix,%
\begin{equation}
T_{\mathbf{p}}^{\dag }T_{\mathbf{p}}=T_{\mathbf{p}}T_{\mathbf{p}}^{\dag }=I,
\end{equation}%
\begin{equation}
T_{\mathbf{p}}^{\dag }\left( \mathbf{p}\cdot \mathbf{\alpha }+m\beta \right)
T_{\mathbf{p}}=\mathrm{diag}(\varepsilon (\mathbf{p}),\varepsilon (\mathbf{p}%
),-\varepsilon (\mathbf{p}),-\varepsilon (\mathbf{p})),
\end{equation}%
and $\varphi (\mathbf{p})$ the new field,%
\begin{equation}
\varphi (\mathbf{p})=\left[
\begin{array}{c}
d_{1}(\mathbf{p}) \\
d_{2}(\mathbf{p}) \\
d_{3}^{\dag }(-\mathbf{p}) \\
d_{4}^{\dag }(-\mathbf{p})%
\end{array}%
\right] ,
\end{equation}%
Here, a particle-hole transformation has been performed for the negative
energies.

After the transformation, $H$ becomes%
\begin{equation}
H=\int \mathrm{d}\mathbf{p}\sum_{\mu =1}^{4}\varepsilon (\mathbf{p})d_{\mu
}^{\dag }(\mathbf{p})d_{\mu }(\mathbf{p}),
\end{equation}%
where%
\begin{eqnarray}
\lbrack d_{\mu }(\mathbf{p}),d_{\nu }^{\dag }(\mathbf{p}^{\prime })]
&=&\delta (\mathbf{p-p}^{\prime })\delta _{\mu \nu }, \\
\lbrack d_{\mu }(\mathbf{p}),d_{\nu }(\mathbf{p}^{\prime })] &=&0, \\
\lbrack d_{\mu }^{\dag }(\mathbf{p}),d_{\nu }^{\dag }(\mathbf{p}^{\prime })]
&=&0.
\end{eqnarray}%
As usual, the vacuum energy has been removed from the Hamiltonian \cite%
{Greiner,Mandl,Berestetskii}.

In this section, the theory of transformation has been applied to examine
totally three kinds of field quanta. As has been seen, it operates neatly
and concisely.

\section{Mathematical Essence of Diagonalizability \label{MEBVD}}

Up to now, we regard BV or Dirac diagonalization as a physical consequence
of the Heisenberg equation of motion. Diagonalization represents the normal
modes of motion of a system. Such a view provides us a concrete picture and
intuitive interpretation of diagonalization, and thus makes it easy to
understand. However, there is yet another view, it is abstract but more
fundamental. In this view, the diagonalization in itself is an intrinsic and
invariant property of a Hermitian quadratic form that is equipped with
commutator or Poisson bracket, neither Hamiltonian nor equation of motion
will be needed any more.

\subsection{Heisenberg operator}

For simplicity, we shall take the Hamiltonian of Eq. (\ref{Ham1}) as an
instance to explain this abstract view.

First, we regard Eq. (\ref{Ham1}) purely as a Hermitian form $Q$ that is
quadratic in creation and annihilation operators,%
\begin{equation}
Q\triangleq \sum_{i,j=1}^{n}(\alpha _{ij}c_{i}^{\dag }c_{j}+\frac{1}{2}%
\gamma _{ij}c_{i}^{\dag }c_{j}^{\dag }+\frac{1}{2}\gamma _{ji}^{\ast
}c_{i}c_{j}).
\end{equation}%
And then we introduce a linear operator $f$ to replace the Heisenberg
equation of motion,%
\begin{equation}
f(x)\triangleq i[Q,x],
\end{equation}%
where%
\begin{equation}
x=\sum_{i=1}^{n}(z_{i}c_{i}+s_{i}c_{i}^{\dag }),\text{ \ }z_{i}\in
\mathbb{C}
,\text{ and }s_{i}\in
\mathbb{C}
.
\end{equation}%
For convenience, we shall call $f$ the Heisenberg operator. Obviously, we
can define such a Heisenberg operator for each Hermitian quadratic form.

If one denotes the $\mu$th component of the field $\psi$ defined in Eq. (\ref%
{Psi}) with $\psi_{\mu}$, he obtains%
\begin{equation}
if(\psi_{\mu})=D_{\mu\nu}\psi_{\nu},
\end{equation}
where $D_{\mu\nu}$ are exactly the entries of the dynamic matrix $D$. In
matrix notation, it can be written as%
\begin{equation}
if(\psi)=D\psi,
\end{equation}
which is the counterpart of Eq. (\ref{Heqb3}).

By replacing $H$ with $Q$ and the Heisenberg equation with the Heisenberg
operator $f$, one can readily show that all the lemmas, propositions, and
theorems in Sec. \ref{DTBS} will still hold. Evidently, it is also true for
the other cases, e.g., the Fermi system. Here and now, there is no
Hamiltonian, no time, and no equation of motion. This demonstrates
unambiguously that the BV or Dirac diagonalization is essentially the
algebraic property of a quadratic Hermitian form. From this standpoint, any
physical picture and interpretation are redundant and unnecessary, they have
nothing to do with the self of the BV or Dirac diagonalization and can be
completely removed away.

It follows immediately from this abstract view that all the conclusions of
Sec. \ref{GBVT} hold for classical systems, up to a real constant.

Such sublation enlarges the scope of the objects of diagonalization, it can
be performed to all the Hermitian quadratic forms besides the physical
Hamiltonians. To show this point of view more straightforwardly, let us look
at the following example.

\begin{example}
\begin{equation}
J_{z}=xp_{y}-yp_{x}.
\end{equation}
\end{example}

As well known, $J_{z}$ is the orbital angular moment along the $z$%
-direction. It is quadratic in coordinates, i.e., $x$ and $y$, and momenta,
i.e., $p_{x}$ and $p_{y}$. Since $J_{z}$ is not a Hamiltonian, the
Heisenberg equation of motion becomes meaningless. Nevertheless, it can be
Diracianly diagonalized.

\begin{solution}
The dynamic matrix is%
\begin{equation}
D=\Sigma _{y}M,
\end{equation}%
where $M$ is the coefficient matrix,%
\begin{equation}
M=\left[
\begin{array}{cccc}
0 & 0 & 0 & -1 \\
0 & 0 & 1 & 0 \\
0 & 1 & 0 & 0 \\
-1 & 0 & 0 & 0%
\end{array}%
\right] .
\end{equation}%
There exists a pair of eigenvalues for $D$,%
\begin{equation}
\omega =\pm 1.
\end{equation}%
It is easy to show that each eigenvalue has two linearly independent
eigenvectors, for example,%
\begin{equation}
v_{1}(1)=\left[
\begin{array}{c}
1 \\
i \\
0 \\
0%
\end{array}%
\right] ,\text{ \ }v_{2}(1)=\left[
\begin{array}{c}
0 \\
0 \\
1 \\
i%
\end{array}%
\right] .
\end{equation}%
They can be linearly combined and orthonormalized as follows,%
\begin{equation}
v_{1}(1)=\frac{1}{2}\left[
\begin{array}{c}
1 \\
i \\
i \\
-1%
\end{array}%
\right] ,\text{ \ }v_{2}(1)=\frac{1}{2}\left[
\begin{array}{c}
1 \\
i \\
-i \\
1%
\end{array}%
\right] ,
\end{equation}%
with the norms being%
\begin{equation}
v_{1}^{\dag }(1)\Sigma _{y}v_{1}(1)=1,\text{ \ }v_{2}^{\dag }(1)\Sigma
_{y}v_{2}(1)=-1.
\end{equation}%
Correspondingly, the Dirac transformation will be%
\begin{equation}
\left[
\begin{array}{c}
p_{x} \\
p_{y} \\
x \\
y%
\end{array}%
\right] =\left[ v_{1}(1),v_{2}(-1),v_{1}(-1),v_{2}(1)\right] \left[
\begin{array}{c}
a_{1} \\
a_{2} \\
a_{1}^{\dag } \\
a_{2}^{\dag }%
\end{array}%
\right] ,
\end{equation}%
where
\begin{equation}
v_{1}(-1)=v_{1}^{\ast }(1),\text{ \ }v_{2}(-1)=v_{2}^{\ast }(1).
\end{equation}%
Finally, the diagonalized form of $J_{z}$ is%
\begin{equation}
J_{z}=a_{1}^{\dag }a_{1}-a_{2}^{\dag }a_{2}.
\end{equation}
\end{solution}

This result is quite similar to the coupled boson representation for angular
momentum due to Schwinger \cite{Schwinger,Mattis}. It is rigorous, concise,
and agrees exactly with the familiar result about the orbital angular
momentum, i.e., the eigenvalues of $J_{z}$ consists of all the integers.

There is another interesting example.

\begin{example}
\begin{equation}
H=-\frac{p^{2}}{2m}-\frac{1}{2}m\omega^{2}q^{2},
\end{equation}
where $m>0$ and $\omega>0$.
\end{example}

Formally, it looks like a \textquotedblleft negative harmonic
oscillator\textquotedblright. Of course, it can not represent a real
physical system. Thereby we treat it just as a Hermitian quadratic form. It
is ready to conjecture that this quadratic form would be diagonalized, by
the same Dirac transformation as for the normal harmonic oscillator.

\begin{solution}
The dynamic matrix is%
\begin{equation}
D=\Sigma _{y}M,
\end{equation}%
where $M$ is the coefficient matrix,%
\begin{equation}
M=\left[
\begin{array}{cc}
-\frac{1}{m} & 0 \\
0 & -m\omega ^{2}%
\end{array}%
\right] .
\end{equation}%
The dynamic matrix $D$ has a pair of eigenvalues,%
\begin{equation}
\varepsilon =\pm \omega .
\end{equation}%
and two orthonormal eigenvectors,%
\begin{eqnarray}
v(\omega ) &=&\frac{1}{\sqrt{2m\omega }}\left[
\begin{array}{c}
im\omega \\
1%
\end{array}%
\right] , \\
v(-\omega ) &=&\frac{1}{\sqrt{2m\omega }}\left[
\begin{array}{c}
-im\omega \\
1%
\end{array}%
\right] ,
\end{eqnarray}%
\begin{eqnarray}
v^{\dag }(\omega )\Sigma _{y}v(\omega ) &=&-1, \\
v^{\dag }(-\omega )\Sigma _{y}v(-\omega ) &=&1.
\end{eqnarray}%
They generate a Dirac transformation,%
\begin{eqnarray}
p &=&-i\sqrt{\frac{m\omega }{2}}(a-a^{\dag }), \\
q &=&\frac{1}{\sqrt{2m\omega }}(a+a^{\dag }).
\end{eqnarray}%
They are the same as Eqs. (\ref{pDirac}) and (\ref{qDirac}). The
diagonalized form is%
\begin{equation}
H=-\omega a^{\dag }a-\frac{1}{2}\omega .
\end{equation}%
Those results confirm our conjecture completely.
\end{solution}

This example reminds us that the two propositions \ref{PPS0} and \ref{PPSPD}
also hold for the case where $\mu $ is negative definite and $\kappa $ is
nonpositive definite, no matter whether the commutation rules are standard,
or time-polarized, or mixing. This fact will be useful in the following
subsection.

\subsection{Time-polarized photons \label{TPP}}

In this subsection, we shall concern ourselves with the diagonalization of
the Maxwell field under Lorentz gauge.

As well known, the Maxwell field contains two spurious degrees of freedom in
Lorentz gauge, viz., the longitudinal and time-polarized components. Because
these degrees of freedom are not of physical quantity, the Hamiltonian is
merely a Hermitian quadratic form. The abstract view of diagonalization
finds its way and role here.

In Lorentz gauge, the Hamiltonian of Maxwell field has the form \cite%
{Greiner,Mandl,Berestetskii},%
\begin{equation}
H=-\frac{1}{2}\int \mathrm{d}\mathbf{x}\left[ \pi ^{\mu }(\mathbf{x})\pi
_{\mu }\mathbf{(x)+\nabla }A^{\mu }(\mathbf{x})\cdot \mathbf{\nabla }A_{\mu
}(\mathbf{x})\right] ,
\end{equation}%
where $A^{\mu }(\mathbf{x})$ ($\mu =0,1,2,3$) are the vector potentials, and
$\pi ^{\mu }(\mathbf{x})$ the corresponding canonical conjugate fields. They
satisfy the following commutation rules,%
\begin{eqnarray}
\lbrack A^{\mu }(\mathbf{x}),\pi ^{\nu }(\mathbf{x}^{\prime })] &=&ig^{\mu
\nu }\delta (\mathbf{x-x}^{\prime }),  \label{GF1} \\
\lbrack A^{\mu }(\mathbf{x}),A^{\nu }(\mathbf{x}^{\prime })] &=&0,
\label{GF2} \\
\lbrack \pi ^{\mu }(\mathbf{x}),\pi ^{\nu }(\mathbf{x}^{\prime })] &=&0,
\label{GF3}
\end{eqnarray}%
where $g^{\mu \nu }$ is the metric tensor,%
\begin{equation}
g^{\mu \nu }=\left[
\begin{array}{cc}
\begin{array}{cc}
1 &  \\
& -1%
\end{array}
& \text{{\LARGE 0}} \\
\text{{\LARGE 0}} &
\begin{array}{cc}
-1 &  \\
& -1%
\end{array}%
\end{array}%
\right] .
\end{equation}%
Equations (\ref{GF1})--(\ref{GF3}) indicate that the Hamiltonian $H$ belongs
to the mixing case where there are both the standard and time-polarized
commutation relations simultaneously. Physically, that roots from the
requirement of Lorentz covariance.

Expanding these fields with plane waves,%
\begin{eqnarray}
A^{\mu }(\mathbf{x}) &=&\frac{1}{\left( 2\pi \right) ^{3/2}}\int \mathrm{d}%
\mathbf{p\,}A^{\mu }(\mathbf{p})\mathrm{e}^{i\mathbf{p}\cdot \mathbf{x}}, \\
\pi ^{\mu }(\mathbf{x}) &=&\frac{1}{\left( 2\pi \right) ^{3/2}}\int \mathrm{d%
}\mathbf{p\,}\pi ^{\mu }(\mathbf{p})\mathrm{e}^{-i\mathbf{p}\cdot \mathbf{x}%
},
\end{eqnarray}%
we obtain%
\begin{equation}
H=-\frac{1}{2}\int \mathrm{d}\mathbf{x}\left\{ \left[ \pi ^{\mu }(\mathbf{p})%
\right] ^{\dag }\pi _{\mu }\mathbf{(p)+p}^{2}\left[ A^{\mu }(\mathbf{p})%
\right] ^{\dag }A_{\mu }(\mathbf{p})\right\} ,
\end{equation}%
where%
\begin{eqnarray}
\lbrack A^{\mu }(\mathbf{p}),\pi ^{\nu }(\mathbf{p}^{\prime })] &=&ig^{\mu
\nu }\delta (\mathbf{p-p}^{\prime }), \\
\lbrack A^{\mu }(\mathbf{p}),A^{\nu }(\mathbf{p}^{\prime })] &=&0, \\
\lbrack \pi ^{\mu }(\mathbf{p}),\pi ^{\nu }(\mathbf{p}^{\prime })] &=&0,
\end{eqnarray}%
and%
\begin{eqnarray}
\left[ A^{\mu }(\mathbf{p})\right] ^{\dag } &=&A^{\mu }(-\mathbf{p}), \\
\left[ \pi ^{\mu }(\mathbf{p})\right] ^{\dag } &=&\pi ^{\mu }(-\mathbf{p}).
\end{eqnarray}

Just like the phonon field, the polarization vectors can be obtained from
the eigenvalue equation for the fields $A^{\mu }(\mathbf{p})$,%
\begin{equation}
\omega ^{2}\phi (\mathbf{p})=D(\mathbf{p})\phi (\mathbf{p}),
\end{equation}%
where%
\begin{equation}
\phi (\mathbf{p})=\left[
\begin{array}{c}
A^{0}(\mathbf{p}) \\
A^{1}(\mathbf{p}) \\
A^{2}(\mathbf{p}) \\
A^{3}(\mathbf{p})%
\end{array}%
\right] ,
\end{equation}%
and%
\begin{equation}
D(\mathbf{p})=\left[
\begin{array}{cc}
\begin{array}{cc}
\mathbf{p}^{2} &  \\
& \mathbf{p}^{2}%
\end{array}
& \text{{\LARGE 0}} \\
\text{{\LARGE 0}} &
\begin{array}{cc}
\mathbf{p}^{2} &  \\
& \mathbf{p}^{2}%
\end{array}%
\end{array}%
\right] .
\end{equation}%
Paying attention to the fact that $D(\mathbf{p})=\mathbf{p}^{2}I$, the
orthonormal basis for the polarization vectors can be chosen, except $%
\mathbf{p}=0$, as follows,%
\begin{eqnarray}
\epsilon (\mathbf{p},0) &=&(1,0,0,0), \\
\epsilon (\mathbf{p},i) &=&(0,\mathbf{e}_{i}(\mathbf{p})),
\end{eqnarray}%
where%
\begin{gather}
\mathbf{e}(\mathbf{p},3)=\frac{\mathbf{p}}{\left\vert \mathbf{p}\right\vert }%
, \\
\mathbf{e}_{i}(\mathbf{p})\cdot \mathbf{e}_{j}(\mathbf{p})=\delta _{ij},%
\text{ \ }i=1,2,3.
\end{gather}%
That is to say, $\epsilon (\mathbf{p},0)$ is the time-like polarization
vector, $\epsilon (\mathbf{p},i)$ ($i=1,2$) the space-like transverse
polarization vectors, and $\epsilon (\mathbf{p},3)$ the space-like
longitudinal polarization vector. Obviously,%
\begin{equation}
\epsilon ^{\mu }(\mathbf{p},\lambda )\epsilon _{\mu }(\mathbf{p},\lambda
^{\prime })=g_{\lambda \lambda ^{\prime }},
\end{equation}%
i.e., the polarization vectors form a four-dimensional orthonormal system.
Using this new basis, the Hamiltonian can be expressed as%
\begin{eqnarray}
H &=&\frac{1}{2}\int \mathrm{d}\mathbf{p}%
\Bigg\{%
\sum_{\lambda =1}^{3}\left( \left[ \pi _{\lambda }(\mathbf{p})\right] ^{\dag
}\pi _{\lambda }(\mathbf{p})\mathbf{+p}^{2}\left[ A_{\lambda }(\mathbf{p})%
\right] ^{\dag }A_{\lambda }(\mathbf{p})\right)  \notag \\
&&-\left( \left[ \pi _{0}(\mathbf{p})\right] ^{\dag }\pi _{0}(\mathbf{p})%
\mathbf{+p}^{2}\left[ A_{0}(\mathbf{p})\right] ^{\dag }A_{0}(\mathbf{p}%
)\right)
\Bigg\}%
,
\end{eqnarray}%
where%
\begin{eqnarray}
\lbrack A_{\lambda }(\mathbf{p}),\pi _{\lambda ^{\prime }}(\mathbf{p}%
^{\prime })] &=&ig_{\lambda \lambda ^{\prime }}\delta (\mathbf{p-p}^{\prime
}), \\
\lbrack A_{\lambda }(\mathbf{p}),A_{\lambda ^{\prime }}(\mathbf{p}^{\prime
})] &=&0, \\
\lbrack \pi _{\lambda }(\mathbf{p}),\pi _{\lambda ^{\prime }}(\mathbf{p}%
^{\prime })] &=&0,
\end{eqnarray}%
and%
\begin{eqnarray}
\left[ A_{\lambda }(\mathbf{p})\right] ^{\dag } &=&A_{\lambda }(-\mathbf{p}),
\\
\left[ \pi _{\lambda }(\mathbf{p})\right] ^{\dag } &=&\pi _{\lambda }(-%
\mathbf{p}).
\end{eqnarray}%
It is worth noting that the time-like components, $A_{0}(\mathbf{p})$ and $%
\pi _{0}(\mathbf{p})$, consititute a \textquotedblleft negative harmonic
oscillator\textquotedblright .

According to Sec. \ref{KGF}, $H$ can be diagonalized by the following Dirac
transformation,%
\begin{eqnarray}
A_{\lambda }(\mathbf{p}) &=&\sqrt{\frac{1}{2\varepsilon (\mathbf{p})}}\left[
a_{\lambda }(\mathbf{p})+a_{\lambda }^{\dag }(-\mathbf{p})\right] , \\
\pi _{\lambda }(\mathbf{p}) &=&i\sqrt{\frac{\varepsilon (\mathbf{p})}{2}}%
\left[ a_{\lambda }(-\mathbf{p})-a_{\lambda }^{\dag }(\mathbf{p})\right] ,
\end{eqnarray}%
where%
\begin{equation}
\varepsilon (\mathbf{p})=\left\vert \mathbf{p}\right\vert ,
\end{equation}%
and%
\begin{eqnarray}
\lbrack a_{\lambda }(\mathbf{p}),a_{\lambda ^{\prime }}^{\dag }(\mathbf{p}%
^{\prime })] &=&-g_{\lambda \lambda ^{\prime }}\delta (\mathbf{p-p}^{\prime
}), \\
\lbrack a_{\lambda }(\mathbf{p}),a_{\lambda ^{\prime }}(\mathbf{p}^{\prime
})] &=&0, \\
\lbrack a_{\lambda }^{\dag }(\mathbf{p}),a_{\lambda ^{\prime }}^{\dag }(%
\mathbf{p}^{\prime })] &=&0.
\end{eqnarray}%
Here, as usual, the commutation rules for the transverse and longitudinal
polarizations ($\lambda =1,2,3$) are chosen as normal or standard; those for
the time-like polarizations ($\lambda =0$) are chosen as abnormal or
time-polarized. The abnormal bosons here are usually called scalar or
time-polarized photons. That is also the reason why we call the particles
satisfying Eq. (\ref{AbCom}) the time-polarized bosons. The diagonalized
form of the Hamiltonian is%
\begin{equation}
H=\int \mathrm{d}\mathbf{p\,}\varepsilon (\mathbf{p})\left[ \sum_{\lambda
=1}^{3}a_{\lambda }^{\dag }(\mathbf{p})a_{\lambda }(\mathbf{p})-a_{0}^{\dag
}(\mathbf{p})a_{0}(\mathbf{p})\right] ,
\end{equation}%
where the vacuum energy has been removed from the Hamiltonian. This equation
holds in the sense that the Hamiltonian is Diracianly diagonalizable almost
everywhere in momentum space (except $\mathbf{p}=0$).

Finally, the quantized fields can be written as
\begin{eqnarray}
A^{\mu }(x) &=&\int \mathrm{d}\mathbf{p}\sqrt{\frac{1}{2\left( 2\pi \right)
^{3}\varepsilon (\mathbf{p})}}\sum_{\lambda =0}^{3}\epsilon ^{\mu }(\mathbf{p%
},\lambda )  \notag \\
&&\times \left[ a_{\lambda }(\mathbf{p})\mathrm{e}^{ip\cdot x}+a_{\lambda
}^{\dag }(\mathbf{p})\mathrm{e}^{-ip\cdot x}\right] , \\
\pi ^{\mu }(x) &=&i\int \mathrm{d}\mathbf{p}\sqrt{\frac{\varepsilon (\mathbf{%
p})}{2\left( 2\pi \right) ^{3}}}\sum_{\lambda =0}^{3}\epsilon ^{\mu }(%
\mathbf{p},\lambda )  \notag \\
&&\times \left[ a_{\lambda }(\mathbf{p})\mathrm{e}^{ip\cdot x}-a_{\lambda
}^{\dag }(\mathbf{p})\mathrm{e}^{-ip\cdot x}\right] .
\end{eqnarray}%
It can be readily seen that all those results are the same as Refs. \cite%
{Greiner,Mandl,Berestetskii}.

The diagonalization of the Maxwell field under Lorentz gauge is rather
complicated. First, it contains unphysical degrees of freedom. Second, it
needs mixing commutation relations, both for initial and final fields. And
third, it is not diagonalizable everywhere but almost everywhere. One sees
that those problems can be resolved naturally by the diagonalization theory
developed in this review.

\section{Conclusions}

In this review, a theory of transformation is set up for the diagonalization
of the Hermitian quadratic form that is equipped with commutator or Poisson
bracket.

The theory is dynamic matrix oriented.

The dynamic matrix can be derived from the Hermitian quadratic form through
the Heisenberg operator. Each Hermitian quadratic form has a dynamic matrix
of its own.

The Bogoliubov-Valatinian or Diracian diagonalizability of a Hermitian
quadratic form is equivalent to the physical diagonalizability of its
dynamic matrix. That is to say, the diagonalization of a Hermitian quadratic
form is essentially an eigenvalue problem of its dynamic matrix.

The dynamic matrix is always physically diagonalizable for a fermionic form.
It may or may not be physically diagonalizable for a bosonic form.
Accordingly, the diagonalization exists and is unique for a fermionic form,
forever. It exists and is unique for a bosonic form only if the dynamic
matrix is physically diagonalizable.

The dynamic matrix is the generator of the Bogoliubov-Valatin or Dirac
transformation. The transformation required for diagonalization can be
constructed immediately from the complete set of the orthonormal
eigenvectors of the dynamic matrix, according to a standard algebraic
procedure.

In a word, the eigenvalue problem of the dynamic matrix determines the
diagonalizability of a Hermitian quadratic form, definitely and completely.

Finally, it is worth emphasizing that the quadratic Hamiltonian is just
regarded as a Hermitian quadratic form in this review, i.e., only its
mathematical properties are considered here. The physical instability as
well as phase transitions of a system has not been concerned at all. That in
itself is another intriguing problem, please refer to Ref. \cite{Xiao}.

\begin{acknowledgments}
We are deeply grateful to Professor Guo-Xing Ju for his helpful discussions.
\end{acknowledgments}

\bibliography{bvt}
\bibliographystyle{apsrev}

\end{document}